\documentclass[onecolumn,journal]{IEEEtran}
\usepackage[utf8]{inputenc}
\usepackage{cite}
\usepackage{amsmath,amssymb,amsfonts,bm,amsthm}
\usepackage{graphicx}
\usepackage{xcolor}
\usepackage{multirow}
\usepackage{booktabs}
\usepackage{arydshln}
\usepackage{stfloats}
\usepackage[switch,pagewise]{lineno}
\usepackage{url}
\usepackage{tikz}
\usepackage[colorlinks=true,linkcolor=blue,citecolor=blue,urlcolor=blue]{hyperref}
\definecolor{lime}{HTML}{A6CE39}
\DeclareRobustCommand{\orcidicon}{%
    \begin{tikzpicture}
    \draw[lime, fill=lime] (0,0) circle [radius=0.16]
    node[white] {{\fontfamily{qag}\selectfont \tiny ID}};
    \draw[white, fill=white] (-0.0625,0.095) circle [radius=0.007];
    \end{tikzpicture}\hspace{-2mm}}
\foreach \x in {A, ..., Z}{%
    \expandafter\xdef\csname orcid\x\endcsname{\noexpand\href{https://orcid.org/\csname orcidauthor\x\endcsname}{\noexpand\orcidicon}}
}

\newtheorem{lemma}{Lemma}[subsection]
\newtheorem{theorem}{Theorem}[subsection]
\newtheorem{corollary}{Corollary}[subsection]
\newsavebox{\simhpersontabbox}
\newsavebox{\simhwalltabbox}
\newsavebox{\traintabbox}
\begin{document}
\title{Generalization Theory for Through-the-Wall Radar Human Activity Recognition
\thanks{This work was supported in part by the National Natural Science Foundation of China Young Students Basic Research Project for Doctoral Students under Grant 625B2023, in part by Beijing Science and Technology Nova Program under Grant 20240484569, and in part by Shandong Provincial Natural Science Foundation under Grant ZR2024MF117.\par
Weicheng Gao is with the School of Information and Electronics, Beijing Institute of Technology, Beijing 100081, China, and with the Key Laboratory of Electronic and Information Technology in Satellite Navigation, Beijing Institute of Technology, Beijing 100081, China (e-mail: JoeyBG@126.com).\par}}
\author{Weicheng~Gao\orcidA{},~\IEEEmembership{Graduate~Student~Member,~IEEE}
        \vspace{-0.2cm}
}
\markboth{IEEE Transactions on Information Theory, arXiv Preprint}%
{Gao \MakeLowercase{\textit{et al.}}: Generalization Theory for Through-the-Wall Radar Human Activity Recognition}
\maketitle
\pagestyle{headings}
\setlength{\abovedisplayskip}{6pt}
\setlength{\belowdisplayskip}{6pt}
\setlength{\abovedisplayshortskip}{4pt}
\setlength{\belowdisplayshortskip}{4pt}
\setlength{\jot}{2pt}
\begin{abstract}
Through-the-wall radar (TWR) human activity recognition (HAR) is important for non-line-of-sight indoor sensing, security monitoring, and emergency rescue. However, structured distribution shifts caused by person variation, observation-view variation, and wall-condition variation severely degrade recognition generalization, while the origin of the target-domain error still lacks a rigorous theoretical explanation. To address this issue, a generalization-analysis framework for TWR HAR is proposed in this paper. First, models for indoor human kinematics, TWR echo generation, radar image formation, feature representation, and bounded-weight neural networks are established within a unified source-to-target learning formulation. Then, the source risk, target risk, empirical risk, and admissible physical domain descriptor are defined, and a unified target-domain generalization bound is derived. Next, the structured shift term is decomposed into cross-person, cross-view, and cross-wall components, and the bound-tightening effects of physical low-dimensional representations, multi-source training, and parameter-space coverage are analyzed. Simulated and measured experiments jointly support the resulting theoretical analysis and illustrate its application value.\par
\end{abstract}
\IEEEpeerreviewmaketitle

\section{Introduction}
Through-the-wall radar (TWR) human activity recognition (HAR) is intended to recognize human motions from radar echoes collected under non-line-of-sight sensing conditions. Unlike optical cameras and wearable sensors, TWR can preserve sensing capability when the body is hidden by occlusion, weak illumination, smoke, or wall obstruction. This capability has supported studies on real-time through-the-wall activity recognition \cite{CuiGuoRadarConf}, single-channel ultrawideband pose estimation \cite{JinMDPose}, radar-system-agnostic HAR \cite{DingRSA}, through-the-wall posture reconstruction \cite{YePosture}, range-max enhanced behind-the-wall micro-Doppler analysis \cite{AnRangeMax}, channel-capacity-aware fine-grained multiple-input multiple-output (MIMO) classification \cite{QiChannelCapacity}, mixed convolutional neural network (CNN) multi-spectrogram recognition \cite{JiaMixedCNN}, and time-frequency-enhanced through-wall localization \cite{PengTFE}. A broader review of recent TWR HAR development was also reported in \cite{UWBDeepLearning2024}. However, the main difficulty of TWR HAR is not merely empirical classification accuracy. The training and testing data are often separated by structured distribution shifts induced by radar propagation, human motion, and observation geometry. Three representative difficulties are cross-person generalization, cross-view generalization, and cross-wall generalization. In cross-person generalization, body structure, gait pattern, limb swing, motion velocity, and motion phase change the micro-Doppler features and radar-image distribution, as also reflected in subject-dependent activity-classification and fine-grained recognition studies \cite{GMFN2022}. In cross-view generalization, human motion direction and radar view angle modify the radial-velocity projection, Doppler frequency shift, micro-Doppler signature, and Doppler-time representation, which is closely related to multistatic orientation-sensitive and lightweight multiview radar recognition \cite{ViewSpatialPIERS2024,LightweightIRUWB2025,RadarFormer2024}. In cross-wall generalization, wall material, thickness, permittivity, loss, and structure alter the through-wall propagation operator, causing attenuation, multipath, clutter, low signal-to-noise ratios, and changed noise statistics, as reported in measured MIMO through-the-wall HAR experiments \cite{GEBTWRHARConf2024}.\par
Existing radar HAR and TWR HAR methods were mostly designed to improve empirical recognition performance. Handcrafted micro-Doppler descriptors and time-frequency pipelines \cite{LiDopplerReview,KimCNN}, deep recurrent or complex-valued networks \cite{DingConvLSTM,LiSemisupervised,PhysicsAwareGAN2023,Trajectory3DChannel2021}, hyperdimensional or multiview architectures \cite{YaoHDC,ZhouCorruption2024,MIMOHTL2026}, and model-based augmentation strategies \cite{RojhaniAug,ThroughWallSmallSample2018,WiringEffectsGAN2021} were used to extract discriminative patterns or enlarge the effective training set. Human body-part Doppler mechanisms were also studied to explain weak scattering and angular sensitivity \cite{AbadpourRCS}, and broader radar HAR surveys were reported in \cite{RadarSurveyMicrowaves2023,RadarDataSurvey2024,AhmedHealthcare2024}. Encouraging results were obtained in many controlled scenarios. Nevertheless, a unified theoretical explanation is still lacking for the origin of the generalization error, the action mechanism of physical domain shifts, and the control of target-domain risk from source-domain observations. It therefore remains unclear whether the performance loss is dominated by sample insufficiency and excessive hypothesis complexity, by person-, view-, and wall-induced mismatch, or by irreducible sensing ambiguity, even though related domain-generalization evidence has been discussed in \cite{SahliGeneralization2021}.\par
\begin{figure}[!t]
\centering
\includegraphics[width=\textwidth]{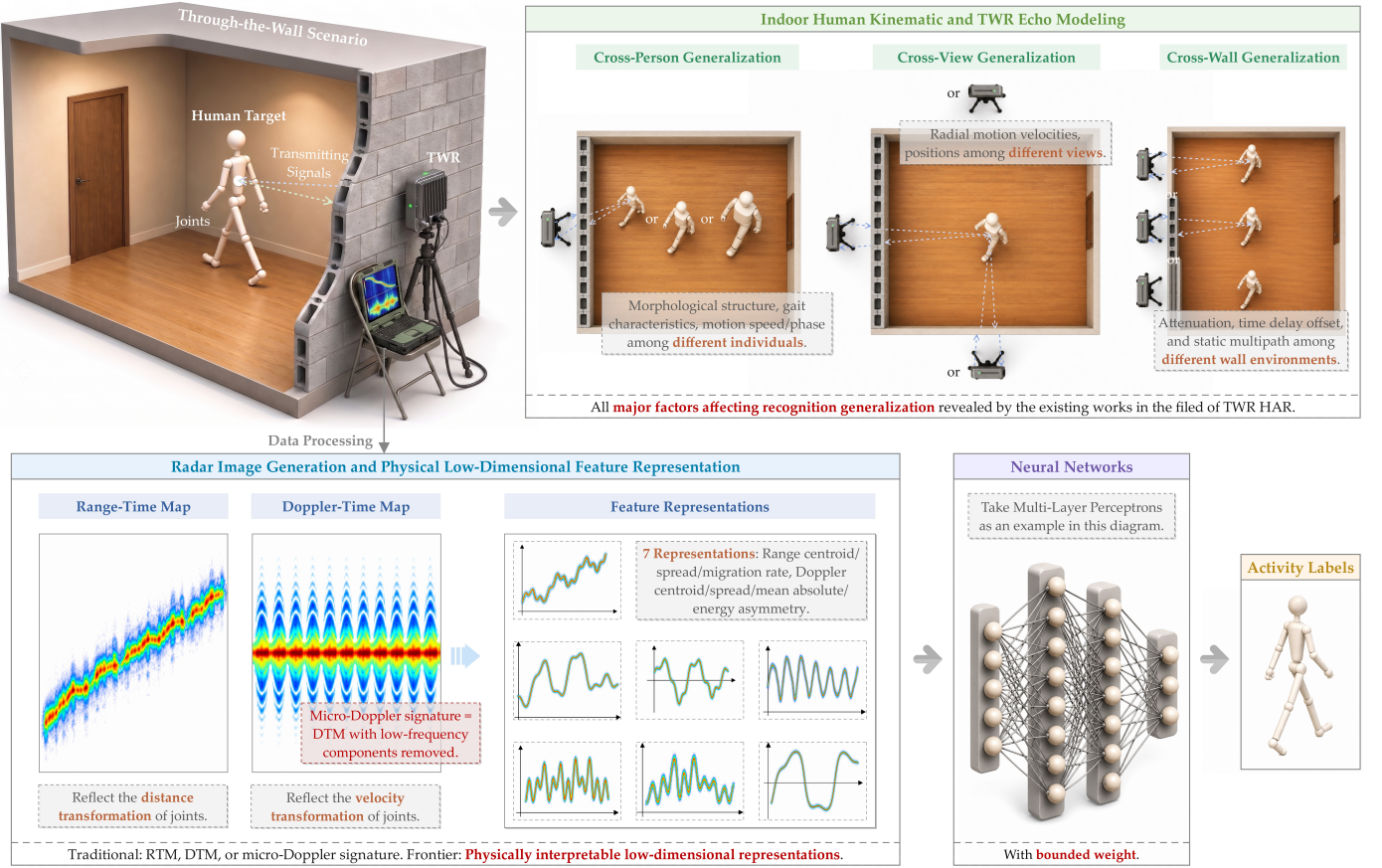}
\caption{Overall illustration of the TWR HAR framework and structured generalization factors.}
\label{fig:TWR_HAR_Overall}
\vspace{-0.2cm}
\end{figure}
In this paper, TWR HAR is studied as a structured source-to-target learning problem jointly determined by human kinematic parameters, view angle, wall parameters, radar image generation, feature representation, and a bounded-weight neural network. The sensing process is regarded as a composition of physical and statistical mappings: human motion generates time-varying scatterers, the view determines observable radial motion, the wall determines the propagation operator and noise condition, and the representation operator supplies the classifier input. Under this viewpoint, the target-domain risk is controlled by source empirical risk and excess-risk mechanisms studied in information-theoretic learning theory \cite{FastRateIT2025,MinimumExcessRisk2022,MinimaxExcessRisk2023}, by transfer and meta-learning discrepancy analyses \cite{TransferMetaLearning2022,TransferLearningIT2024,DomainGeneralizationIT2025}, and by bounded-network sensitivity together with irreducible-error considerations \cite{LossSurfaceIR2023,GibbsGeneralization2024,InfoGenDNN2025}.\par
As shown in Fig.~\ref{fig:TWR_HAR_Overall}, the main contributions of this paper are summarized as follows.
\begin{enumerate}
\renewcommand{\labelenumi}{(\arabic{enumi})}
\setlength{\itemsep}{0em}
\setlength{\parsep}{0em}
\item \textbf{Statistical Learning Formulation:} A statistical learning formulation is established for TWR HAR, where human kinematics, view geometry, wall propagation, feature generation, bounded network weights, and domain risks are described in one source-to-target framework.
\item \textbf{Source-to-Target Generalization Bound:} Under the bounded-network and domain-descriptor assumptions, a source-to-target generalization error bound is derived. The target-domain risk is separated into source empirical risk, model complexity, TWR-induced distribution shift, and a joint source-target approximation term.
\item \textbf{Physically Interpretable Shift Decomposition:} The TWR-induced shift is decomposed into cross-person, cross-view, and cross-wall components, and the bound-level roles of physical low-dimensional representations, multi-source training, and parameter-space coverage are analyzed.
\end{enumerate}\par
The remainder of this paper is organized as follows. Section II establishes the TWR HAR problem formulation. Section III derives generalization error bounds under structured TWR domain shifts. Section IV presents simulated and measured experiments. Section V concludes this paper.\par

\section{TWR-HAR Problem Formulation}

\subsection{Indoor Human Kinematic Model}
For the TWR HAR task, human motion is represented in an indoor Cartesian coordinate system. Let $\mathcal{A}=\{1,\ldots,C\}$ denote the activity-label set, where $C$ is the number of activities. Let $p$ denote a subject index, let $a\in\mathcal{A}$ denote the activity executed by subject $p$, let $m\in\{1,\ldots,M\}$ denote the index of a dominant body scattering center, where $M$ is the number of retained scattering centers, and let $t\in[0,T_{a}]$ denote time, where $T_{a}>0$ is the duration of activity $a$. For cross-person analysis, it is assumed that the slow-time sampling window length used for HAR is fixed, that is, the same activity has a subject-independent duration, namely
\begin{equation}
T_{p,a}=T_{a},
\label{eq:uniform_activity_duration}
\end{equation}
for every admissible subject $p$ executing activity $a$. The normalized motion phase is defined as
\begin{equation}
\phi_{a}(t)=\frac{t}{T_{a}},
\label{eq:activity_phase}
\end{equation}
where $\phi_{a}(t)\in[0,1]$ is used for phase-aligned cross-person comparison.\par
A hybrid kinematic model is adopted. Let $\boldsymbol{\eta}_{p}\in\mathbb{R}^{d_{\eta}}$ denote the morphology vector of subject $p$, where $d_{\eta}$ is the morphology dimension. Let $\mathbf{x}_{p,a}(t)\in\mathbb{R}^{3}$ denote the body-centroid translation, let $\boldsymbol{\vartheta}_{p,a}(t)\in\mathbb{R}^{d_{\vartheta}}$ denote the posture vector, where $d_{\vartheta}$ is the posture dimension, and let $\mathbf{R}_{p,a}(t)\in\mathbb{R}^{3\times 3}$ denote the body-orientation matrix from the local body frame to the indoor global frame. The human state vector is
\begin{equation}
\mathbf{s}_{p,a}^{(\mathrm{hum})}(t)=
\begin{bmatrix}
\mathbf{x}_{p,a}^{\top}(t) &
\boldsymbol{\vartheta}_{p,a}^{\top}(t) &
\dot{\mathbf{x}}_{p,a}^{\top}(t) &
\dot{\boldsymbol{\vartheta}}_{p,a}^{\top}(t) &
\phi_{a}(t) &
\boldsymbol{\eta}_{p}^{\top}
\end{bmatrix}^{\top},
\end{equation}
where $\dot{\mathbf{x}}_{p,a}(t)\in\mathbb{R}^{3}$ and $\dot{\boldsymbol{\vartheta}}_{p,a}(t)\in\mathbb{R}^{d_{\vartheta}}$ denote the centroid velocity and posture rate, respectively.\par
Note that the components of $\mathbf{s}_{p,a}^{(\mathrm{hum})}(t)$ carry different physical units. Therefore, whenever Euclidean norms of state differences are evaluated later, the corresponding coordinates are implicitly assumed to have been normalized beforehand.\par
Let $\mathbf{g}_{m}(\boldsymbol{\vartheta}_{p,a}(t),\boldsymbol{\eta}_{p})\in\mathbb{R}^{3}$ denote the local-coordinate position of the $m$th scattering center. Its global trajectory is
\begin{equation}
\mathbf{r}_{p,a,m}(t)=\mathbf{x}_{p,a}(t)+\mathbf{R}_{p,a}(t)\mathbf{g}_{m}\big(\boldsymbol{\vartheta}_{p,a}(t),\boldsymbol{\eta}_{p}\big),
\end{equation}
where $\mathbf{r}_{p,a,m}(t)\in\mathbb{R}^{3}$ denotes the spatial position of the $m$th scatterer.\par
To expose the Doppler dependence, define the posture Jacobian
\begin{equation}
\mathbf{J}_{m,p,a}(t)=\frac{\partial \mathbf{g}_{m}(\boldsymbol{\vartheta},\boldsymbol{\eta}_{p})}{\partial \boldsymbol{\vartheta}}\bigg|_{\boldsymbol{\vartheta}=\boldsymbol{\vartheta}_{p,a}(t)},
\end{equation}
where $\mathbf{J}_{m,p,a}(t)\in\mathbb{R}^{3\times d_{\vartheta}}$. By the chain rule,
\begin{equation}
\mathbf{v}_{p,a,m}(t)=\dot{\mathbf{r}}_{p,a,m}(t)=\dot{\mathbf{x}}_{p,a}(t)+\dot{\mathbf{R}}_{p,a}(t)\mathbf{g}_{m}\big(\boldsymbol{\vartheta}_{p,a}(t),\boldsymbol{\eta}_{p}\big)+\mathbf{R}_{p,a}(t)\mathbf{J}_{m,p,a}(t)\dot{\boldsymbol{\vartheta}}_{p,a}(t),
\end{equation}
where $\mathbf{v}_{p,a,m}(t)\in\mathbb{R}^{3}$ contains translation, rigid-body rotation, and local articulation terms.\par
For phase alignment, define
\begin{equation}
\widetilde{\mathbf{s}}_{p,a}(\phi):=\widetilde{\mathbf{s}}_{p,a}^{(\mathrm{hum})}(\phi)=\mathbf{s}_{p,a}^{(\mathrm{hum})}(T_{a}\phi),
\label{eq:phase_aligned_human_state}
\end{equation}
and similarly $\widetilde{\mathbf{r}}_{p,a,m}(\phi)=\mathbf{r}_{p,a,m}(T_{a}\phi)$ and $\widetilde{\mathbf{v}}_{p,a,m}(\phi)=\mathbf{v}_{p,a,m}(T_{a}\phi)$. For compactness, define
\begin{subequations}\label{eq:cross_person_component_discrepancies}
\begin{align}
\Delta_{s,p,p'}(a,\phi)&=\|\widetilde{\mathbf{s}}_{p,a}(\phi)-\widetilde{\mathbf{s}}_{p',a}(\phi)\|_{2},\\
\overline{\Delta}_{s,p,p'}(a)&=\sup_{\phi\in[0,1]}\Delta_{s,p,p'}(a,\phi),\\
\overline{\Delta}_{r,p,p'}(a)&=M^{-1}\sum_{m=1}^{M}\sup_{\phi\in[0,1]}\|\widetilde{\mathbf{r}}_{p,a,m}(\phi)-\widetilde{\mathbf{r}}_{p',a,m}(\phi)\|_{2},\\
\overline{\Delta}_{v,p,p'}(a)&=M^{-1}\sum_{m=1}^{M}\sup_{\phi\in[0,1]}\|\widetilde{\mathbf{v}}_{p,a,m}(\phi)-\widetilde{\mathbf{v}}_{p',a,m}(\phi)\|_{2},
\end{align}
\end{subequations}
where $\|\cdot\|_{2}$ denotes the Euclidean norm.\par
The activity-conditional cross-person discrepancy is
\begin{equation}
\delta_{p,p'}^{(\mathrm{cp})}(a)=\omega_{s}\overline{\Delta}_{s,p,p'}(a)+\omega_{r}\overline{\Delta}_{r,p,p'}(a)+\omega_{v}\overline{\Delta}_{v,p,p'}(a),
\label{eq:cross_person_discrepancy}
\end{equation}
where $\omega_{s}\geq 0$, $\omega_{r}>0$, and $\omega_{v}>0$ are fixed weights. The aggregated discrepancy is
\begin{equation}
D_{p,p'}^{(\mathrm{cp})}=\sum_{a=1}^{C}\rho_{a}\,\delta_{p,p'}^{(\mathrm{cp})}(a),
\label{eq:cross_person_average_discrepancy}
\end{equation}
where $\rho_{a}\geq 0$ is the prior weight of activity $a$ and $\sum_{a=1}^{C}\rho_{a}=1$.\par
Assume there exist finite constants $B_{x}$, $B_{\dot{x}}$, $B_{\vartheta}$, $B_{\dot{\vartheta}}$, $B_{g}$, $B_{J}$, and $B_{\Omega}$ such that, for all admissible $p$, $a$, $m$, and $t$,
\begin{equation}
\begin{gathered}
\|\mathbf{x}_{p,a}(t)\|_{2} \leq B_{x},\quad
\|\dot{\mathbf{x}}_{p,a}(t)\|_{2} \leq B_{\dot{x}},\quad
\|\boldsymbol{\vartheta}_{p,a}(t)\|_{2} \leq B_{\vartheta},\quad
\|\dot{\boldsymbol{\vartheta}}_{p,a}(t)\|_{2} \leq B_{\dot{\vartheta}},\\
\|\mathbf{g}_{m}(\boldsymbol{\vartheta}_{p,a}(t),\boldsymbol{\eta}_{p})\|_{2} \leq B_{g},\quad
\|\mathbf{J}_{m,p,a}(t)\|_{2} \leq B_{J},\quad
\|\dot{\mathbf{R}}_{p,a}(t)\|_{2} \leq B_{\Omega},
\end{gathered}
\end{equation}
where the constants bound translation, speed, posture amplitude, posture rate, local displacement, posture sensitivity, and rotation rate, respectively.\par
Since $\|\mathbf{R}_{p,a}(t)\|_{2}=1$, the triangle inequality and submultiplicativity give
\begin{equation}
\|\mathbf{r}_{p,a,m}(t)\|_{2}\leq B_{r}:=B_{x}+B_{g},\qquad
\|\mathbf{v}_{p,a,m}(t)\|_{2}\leq B_{v}:=B_{\dot{x}}+B_{\Omega}B_{g}+B_{J}B_{\dot{\vartheta}}.
\end{equation}
\par
The constants $B_{r}$ and $B_{v}$ are retained for the subsequent echo and risk analysis.\par
Assume further that $\mathbf{g}_{m}(\boldsymbol{\vartheta},\boldsymbol{\eta})$ is Lipschitz continuous on the admissible set, namely
\begin{equation}
\big\|\mathbf{g}_{m}(\boldsymbol{\vartheta}_{1},\boldsymbol{\eta}_{1})-\mathbf{g}_{m}(\boldsymbol{\vartheta}_{2},\boldsymbol{\eta}_{2})\big\|_{2}\leq L_{g,\vartheta}\|\boldsymbol{\vartheta}_{1}-\boldsymbol{\vartheta}_{2}\|_{2}+L_{g,\eta}\|\boldsymbol{\eta}_{1}-\boldsymbol{\eta}_{2}\|_{2},
\end{equation}
where $L_{g,\vartheta}\geq 0$ and $L_{g,\eta}\geq 0$ are the posture- and morphology-dependent Lipschitz constants, respectively. For any admissible $\boldsymbol{\vartheta}_{1}$, $\boldsymbol{\vartheta}_{2}$, $\boldsymbol{\eta}_{1}$, and $\boldsymbol{\eta}_{2}$, by inserting and subtracting $\widetilde{\mathbf{R}}_{p',a}(\phi)\mathbf{g}_{m}(\widetilde{\boldsymbol{\vartheta}}_{p,a}(\phi),\boldsymbol{\eta}_{p})$ and using $\|\widetilde{\mathbf{R}}_{p',a}(\phi)\|_{2}=1$, one obtains
\begin{equation}
\begin{aligned}
\big\|\widetilde{\mathbf{r}}_{p,a,m}(\phi)-\widetilde{\mathbf{r}}_{p',a,m}(\phi)\big\|_{2}
\leq{}&\big\|\widetilde{\mathbf{x}}_{p,a}(\phi)-\widetilde{\mathbf{x}}_{p',a}(\phi)\big\|_{2}+B_{g}\big\|\widetilde{\mathbf{R}}_{p,a}(\phi)-\widetilde{\mathbf{R}}_{p',a}(\phi)\big\|_{2}\\
&+L_{g,\vartheta}\big\|\widetilde{\boldsymbol{\vartheta}}_{p,a}(\phi)-\widetilde{\boldsymbol{\vartheta}}_{p',a}(\phi)\big\|_{2}+L_{g,\eta}\|\boldsymbol{\eta}_{p}-\boldsymbol{\eta}_{p'}\|_{2},
\end{aligned}
\end{equation}
where $\widetilde{\mathbf{x}}_{p,a}(\phi)=\mathbf{x}_{p,a}(T_{a}\phi)$, $\widetilde{\mathbf{R}}_{p,a}(\phi)=\mathbf{R}_{p,a}(T_{a}\phi)$, and $\widetilde{\boldsymbol{\vartheta}}_{p,a}(\phi)=\boldsymbol{\vartheta}_{p,a}(T_{a}\phi)$. Therefore, cross-person and cross-view shifts are reduced to morphology, posture, translation, orientation, and velocity mismatch, which are the physical inputs to the later echo and target risk bounds.\par

\subsection{TWR Echo Model}
Based on the scatterer trajectories and velocities established in the previous subsection, the observation view, wall propagation, and receiver disturbance are now incorporated into a unified radar echo model.\par
A single effective observation channel is adopted for the TWR HAR task. Let $\mathcal{A}=\{1,\ldots,C\}$ denote the activity-label set, where $C$ is the number of activities. The subject index $p$, the activity label $a\in\mathcal{A}$, the scattering-center index $m\in\{1,\ldots,M\}$, the slow time $t$, the scatterer trajectory $\mathbf{r}_{p,a,m}(t)\in\mathbb{R}^{3}$, the scatterer velocity $\mathbf{v}_{p,a,m}(t)\in\mathbb{R}^{3}$, and the bounds $B_{r}$ and $B_{v}$ retain the definitions introduced in the indoor human kinematic model. Let $\beta$ denote the observation-view condition, let $\mathbf{q}_{\beta}\in\mathbb{R}^{3}$ denote the radar phase-center location under view $\beta$, let $\mathbf{u}_{\beta}\in\mathbb{R}^{3}$ denote the unit line-of-sight vector from the scene center toward the radar, let $\boldsymbol{\xi}\in\mathbb{R}^{d_{\xi}}$ denote the wall-parameter vector, where $d_{\xi}$ is the wall-parameter dimension, let $\tau$ denote fast time, let $u(\tau)\in\mathbb{C}$ denote the transmitted baseband waveform, and let $\varrho_{p,a,m}(t)\in\mathbb{C}$ denote the complex scattering coefficient of the $m$th moving scatterer.\par
For the $m$th scatterer, the round-trip propagation distance under view $\beta$ is modeled by
\begin{equation}
d_{p,a,m}^{(\beta)}(t)=2\big\|\mathbf{r}_{p,a,m}(t)-\mathbf{q}_{\beta}\big\|_{2},
\end{equation}
where $d_{p,a,m}^{(\beta)}(t)>0$ denotes the two-way path length. The corresponding propagation delay is
\begin{equation}
\tau_{p,a,m}^{(\beta)}(t)=\frac{d_{p,a,m}^{(\beta)}(t)}{c},
\end{equation}
where $c>0$ is the speed of light. Since
\begin{equation}
\frac{d}{dt}\big\|\mathbf{r}_{p,a,m}(t)-\mathbf{q}_{\beta}\big\|_{2}
=\frac{\big(\mathbf{r}_{p,a,m}(t)-\mathbf{q}_{\beta}\big)^{\top}\mathbf{v}_{p,a,m}(t)}{\big\|\mathbf{r}_{p,a,m}(t)-\mathbf{q}_{\beta}\big\|_{2}},
\end{equation}
the exact delay derivative is written as
\begin{equation}
\dot{\tau}_{p,a,m}^{(\beta)}(t)=\frac{2}{c}\frac{\big(\mathbf{r}_{p,a,m}(t)-\mathbf{q}_{\beta}\big)^{\top}\mathbf{v}_{p,a,m}(t)}{\big\|\mathbf{r}_{p,a,m}(t)-\mathbf{q}_{\beta}\big\|_{2}}.
\end{equation}
\par
Under the usual far-field linearization with respect to the indoor scene extent, a view-dependent reference delay $\tau_{0,\beta}>0$ is introduced such that
\begin{equation}
\tau_{p,a,m}^{(\beta)}(t)\approx \tau_{0,\beta}-\frac{2}{c}\mathbf{u}_{\beta}^{\top}\mathbf{r}_{p,a,m}(t),
\end{equation}
where $\tau_{0,\beta}$ absorbs the nominal scene-to-radar range under view $\beta$. Therefore,
\begin{equation}
\dot{\tau}_{p,a,m}^{(\beta)}(t)\approx -\frac{2}{c}\mathbf{u}_{\beta}^{\top}\mathbf{v}_{p,a,m}(t),
\end{equation}
and the instantaneous Doppler frequency shift is represented by
\begin{equation}
f_{p,a,m}^{(\beta)}(t)=-f_{c}\dot{\tau}_{p,a,m}^{(\beta)}(t)\approx \frac{2f_{c}}{c}\mathbf{u}_{\beta}^{\top}\mathbf{v}_{p,a,m}(t)=\frac{2}{\lambda_{c}}\mathbf{u}_{\beta}^{\top}\mathbf{v}_{p,a,m}(t),
\end{equation}
where $f_{c}>0$ is the carrier frequency and $\lambda_{c}=c/f_{c}$ is the carrier wavelength. It is therefore seen that the cross-view effect enters the echo model through the radial-velocity projection $\mathbf{u}_{\beta}^{\top}\mathbf{v}_{p,a,m}(t)$.\par
The wall is described by an effective two-way propagation operator rather than by an explicit multilayer interface expansion. Let
\begin{equation}
H_{\mathrm{w}}(\omega;\beta,\boldsymbol{\xi})=A_{\mathrm{w}}(\omega;\beta,\boldsymbol{\xi})e^{-j\omega\tau_{\mathrm{w}}(\beta,\boldsymbol{\xi})},
\end{equation}
where $\omega\in\mathbb{R}$ denotes angular frequency, $A_{\mathrm{w}}(\omega;\beta,\boldsymbol{\xi})\in\mathbb{C}$ denotes the effective transmission factor, and $\tau_{\mathrm{w}}(\beta,\boldsymbol{\xi})\geq 0$ denotes the wall-induced excess delay. The parameter vector $\boldsymbol{\xi}$ is allowed to encode wall thickness, effective permittivity, conductivity or loss tangent, and structural descriptors such as layered composition or internal inhomogeneity. When the wall response is sufficiently smooth over the occupied signal band, the dominant amplitude distortion may be evaluated around the carrier, which yields a compact baseband approximation with $A_{\mathrm{w}}(\omega_{c};\beta,\boldsymbol{\xi})$, where $\omega_{c}=2\pi f_{c}$.\par
Let $\kappa_{0}\in\mathbb{C}$ denote the lumped system coefficient containing the transmit power, antenna gain, front-end scaling, and deterministic calibration constants. Let $b^{(\beta,\boldsymbol{\xi})}(\tau,t)\in\mathbb{C}$ denote the static clutter term induced by the wall and indoor stationary reflectors, and let $n^{(\beta,\boldsymbol{\xi})}(\tau,t)\in\mathbb{C}$ denote the additive receiver noise. Then the received continuous-time baseband echo is modeled by
\begin{equation}
\begin{aligned}
y_{p,a}^{(\beta,\boldsymbol{\xi})}(\tau,t)
=&\sum_{m=1}^{M}\kappa_{0}\varrho_{p,a,m}(t)A_{\mathrm{w}}(\omega_{c};\beta,\boldsymbol{\xi})\big[d_{p,a,m}^{(\beta)}(t)\big]^{-2}u\big(\tau-\tau_{\mathrm{w}}(\beta,\boldsymbol{\xi})-\tau_{p,a,m}^{(\beta)}(t)\big)\\
&\times e^{-j2\pi f_{c}\left(\tau_{\mathrm{w}}(\beta,\boldsymbol{\xi})+\tau_{p,a,m}^{(\beta)}(t)\right)}+b^{(\beta,\boldsymbol{\xi})}(\tau,t)+n^{(\beta,\boldsymbol{\xi})}(\tau,t),
\end{aligned}
\end{equation}
where the factor $\big[d_{p,a,m}^{(\beta)}(t)\big]^{-2}$ models the two-way spherical spreading loss and the phase term retains the carrier-sensitive round-trip delay. By substituting the far-field delay approximation into the carrier phase, the oscillatory component becomes
\begin{equation}
e^{-j2\pi f_{c}\left(\tau_{\mathrm{w}}+\tau_{p,a,m}^{(\beta)}(t)\right)}\approx e^{-j2\pi f_{c}\left(\tau_{\mathrm{w}}+\tau_{0,\beta}\right)}e^{j\frac{4\pi}{\lambda_{c}}\mathbf{u}_{\beta}^{\top}\mathbf{r}_{p,a,m}(t)},
\end{equation}
where $\tau_{\mathrm{w}}=\tau_{\mathrm{w}}(\beta,\boldsymbol{\xi})$ is used for compactness. Moreover, by introducing an auxiliary integration variable $\zeta$, the identity
\begin{equation}
\mathbf{u}_{\beta}^{\top}\mathbf{r}_{p,a,m}(t)=\mathbf{u}_{\beta}^{\top}\mathbf{r}_{p,a,m}(0)+\int_{0}^{t}\mathbf{u}_{\beta}^{\top}\mathbf{v}_{p,a,m}(\zeta)\,d\zeta,
\end{equation}
shows that the time-varying part of the carrier phase is equivalently represented, up to an initial constant phase factor, by $\exp\!\left(j2\pi\int_{0}^{t}f_{p,a,m}^{(\beta)}(\zeta)\,d\zeta\right)$. This form makes the dependence on articulation-induced micro-Doppler modulation explicit.\par
The matched-filter output with respect to $u(\tau)$ is defined by
\begin{equation}
z_{p,a}^{(\beta,\boldsymbol{\xi})}(\tau,t)=\int_{\mathbb{R}}y_{p,a}^{(\beta,\boldsymbol{\xi})}(\tau',t)u^{*}(\tau'-\tau)\,d\tau',
\end{equation}
where $(\cdot)^{*}$ denotes complex conjugation and $\tau'$ is an integration variable on the fast-time axis. Let the waveform autocorrelation be
\begin{equation}
\chi_{u}(\Delta)=\int_{\mathbb{R}}u(\tau')u^{*}(\tau'-\Delta)\,d\tau',
\end{equation}
where $\Delta\in\mathbb{R}$ is a delay offset. Then direct substitution of the continuous echo model gives
\begin{equation}
z_{p,a}^{(\beta,\boldsymbol{\xi})}(\tau,t)=\sum_{m=1}^{M}\Gamma_{p,a,m}^{(\beta,\boldsymbol{\xi})}(t)\chi_{u}\big(\tau-\tau_{\mathrm{w}}(\beta,\boldsymbol{\xi})-\tau_{p,a,m}^{(\beta)}(t)\big)+b_{z}^{(\beta,\boldsymbol{\xi})}(\tau,t)+n_{z}^{(\beta,\boldsymbol{\xi})}(\tau,t),
\end{equation}
where
\begin{equation}
\Gamma_{p,a,m}^{(\beta,\boldsymbol{\xi})}(t)=\kappa_{0}\varrho_{p,a,m}(t)A_{\mathrm{w}}(\omega_{c};\beta,\boldsymbol{\xi})\big[d_{p,a,m}^{(\beta)}(t)\big]^{-2}e^{-j2\pi f_{c}\left(\tau_{\mathrm{w}}(\beta,\boldsymbol{\xi})+\tau_{p,a,m}^{(\beta)}(t)\right)},
\end{equation}
\begin{equation}
b_{z}^{(\beta,\boldsymbol{\xi})}(\tau,t)=\int_{\mathbb{R}}b^{(\beta,\boldsymbol{\xi})}(\tau',t)u^{*}(\tau'-\tau)\,d\tau',
\end{equation}
and
\begin{equation}
n_{z}^{(\beta,\boldsymbol{\xi})}(\tau,t)=\int_{\mathbb{R}}n^{(\beta,\boldsymbol{\xi})}(\tau',t)u^{*}(\tau'-\tau)\,d\tau'.
\end{equation}
\par
The coefficient $\Gamma_{p,a,m}^{(\beta,\boldsymbol{\xi})}(t)$ absorbs the moving-scatterer reflectivity, the wall attenuation, the spherical spreading loss, and the carrier phase, whereas $\chi_{u}$ determines the range-resolution kernel.\par
Let $\Delta_{\tau}>0$ denote the fast-time sampling interval, let $\Delta_{t}>0$ denote the slow-time sampling interval, let $N_{\tau}$ denote the number of fast-time samples, and let $L_{t}$ denote the number of slow-time samples. The discrete radar observation matrix is defined by
\begin{equation}
\big[\mathbf{Z}_{p,a}^{(\beta,\boldsymbol{\xi})}\big]_{i,\ell}=z_{p,a}^{(\beta,\boldsymbol{\xi})}(i\Delta_{\tau},\ell\Delta_{t}),
\end{equation}
where $i\in\{0,\ldots,N_{\tau}-1\}$ and $\ell\in\{0,\ldots,L_{t}-1\}$. If the clutter and noise matrices are defined entrywise by
\begin{equation}
\big[\mathbf{K}^{(\beta,\boldsymbol{\xi})}\big]_{i,\ell}=b_{z}^{(\beta,\boldsymbol{\xi})}(i\Delta_{\tau},\ell\Delta_{t}),\qquad
\big[\mathbf{W}^{(\beta,\boldsymbol{\xi})}\big]_{i,\ell}=n_{z}^{(\beta,\boldsymbol{\xi})}(i\Delta_{\tau},\ell\Delta_{t}),
\end{equation}
then the entire echo-generation process is summarized as
\begin{equation}
\mathbf{Z}_{p,a}^{(\beta,\boldsymbol{\xi})}
=\mathcal{G}_{\mathrm{echo}}\!\left(\{\mathbf{r}_{p,a,m},\mathbf{v}_{p,a,m},\varrho_{p,a,m}\}_{m=1}^{M};\beta,\boldsymbol{\xi}\right)+\mathbf{K}^{(\beta,\boldsymbol{\xi})}+\mathbf{W}^{(\beta,\boldsymbol{\xi})},
\end{equation}
where $\mathcal{G}_{\mathrm{echo}}$ denotes the deterministic nonlinear operator that maps human kinematics, scatterer reflectivity, view geometry, and wall parameters to the ideal matched-filtered radar observation. The subsequent range-time map (RTM), Doppler-time map (DTM), micro-Doppler signature, and physics-guided low-dimensional feature representations are all constructed in the next subsection from clutter-suppressed observations derived from $\mathbf{Z}_{p,a}^{(\beta,\boldsymbol{\xi})}$.\par
For the later generalization analysis, uniform boundedness assumptions are introduced. Assume there exist finite constants $B_{\varrho}>0$, $B_{\mathrm{w}}>0$, $B_{\chi}>0$, $B_{\mathrm{cl}}>0$, and $d_{\min}>0$ such that, for all admissible $p$, $a$, $m$, $t$, $\beta$, $\boldsymbol{\xi}$, and $\Delta$,
\begin{equation}
\big|\varrho_{p,a,m}(t)\big|\leq B_{\varrho},\qquad
\big|A_{\mathrm{w}}(\omega_{c};\beta,\boldsymbol{\xi})\big|\leq B_{\mathrm{w}},\qquad
\big|\chi_{u}(\Delta)\big|\leq B_{\chi},\qquad
\big|b^{(\beta,\boldsymbol{\xi})}(\tau,t)\big|\leq B_{\mathrm{cl}},
\end{equation}
and
\begin{equation}
d_{p,a,m}^{(\beta)}(t)\geq d_{\min}.
\end{equation}
\par
Let $B_{u,1}=\int_{\mathbb{R}}|u(\tau')|\,d\tau'$ denote the waveform $\ell_{1}$-type envelope constant. Then
\begin{equation}
\big|b_{z}^{(\beta,\boldsymbol{\xi})}(\tau,t)\big|
\leq \int_{\mathbb{R}}\big|b^{(\beta,\boldsymbol{\xi})}(\tau',t)\big|\big|u(\tau'-\tau)\big|\,d\tau'
\leq B_{\mathrm{cl}}B_{u,1}
=:B_{\mathrm{cl},z},
\end{equation}
where $B_{\mathrm{cl},z}>0$ is the matched-clutter bound. Moreover,
\begin{equation}
\big|\Gamma_{p,a,m}^{(\beta,\boldsymbol{\xi})}(t)\big|
\leq |\kappa_{0}|B_{\varrho}B_{\mathrm{w}}d_{\min}^{-2}
=:B_{\Gamma},
\end{equation}
where $B_{\Gamma}>0$ is a view- and wall-uniform echo-amplitude bound. Consequently, the deterministic part of the matched-filter output satisfies
\begin{equation}
\left|\sum_{m=1}^{M}\Gamma_{p,a,m}^{(\beta,\boldsymbol{\xi})}(t)\chi_{u}\big(\tau-\tau_{\mathrm{w}}(\beta,\boldsymbol{\xi})-\tau_{p,a,m}^{(\beta)}(t)\big)\right|
\leq MB_{\Gamma}B_{\chi}.
\end{equation}
\par
Assume finally that $n_{z}^{(\beta,\boldsymbol{\xi})}(\tau,t)$ is a zero-mean circularly symmetric complex sub-Gaussian random variable with scale parameter $\varsigma_{\mathrm{n}}(\beta,\boldsymbol{\xi})$. Then
\begin{equation}
\big|z_{p,a}^{(\beta,\boldsymbol{\xi})}(\tau,t)\big|
\leq MB_{\Gamma}B_{\chi}+B_{\mathrm{cl},z}+\big|n_{z}^{(\beta,\boldsymbol{\xi})}(\tau,t)\big|,
\end{equation}
and, by the inequality $|x_{1}+x_{2}+x_{3}|^{2}\leq 3(|x_{1}|^{2}+|x_{2}|^{2}+|x_{3}|^{2})$ together with $\mathbb{E}|n_{z}^{(\beta,\boldsymbol{\xi})}(\tau,t)|^{2}\leq \varsigma_{\mathrm{n}}^{2}(\beta,\boldsymbol{\xi})$,
\begin{equation}
\mathbb{E}\big|z_{p,a}^{(\beta,\boldsymbol{\xi})}(\tau,t)\big|^{2}
\leq 3M^{2}B_{\Gamma}^{2}B_{\chi}^{2}+3B_{\mathrm{cl},z}^{2}+3\varsigma_{\mathrm{n}}^{2}(\beta,\boldsymbol{\xi}).
\end{equation}
\par
This uniform envelope is retained as the signal-domain entry point for the later source-to-target risk bound, where cross-person, cross-view, and cross-wall shifts are quantified through the dependence of $\mathcal{G}_{\mathrm{echo}}$ on $\{\mathbf{r}_{p,a,m},\mathbf{v}_{p,a,m},\varrho_{p,a,m}\}_{m=1}^{M}$, $\beta$, and $\boldsymbol{\xi}$.\par

\subsection{Radar Image Generation and Feature Representation}
Based on the matched-filtered observation matrix established in the previous subsection, the radar image formation chain is introduced following common radar HAR processing pipelines \cite{RadarSurveyMicrowaves2023,RadarDataSurvey2024}. A classifier-oriented feature representation is then constructed in a form consistent with radar sensing applications in indoor security and health monitoring \cite{AhmedHealthcare2024}. Let $\mathbf{Z}_{p,a}^{(\beta,\boldsymbol{\xi})}\in\mathbb{C}^{N_{\tau}\times L_{t}}$ denote the discrete matched-filtered radar observation associated with subject $p$, activity label $a$, view condition $\beta$, and wall-parameter vector $\boldsymbol{\xi}$, where $N_{\tau}$ is the number of fast-time samples and $L_{t}$ is the number of slow-time samples. Since the dominant wall return and static indoor clutter are approximately invariant along slow time, a row-wise clutter-suppression operator is first introduced as
\begin{equation}
\widehat{\mathbf{Z}}_{p,a}^{(\beta,\boldsymbol{\xi})}
=
\mathcal{P}_{\mathrm{cs}}\!\left(\mathbf{Z}_{p,a}^{(\beta,\boldsymbol{\xi})}\right),
\end{equation}
where $\mathcal{P}_{\mathrm{cs}}$ is the row-wise mean-cancellation mapping. Its entries are written as
\begin{equation}
\big[\widehat{\mathbf{Z}}_{p,a}^{(\beta,\boldsymbol{\xi})}\big]_{i,\ell}
=
\big[\mathbf{Z}_{p,a}^{(\beta,\boldsymbol{\xi})}\big]_{i,\ell}
-\frac{1}{L_{t}}\sum_{\ell'=0}^{L_{t}-1}\big[\mathbf{Z}_{p,a}^{(\beta,\boldsymbol{\xi})}\big]_{i,\ell'},
\end{equation}
where $i\in\{0,\ldots,N_{\tau}-1\}$ and $\ell\in\{0,\ldots,L_{t}-1\}$ denote the fast-time and slow-time sample indices, respectively.\par
The range-time map (RTM) is then generated by the entrywise energy of the clutter-suppressed matrix, namely
\begin{equation}
\big[\mathbf{X}_{p,a}^{(\beta,\boldsymbol{\xi},\mathrm{rt})}\big]_{i,\ell}
=
\left|\big[\widehat{\mathbf{Z}}_{p,a}^{(\beta,\boldsymbol{\xi})}\big]_{i,\ell}\right|^{2},
\end{equation}
where $\mathbf{X}_{p,a}^{(\beta,\boldsymbol{\xi},\mathrm{rt})}\in\mathbb{R}_{+}^{N_{\tau}\times L_{t}}$ is the RTM and $\mathbb{R}_{+}$ denotes the nonnegative real axis. This representation preserves the joint evolution of range migration and temporal energy concentration.\par
To form a Doppler-oriented representation, the clutter-suppressed matrix is first collapsed along the fast-time axis as
\begin{equation}
g_{p,a}^{(\beta,\boldsymbol{\xi})}[\ell]
=
\sum_{i=0}^{N_{\tau}-1}\big[\widehat{\mathbf{Z}}_{p,a}^{(\beta,\boldsymbol{\xi})}\big]_{i,\ell},
\end{equation}
where $g_{p,a}^{(\beta,\boldsymbol{\xi})}[\ell]\in\mathbb{C}$ is the range-collapsed slow-time signal. Let $g_{\mathrm{win}}[r]\in\mathbb{R}$ denote an analysis window of length $L_{\mathrm{win}}$, let $S_{\mathrm{hop}}$ denote the frame hop size, and let $N_{f}$ denote the number of Doppler bins. The number of short-time frames is
\begin{equation}
N_{\mathrm{frm}}
=
1+\left\lfloor\frac{L_{t}-L_{\mathrm{win}}}{S_{\mathrm{hop}}}\right\rfloor,
\end{equation}
where $\lfloor\cdot\rfloor$ denotes the floor operator and $L_{t}\geq L_{\mathrm{win}}$ is assumed. The short-time Fourier transform (STFT) coefficient at Doppler bin $k\in\{0,\ldots,N_{f}-1\}$ and frame index $n\in\{0,\ldots,N_{\mathrm{frm}}-1\}$ is defined by
\begin{equation}
S_{p,a}^{(\beta,\boldsymbol{\xi})}[k,n]
=
\sum_{r=0}^{L_{\mathrm{win}}-1}
g_{\mathrm{win}}[r]\,
g_{p,a}^{(\beta,\boldsymbol{\xi})}[nS_{\mathrm{hop}}+r]\,
e^{-j\frac{2\pi}{N_{f}}kr},
\end{equation}
where $r$ is the local window index. The corresponding Doppler-time map (DTM) is
\begin{equation}
\big[\mathbf{X}_{p,a}^{(\beta,\boldsymbol{\xi},\mathrm{dt})}\big]_{k,n}
=
\left|S_{p,a}^{(\beta,\boldsymbol{\xi})}[k,n]\right|^{2},
\end{equation}
where $\mathbf{X}_{p,a}^{(\beta,\boldsymbol{\xi},\mathrm{dt})}\in\mathbb{R}_{+}^{N_{f}\times N_{\mathrm{frm}}}$ collects the time-varying Doppler energy.\par
Let $\Delta_{t}>0$ denote the slow-time sampling interval introduced in the echo model. After spectral centering, the Doppler coordinate associated with bin $k$ is written as
\begin{equation}
f_{k}
=
\frac{k-\lfloor N_{f}/2\rfloor}{N_{f}\Delta_{t}},
\end{equation}
where $f_{k}\in\mathbb{R}$ denotes the discrete Doppler frequency associated with the $k$th row of the DTM. To suppress the residual near-zero Doppler component, let $f_{\mathrm{th}}>0$ denote a deterministic threshold. The micro-Doppler signature is then defined as the zero-Doppler-suppressed special case
\begin{equation}
\big[\mathbf{X}_{p,a}^{(\beta,\boldsymbol{\xi},\mathrm{md})}\big]_{k,n}
=
\mathbf{1}\{|f_{k}|\geq f_{\mathrm{th}}\}
\big[\mathbf{X}_{p,a}^{(\beta,\boldsymbol{\xi},\mathrm{dt})}\big]_{k,n},
\end{equation}
where $\mathbf{X}_{p,a}^{(\beta,\boldsymbol{\xi},\mathrm{md})}\in\mathbb{R}_{+}^{N_{f}\times N_{\mathrm{frm}}}$ is the micro-Doppler signature and $\mathbf{1}\{\cdot\}$ denotes the indicator function.\par
Besides the high-dimensional map inputs, a physics-guided low-dimensional representation is introduced to summarize the dominant kinematic and Doppler attributes. Let $\varepsilon_{\mathrm{nrm}}>0$ denote a deterministic normalization floor. The normalized RTM and micro-Doppler signature energy maps are defined by
\begin{subequations}
\begin{align}
\big[\mathbf{P}_{p,a}^{(\beta,\boldsymbol{\xi},\mathrm{rt})}\big]_{i,\ell}
&=
\frac{\big[\mathbf{X}_{p,a}^{(\beta,\boldsymbol{\xi},\mathrm{rt})}\big]_{i,\ell}+\varepsilon_{\mathrm{nrm}}}
{\sum_{\bar{i}=0}^{N_{\tau}-1}\sum_{\bar{\ell}=0}^{L_{t}-1}\left(\big[\mathbf{X}_{p,a}^{(\beta,\boldsymbol{\xi},\mathrm{rt})}\big]_{\bar{i},\bar{\ell}}+\varepsilon_{\mathrm{nrm}}\right)},\\
\big[\mathbf{P}_{p,a}^{(\beta,\boldsymbol{\xi},\mathrm{md})}\big]_{k,n}
&=
\frac{\big[\mathbf{X}_{p,a}^{(\beta,\boldsymbol{\xi},\mathrm{md})}\big]_{k,n}+\varepsilon_{\mathrm{nrm}}}
{\sum_{\bar{k}=0}^{N_{f}-1}\sum_{\bar{n}=0}^{N_{\mathrm{frm}}-1}\left(\big[\mathbf{X}_{p,a}^{(\beta,\boldsymbol{\xi},\mathrm{md})}\big]_{\bar{k},\bar{n}}+\varepsilon_{\mathrm{nrm}}\right)},
\end{align}
\end{subequations}
where $\mathbf{P}_{p,a}^{(\beta,\boldsymbol{\xi},\mathrm{rt})}\in\mathbb{R}_{+}^{N_{\tau}\times L_{t}}$ and $\mathbf{P}_{p,a}^{(\beta,\boldsymbol{\xi},\mathrm{md})}\in\mathbb{R}_{+}^{N_{f}\times N_{\mathrm{frm}}}$ both have unit total mass.\par
The delay marginal and the Doppler marginal are then given by
\begin{equation}
p_{\tau,p,a}^{(\beta,\boldsymbol{\xi})}[i]
=
\sum_{\ell=0}^{L_{t}-1}\big[\mathbf{P}_{p,a}^{(\beta,\boldsymbol{\xi},\mathrm{rt})}\big]_{i,\ell},\qquad
p_{f,p,a}^{(\beta,\boldsymbol{\xi})}[k]
=
\sum_{n=0}^{N_{\mathrm{frm}}-1}\big[\mathbf{P}_{p,a}^{(\beta,\boldsymbol{\xi},\mathrm{md})}\big]_{k,n},
\end{equation}
where $\sum_{i=0}^{N_{\tau}-1}p_{\tau,p,a}^{(\beta,\boldsymbol{\xi})}[i]=1$ and $\sum_{k=0}^{N_{f}-1}p_{f,p,a}^{(\beta,\boldsymbol{\xi})}[k]=1$. Let $\tau_{i}=i\Delta_{\tau}$ denote the delay coordinate of the $i$th fast-time bin, where $\Delta_{\tau}>0$ is the fast-time sampling interval introduced in the echo model. The delay centroid and delay spread are defined as
\begin{equation}
\mu_{\tau,p,a}^{(\beta,\boldsymbol{\xi})}
=
\sum_{i=0}^{N_{\tau}-1}\tau_{i}\,p_{\tau,p,a}^{(\beta,\boldsymbol{\xi})}[i],\qquad
\sigma_{\tau,p,a}^{(\beta,\boldsymbol{\xi})}
=
\left(\sum_{i=0}^{N_{\tau}-1}\big(\tau_{i}-\mu_{\tau,p,a}^{(\beta,\boldsymbol{\xi})}\big)^{2}p_{\tau,p,a}^{(\beta,\boldsymbol{\xi})}[i]\right)^{\frac{1}{2}},
\end{equation}
where $\mu_{\tau,p,a}^{(\beta,\boldsymbol{\xi})}$ measures the dominant delay location and $\sigma_{\tau,p,a}^{(\beta,\boldsymbol{\xi})}$ measures the delay dispersion. To quantify temporal range migration, the frame-wise delay centroid is defined by
\begin{equation}
\overline{\tau}_{p,a}^{(\beta,\boldsymbol{\xi})}[\ell]
=
\frac{\sum_{i=0}^{N_{\tau}-1}\tau_{i}\left(\big[\mathbf{X}_{p,a}^{(\beta,\boldsymbol{\xi},\mathrm{rt})}\big]_{i,\ell}+\varepsilon_{\mathrm{nrm}}\right)}
{\sum_{i=0}^{N_{\tau}-1}\left(\big[\mathbf{X}_{p,a}^{(\beta,\boldsymbol{\xi},\mathrm{rt})}\big]_{i,\ell}+\varepsilon_{\mathrm{nrm}}\right)},
\end{equation}
where $\overline{\tau}_{p,a}^{(\beta,\boldsymbol{\xi})}[\ell]$ is the delay centroid at slow-time index $\ell$. The mean absolute inter-frame range migration is then written as
\begin{equation}
\eta_{\tau,p,a}^{(\beta,\boldsymbol{\xi})}
=
\frac{1}{L_{t}-1}
\sum_{\ell=0}^{L_{t}-2}
\left|
\overline{\tau}_{p,a}^{(\beta,\boldsymbol{\xi})}[\ell+1]
-\overline{\tau}_{p,a}^{(\beta,\boldsymbol{\xi})}[\ell]
\right|,
\end{equation}
where $\eta_{\tau,p,a}^{(\beta,\boldsymbol{\xi})}\geq 0$ measures the average delay migration speed across consecutive frames.\par
The Doppler centroid, Doppler spread, and mean absolute Doppler are defined from the micro-Doppler marginal as
\begin{equation}
\begin{gathered}
\mu_{f,p,a}^{(\beta,\boldsymbol{\xi})}
=
\sum_{k=0}^{N_{f}-1}f_{k}\,p_{f,p,a}^{(\beta,\boldsymbol{\xi})}[k],\qquad
\sigma_{f,p,a}^{(\beta,\boldsymbol{\xi})}
=
\Bigl(\sum_{k=0}^{N_{f}-1}\big(f_{k}-\mu_{f,p,a}^{(\beta,\boldsymbol{\xi})}\big)^{2}p_{f,p,a}^{(\beta,\boldsymbol{\xi})}[k]\Bigr)^{\frac{1}{2}},\\
\eta_{f,p,a}^{(\beta,\boldsymbol{\xi})}
=
\sum_{k=0}^{N_{f}-1}|f_{k}|\,p_{f,p,a}^{(\beta,\boldsymbol{\xi})}[k],
\end{gathered}
\end{equation}
where $\mu_{f,p,a}^{(\beta,\boldsymbol{\xi})}$ captures the signed Doppler center, $\sigma_{f,p,a}^{(\beta,\boldsymbol{\xi})}$ captures the Doppler bandwidth, and $\eta_{f,p,a}^{(\beta,\boldsymbol{\xi})}$ captures the average absolute radial-motion intensity. To characterize directional asymmetry, define the positive- and negative-Doppler masses as
\begin{equation}
\pi_{f,+,p,a}^{(\beta,\boldsymbol{\xi})}
=
\sum_{k:\,f_{k}>0}p_{f,p,a}^{(\beta,\boldsymbol{\xi})}[k],\qquad
\pi_{f,-,p,a}^{(\beta,\boldsymbol{\xi})}
=
\sum_{k:\,f_{k}<0}p_{f,p,a}^{(\beta,\boldsymbol{\xi})}[k],
\end{equation}
where $\pi_{f,+,p,a}^{(\beta,\boldsymbol{\xi})}\in[0,1]$ and $\pi_{f,-,p,a}^{(\beta,\boldsymbol{\xi})}\in[0,1]$. The positive-minus-negative Doppler energy asymmetry is then defined as
\begin{equation}
\kappa_{f,p,a}^{(\beta,\boldsymbol{\xi})}
=
\pi_{f,+,p,a}^{(\beta,\boldsymbol{\xi})}
-\pi_{f,-,p,a}^{(\beta,\boldsymbol{\xi})},
\end{equation}
where $\kappa_{f,p,a}^{(\beta,\boldsymbol{\xi})}\in[-1,1]$ indicates whether the dominant micro-motion energy is biased toward approaching or receding radial motion.\par
The resulting seven-dimensional physical feature vector is constructed as
\begin{equation}
\boldsymbol{\psi}_{p,a}^{(\beta,\boldsymbol{\xi})}
=
\begin{bmatrix}
\mu_{\tau,p,a}^{(\beta,\boldsymbol{\xi})} &
\sigma_{\tau,p,a}^{(\beta,\boldsymbol{\xi})} &
\eta_{\tau,p,a}^{(\beta,\boldsymbol{\xi})} &
\mu_{f,p,a}^{(\beta,\boldsymbol{\xi})} &
\sigma_{f,p,a}^{(\beta,\boldsymbol{\xi})} &
\eta_{f,p,a}^{(\beta,\boldsymbol{\xi})} &
\kappa_{f,p,a}^{(\beta,\boldsymbol{\xi})}
\end{bmatrix}^{\top},
\end{equation}
where $\boldsymbol{\psi}_{p,a}^{(\beta,\boldsymbol{\xi})}\in\mathbb{R}^{7}$ summarizes delay localization, delay spread, range migration, signed Doppler center, Doppler spread, mean absolute Doppler, and directional Doppler asymmetry, respectively.\par
To unify the later learning model, define the representation index set
\begin{equation}
\mathcal{V}
=
\{\mathrm{rt},\mathrm{dt},\mathrm{md},\mathrm{phy}\},
\end{equation}
where $\mathrm{rt}$, $\mathrm{dt}$, $\mathrm{md}$, and $\mathrm{phy}$ denote the RTM, DTM, micro-Doppler signature, and physics-guided low-dimensional representation, respectively. For each $\upsilon\in\mathcal{V}$, let $\mathcal{R}_{\upsilon}$ denote the corresponding representation operator. The classifier input is written uniformly as
\begin{equation}
\boldsymbol{\varphi}_{p,a}^{(\beta,\boldsymbol{\xi},\upsilon)}
=
\mathcal{R}_{\upsilon}\!\left(\widehat{\mathbf{Z}}_{p,a}^{(\beta,\boldsymbol{\xi})}\right),
\end{equation}
where $\boldsymbol{\varphi}_{p,a}^{(\beta,\boldsymbol{\xi},\upsilon)}\in\mathbb{R}^{d_{\upsilon}}$ is the final feature vector supplied to the subsequent recognition model. More explicitly,
\begin{equation}
\begin{aligned}
\mathcal{R}_{\mathrm{rt}}\!\left(\widehat{\mathbf{Z}}_{p,a}^{(\beta,\boldsymbol{\xi})}\right)
&=
\operatorname{vec}\!\left(\mathbf{X}_{p,a}^{(\beta,\boldsymbol{\xi},\mathrm{rt})}\right),\qquad
\mathcal{R}_{\mathrm{dt}}\!\left(\widehat{\mathbf{Z}}_{p,a}^{(\beta,\boldsymbol{\xi})}\right)
=
\operatorname{vec}\!\left(\mathbf{X}_{p,a}^{(\beta,\boldsymbol{\xi},\mathrm{dt})}\right),\\
\mathcal{R}_{\mathrm{md}}\!\left(\widehat{\mathbf{Z}}_{p,a}^{(\beta,\boldsymbol{\xi})}\right)
&=
\operatorname{vec}\!\left(\mathbf{X}_{p,a}^{(\beta,\boldsymbol{\xi},\mathrm{md})}\right),\qquad
\mathcal{R}_{\mathrm{phy}}\!\left(\widehat{\mathbf{Z}}_{p,a}^{(\beta,\boldsymbol{\xi})}\right)
=
\boldsymbol{\psi}_{p,a}^{(\beta,\boldsymbol{\xi})},
\end{aligned}
\end{equation}
where $\operatorname{vec}(\cdot)$ denotes column-wise vectorization. The corresponding input dimensions are
\begin{equation}
d_{\mathrm{rt}}
=
N_{\tau}L_{t},\qquad
d_{\mathrm{dt}}
=
N_{f}N_{\mathrm{frm}},\qquad
d_{\mathrm{md}}
=
N_{f}N_{\mathrm{frm}},\qquad
d_{\mathrm{phy}}
=
7.
\end{equation}
\par
For the later generalization analysis, the admissible clutter-suppressed observation set is denoted by
\begin{equation}
\mathfrak{Z}
=
\left\{
\widehat{\mathbf{Z}}\in\mathbb{C}^{N_{\tau}\times L_{t}}
:
\|\widehat{\mathbf{Z}}\|_{F}\leq B_{\widehat{z}}
\right\},
\end{equation}
where $B_{\widehat{z}}>0$ is a finite envelope constant and $\|\cdot\|_{F}$ denotes the Frobenius norm. Since $\mathcal{P}_{\mathrm{cs}}$ is linear, the STFT is linear, the modulus-square operation is locally Lipschitz on bounded sets, and the normalized moment map is stabilized by $\varepsilon_{\mathrm{nrm}}>0$, each admissible representation operator is Lipschitz continuous on $\mathfrak{Z}$. Therefore, for every $\upsilon\in\mathcal{V}$, there exists a finite constant $L_{\upsilon}>0$ such that, for all $\widehat{\mathbf{Z}},\widehat{\mathbf{Z}}'\in\mathfrak{Z}$,
\begin{equation}
\left\|
\mathcal{R}_{\upsilon}(\widehat{\mathbf{Z}})
-\mathcal{R}_{\upsilon}(\widehat{\mathbf{Z}}')
\right\|_{2}
\leq
L_{\upsilon}
\left\|
\widehat{\mathbf{Z}}
-\widehat{\mathbf{Z}}'
\right\|_{F},
\end{equation}
where $\|\cdot\|_{2}$ denotes the Euclidean norm in the representation space. Since
\begin{equation}
d_{\mathrm{phy}}\ll d_{\mathrm{rt}},\qquad
d_{\mathrm{phy}}\ll d_{\mathrm{dt}},\qquad
d_{\mathrm{phy}}\ll d_{\mathrm{md}},
\end{equation}
the physics-guided low-dimensional representation is expected to reduce both input dimensionality and representation sensitivity. The generic input $\boldsymbol{\varphi}_{p,a}^{(\beta,\boldsymbol{\xi},\upsilon)}$ defined here is passed directly to the bounded-weight neural-network hypothesis class in Subsection II-D. It also serves as the entry point for the representation-induced bound tightening analysis in Subsection III-C.\par

\subsection{Bounded-Weight Neural Network for TWR HAR}
Based on the unified feature representation introduced in the previous subsection, the TWR HAR recognizer is modeled as a bounded-weight feedforward neural network whose score magnitude and input sensitivity are both uniformly controlled. Such a formulation is consistent with entropy-regularized learning analyses \cite{LossSurfaceIR2023} and recent information-theoretic deep-network generalization bounds \cite{InfoGenDNN2025}. For every admissible representation index $\upsilon\in\mathcal{V}$, the classifier input is
\begin{equation}
\boldsymbol{\varphi}_{p,a}^{(\beta,\boldsymbol{\xi},\upsilon)}
=
\mathcal{R}_{\upsilon}\!\left(\widehat{\mathbf{Z}}_{p,a}^{(\beta,\boldsymbol{\xi})}\right)
\in\mathbb{R}^{d_{\upsilon}},
\end{equation}
where $\mathcal{V}=\{\mathrm{rt},\mathrm{dt},\mathrm{md},\mathrm{phy}\}$, $\widehat{\mathbf{Z}}_{p,a}^{(\beta,\boldsymbol{\xi})}\in\mathfrak{Z}$, and $d_{\upsilon}$ retain the definitions introduced earlier in Section II. Since $\mathbf{0}_{N_{\tau}\times L_{t}}\in\mathfrak{Z}$, by inserting and subtracting $\mathcal{R}_{\upsilon}(\mathbf{0}_{N_{\tau}\times L_{t}})$ and using the representation Lipschitz bound in Subsection II-C, one obtains
\begin{equation}
\begin{aligned}
\left\|
\boldsymbol{\varphi}_{p,a}^{(\beta,\boldsymbol{\xi},\upsilon)}
\right\|_{2}
&=
\left\|
\mathcal{R}_{\upsilon}\!\left(\widehat{\mathbf{Z}}_{p,a}^{(\beta,\boldsymbol{\xi})}\right)
-\mathcal{R}_{\upsilon}(\mathbf{0}_{N_{\tau}\times L_{t}})
+\mathcal{R}_{\upsilon}(\mathbf{0}_{N_{\tau}\times L_{t}})
\right\|_{2}\\
&\leq
L_{\upsilon}
\left\|
\widehat{\mathbf{Z}}_{p,a}^{(\beta,\boldsymbol{\xi})}
\right\|_{F}
+
\left\|
\mathcal{R}_{\upsilon}(\mathbf{0}_{N_{\tau}\times L_{t}})
\right\|_{2}
\leq
L_{\upsilon}B_{\widehat{z}}
+
\left\|
\mathcal{R}_{\upsilon}(\mathbf{0}_{N_{\tau}\times L_{t}})
\right\|_{2}
=:
B_{\varphi,\upsilon},
\end{aligned}
\end{equation}
where $B_{\varphi,\upsilon}>0$ is the representation-dependent input-radius constant, $L_{\upsilon}>0$ is the Lipschitz constant of $\mathcal{R}_{\upsilon}$, and $B_{\widehat{z}}>0$ is the clutter-suppressed observation envelope.\par
Define the admissible representation-input set as
\begin{equation}
\mathfrak{X}_{\upsilon}
=
\left\{
\boldsymbol{\varphi}\in\mathbb{R}^{d_{\upsilon}}
:
\|\boldsymbol{\varphi}\|_{2}\leq B_{\varphi,\upsilon}
\right\},
\end{equation}
where every admissible representation output $\boldsymbol{\varphi}_{p,a}^{(\beta,\boldsymbol{\xi},\upsilon)}$ belongs to $\mathfrak{X}_{\upsilon}$ by construction.\par
Let $Q_{\mathrm{net}}\geq 1$ denote the number of affine layers, let $\nu_{0}=d_{\upsilon}$ denote the input width, let $\nu_{Q_{\mathrm{net}}}=C$ denote the output width, and let $\nu_{1},\ldots,\nu_{Q_{\mathrm{net}}-1}$ denote the hidden-layer widths. For a generic input $\boldsymbol{\varphi}\in\mathbb{R}^{d_{\upsilon}}$, the layerwise state recursion is defined as
\begin{equation}
\mathbf{o}_{0}(\boldsymbol{\varphi})=\boldsymbol{\varphi},
\end{equation}
and
\begin{equation}
\mathbf{o}_{q}(\boldsymbol{\varphi})
=
\mathfrak{a}_{q}\!\left(
\mathbf{A}_{q,\upsilon}\mathbf{o}_{q-1}(\boldsymbol{\varphi})
+
\mathbf{e}_{q,\upsilon}
\right),\qquad q\in\{1,\ldots,Q_{\mathrm{net}}-1\},
\end{equation}
where $\mathbf{o}_{q}(\boldsymbol{\varphi})\in\mathbb{R}^{\nu_{q}}$ is the $q$th hidden state, $\mathbf{A}_{q,\upsilon}\in\mathbb{R}^{\nu_{q}\times \nu_{q-1}}$ is the $q$th weight matrix, $\mathbf{e}_{q,\upsilon}\in\mathbb{R}^{\nu_{q}}$ is the $q$th affine offset vector, and $\mathfrak{a}_{q}:\mathbb{R}^{\nu_{q}}\rightarrow\mathbb{R}^{\nu_{q}}$ is the $q$th activation map. The final multiclass score vector is written as
\begin{equation}
\mathbf{s}_{\Theta_{\upsilon},\upsilon}(\boldsymbol{\varphi})
=
\mathbf{A}_{Q_{\mathrm{net}},\upsilon}\mathbf{o}_{Q_{\mathrm{net}}-1}(\boldsymbol{\varphi})
+
\mathbf{e}_{Q_{\mathrm{net}},\upsilon},
\end{equation}
where $\mathbf{s}_{\Theta_{\upsilon},\upsilon}(\boldsymbol{\varphi})\in\mathbb{R}^{C}$ and
\begin{equation}
\Theta_{\upsilon}
=
\left\{
\left(\mathbf{A}_{q,\upsilon},\mathbf{e}_{q,\upsilon}\right)
:
q=1,\ldots,Q_{\mathrm{net}}
\right\},
\end{equation}
collects all trainable parameters under representation $\upsilon$.\par
To support the later generalization analysis, the admissible parameter set is restricted as
\begin{equation}
\mathfrak{P}_{\upsilon}
=
\left\{
\Theta_{\upsilon}
:
\left\|
\mathbf{A}_{q,\upsilon}
\right\|_{2}
\leq
S_{q,\upsilon},\ 
\left\|
\mathbf{A}_{q,\upsilon}
\right\|_{F}
\leq
F_{q,\upsilon},\ 
\left\|
\mathbf{e}_{q,\upsilon}
\right\|_{2}
\leq
E_{q,\upsilon},\ 
q\in\{1,\ldots,Q_{\mathrm{net}}\}
\right\},
\end{equation}
where $S_{q,\upsilon}>0$, $F_{q,\upsilon}>0$, and $E_{q,\upsilon}>0$ are deterministic layerwise bounds, $\|\cdot\|_{2}$ denotes the spectral norm for matrices and the Euclidean norm for vectors, and $\|\cdot\|_{F}$ denotes the Frobenius norm. These bounds are assumed to be compatible so that $\mathfrak{P}_{\upsilon}\neq\varnothing$; in particular, one sufficient condition is $S_{q,\upsilon}\leq F_{q,\upsilon}$ for every $q\in\{1,\ldots,Q_{\mathrm{net}}\}$, since $\|\mathbf{A}_{q,\upsilon}\|_{2}\leq\|\mathbf{A}_{q,\upsilon}\|_{F}$. The spectral bounds are used to control perturbation propagation across layers, whereas the Frobenius bounds are retained for the later complexity terms in the source-to-target generalization bound. Assume further that each activation map is zero-preserving and Lipschitz continuous, namely
\begin{equation}
\mathfrak{a}_{q}(\mathbf{0}_{\nu_{q}})=\mathbf{0}_{\nu_{q}},\qquad
\left\|
\mathfrak{a}_{q}(\mathbf{u})-\mathfrak{a}_{q}(\mathbf{u}')
\right\|_{2}
\leq
\Lambda_{q}
\left\|
\mathbf{u}-\mathbf{u}'
\right\|_{2},
\end{equation}
where $\mathbf{u}\in\mathbb{R}^{\nu_{q}}$, $\mathbf{u}'\in\mathbb{R}^{\nu_{q}}$, and $\Lambda_{q}>0$ is the activation Lipschitz constant.\par
The hidden-state envelope is then obtained recursively. Define
\begin{equation}
\Omega_{0,\upsilon}=B_{\varphi,\upsilon},\qquad
\Omega_{q,\upsilon}=
\Lambda_{q}\left(S_{q,\upsilon}\Omega_{q-1,\upsilon}+E_{q,\upsilon}\right),\qquad
q\in\{1,\ldots,Q_{\mathrm{net}}-1\},
\end{equation}
where $\Omega_{q,\upsilon}>0$ is the deterministic upper bound of the $q$th hidden-state norm. Since $\mathfrak{a}_{q}(\mathbf{0}_{\nu_{q}})=\mathbf{0}_{\nu_{q}}$, the Lipschitz property implies
\begin{equation}
\left\|
\mathfrak{a}_{q}(\mathbf{u})
\right\|_{2}
=
\left\|
\mathfrak{a}_{q}(\mathbf{u})-\mathfrak{a}_{q}(\mathbf{0}_{\nu_{q}})
\right\|_{2}
\leq
\Lambda_{q}
\left\|
\mathbf{u}
\right\|_{2},
\end{equation}
where $\mathbf{u}\in\mathbb{R}^{\nu_{q}}$ is arbitrary.\par
Therefore, for every admissible sample $(p,a,\beta,\boldsymbol{\xi},\upsilon)$ and every hidden-layer index $q\in\{1,\ldots,Q_{\mathrm{net}}-1\}$,
\begin{equation}
\left\|
\mathbf{o}_{q}\!\left(\boldsymbol{\varphi}_{p,a}^{(\beta,\boldsymbol{\xi},\upsilon)}\right)
\right\|_{2}
\leq
\Omega_{q,\upsilon},
\end{equation}
where the bound is obtained by induction from the inequality
\begin{equation}
\left\|
\mathbf{o}_{q}\!\left(\boldsymbol{\varphi}_{p,a}^{(\beta,\boldsymbol{\xi},\upsilon)}\right)
\right\|_{2}
\leq
\Lambda_{q}
\left\|
\mathbf{A}_{q,\upsilon}\mathbf{o}_{q-1}\!\left(\boldsymbol{\varphi}_{p,a}^{(\beta,\boldsymbol{\xi},\upsilon)}\right)
+
\mathbf{e}_{q,\upsilon}
\right\|_{2}
\leq
\Lambda_{q}\left(
S_{q,\upsilon}
\left\|
\mathbf{o}_{q-1}\!\left(\boldsymbol{\varphi}_{p,a}^{(\beta,\boldsymbol{\xi},\upsilon)}\right)
\right\|_{2}
+
E_{q,\upsilon}
\right).
\end{equation}
\par
The recursion admits the explicit expansion
\begin{equation}
\Omega_{q,\upsilon}
=
\left(
\prod_{r=1}^{q}\Lambda_{r}S_{r,\upsilon}
\right)
B_{\varphi,\upsilon}
+
\sum_{\kappa=1}^{q}
\left(
\prod_{r=\kappa}^{q}\Lambda_{r}
\right)
\left(
\prod_{r=\kappa+1}^{q}S_{r,\upsilon}
\right)
E_{\kappa,\upsilon},
\end{equation}
where the empty product is understood as one.\par
Consequently, for every admissible input $\boldsymbol{\varphi}\in\mathfrak{X}_{\upsilon}$, the final score magnitude is uniformly bounded as
\begin{equation}
\left\|
\mathbf{s}_{\Theta_{\upsilon},\upsilon}(\boldsymbol{\varphi})
\right\|_{2}
\leq
S_{Q_{\mathrm{net}},\upsilon}\Omega_{Q_{\mathrm{net}}-1,\upsilon}
+
E_{Q_{\mathrm{net}},\upsilon}
=:
B_{s,\upsilon},
\end{equation}
where $B_{s,\upsilon}>0$ is the score-envelope constant. Since the absolute value of each coordinate is dominated by the Euclidean norm, the classwise logits satisfy
\begin{equation}
\left|
\left[
\mathbf{s}_{\Theta_{\upsilon},\upsilon}(\boldsymbol{\varphi})
\right]_{a}
\right|
\leq
B_{s,\upsilon},\qquad \forall a\in\mathcal{A},
\end{equation}
where $\boldsymbol{\varphi}\in\mathfrak{X}_{\upsilon}$ is arbitrary. The posterior probability vector is then defined by the softmax mapping
\begin{equation}
\left[
h_{\Theta_{\upsilon},\upsilon}(\boldsymbol{\varphi})
\right]_{a}
=
\frac{
\exp\!\left(
\left[
\mathbf{s}_{\Theta_{\upsilon},\upsilon}(\boldsymbol{\varphi})
\right]_{a}
\right)
}{
\sum_{\bar{a}=1}^{C}
\exp\!\left(
\left[
\mathbf{s}_{\Theta_{\upsilon},\upsilon}(\boldsymbol{\varphi})
\right]_{\bar{a}}
\right)
},
\end{equation}
where $h_{\Theta_{\upsilon},\upsilon}(\boldsymbol{\varphi})\in\Delta_{C}$ and
\begin{equation}
\Delta_{C}
=
\left\{
\boldsymbol{\pi}\in\mathbb{R}_{+}^{C}
:
\sum_{a=1}^{C}\pi_{a}=1
\right\},
\end{equation}
is the $C$-class probability simplex. The induced TWR HAR decision rule is
\begin{equation}
\widehat{a}_{\Theta_{\upsilon},\upsilon}(\boldsymbol{\varphi})
=
\arg\max_{a\in\mathcal{A}}
\left[
h_{\Theta_{\upsilon},\upsilon}(\boldsymbol{\varphi})
\right]_{a}
=
\arg\max_{a\in\mathcal{A}}
\left[
\mathbf{s}_{\Theta_{\upsilon},\upsilon}(\boldsymbol{\varphi})
\right]_{a},
\end{equation}
where the second equality follows from the monotonicity of the exponential function. Therefore, the representation-dependent bounded-weight hypothesis class is written as
\begin{equation}
\mathcal{H}_{\upsilon}
=
\left\{
h_{\Theta_{\upsilon},\upsilon}
:
\mathbb{R}^{d_{\upsilon}}
\rightarrow
\Delta_{C},
\Theta_{\upsilon}\in\mathfrak{P}_{\upsilon}
\right\},
\end{equation}
and the overall admissible hypothesis class is
\begin{equation}
\mathcal{H}
=
\bigcup_{\upsilon\in\mathcal{V}}\mathcal{H}_{\upsilon}.
\end{equation}
\par
It remains to quantify the sensitivity of the score map to representation perturbations. For any $\boldsymbol{\varphi},\boldsymbol{\varphi}'\in\mathbb{R}^{d_{\upsilon}}$ and any hidden-layer index $q\in\{1,\ldots,Q_{\mathrm{net}}-1\}$,
\begin{equation}
\left\|
\mathbf{o}_{q}(\boldsymbol{\varphi})
-
\mathbf{o}_{q}(\boldsymbol{\varphi}')
\right\|_{2}
\leq
\Lambda_{q}
\left\|
\mathbf{A}_{q,\upsilon}
\bigl(
\mathbf{o}_{q-1}(\boldsymbol{\varphi})
-
\mathbf{o}_{q-1}(\boldsymbol{\varphi}')
\bigr)
\right\|_{2}
\leq
\Lambda_{q}S_{q,\upsilon}
\left\|
\mathbf{o}_{q-1}(\boldsymbol{\varphi})
-
\mathbf{o}_{q-1}(\boldsymbol{\varphi}')
\right\|_{2},
\end{equation}
where the affine offsets cancel. Iterating this inequality gives
\begin{equation}
\left\|
\mathbf{o}_{q}(\boldsymbol{\varphi})
-
\mathbf{o}_{q}(\boldsymbol{\varphi}')
\right\|_{2}
\leq
\left(
\prod_{r=1}^{q}\Lambda_{r}S_{r,\upsilon}
\right)
\left\|
\boldsymbol{\varphi}
-
\boldsymbol{\varphi}'
\right\|_{2}.
\end{equation}
\par
Therefore, the multiclass score map satisfies
\begin{equation}
\left\|
\mathbf{s}_{\Theta_{\upsilon},\upsilon}(\boldsymbol{\varphi})
-
\mathbf{s}_{\Theta_{\upsilon},\upsilon}(\boldsymbol{\varphi}')
\right\|_{2}
\leq
J_{\upsilon}
\left\|
\boldsymbol{\varphi}
-
\boldsymbol{\varphi}'
\right\|_{2},
\end{equation}
where
\begin{equation}
J_{\upsilon}
=
S_{Q_{\mathrm{net}},\upsilon}
\prod_{r=1}^{Q_{\mathrm{net}}-1}\Lambda_{r}S_{r,\upsilon}
=
\left(
\prod_{r=1}^{Q_{\mathrm{net}}-1}\Lambda_{r}
\right)
\left(
\prod_{r=1}^{Q_{\mathrm{net}}}S_{r,\upsilon}
\right),
\end{equation}
is the representation-dependent network Lipschitz constant, with the empty activation product interpreted as one when $Q_{\mathrm{net}}=1$. By composition with the representation map from Subsection II-C, the end-to-end signal-to-score sensitivity becomes
\begin{equation}
\left\|
\mathbf{s}_{\Theta_{\upsilon},\upsilon}\!\left(
\mathcal{R}_{\upsilon}(\widehat{\mathbf{Z}})
\right)
-
\mathbf{s}_{\Theta_{\upsilon},\upsilon}\!\left(
\mathcal{R}_{\upsilon}(\widehat{\mathbf{Z}}')
\right)
\right\|_{2}
\leq
J_{\upsilon}
\left\|
\mathcal{R}_{\upsilon}(\widehat{\mathbf{Z}})
-
\mathcal{R}_{\upsilon}(\widehat{\mathbf{Z}}')
\right\|_{2}
\leq
G_{\upsilon}
\left\|
\widehat{\mathbf{Z}}
-
\widehat{\mathbf{Z}}'
\right\|_{F},
\end{equation}
where $\widehat{\mathbf{Z}},\widehat{\mathbf{Z}}'\in\mathfrak{Z}$ are arbitrary admissible clutter-suppressed observations and $G_{\upsilon}:=J_{\upsilon}L_{\upsilon}>0$ is the end-to-end TWR representation-classifier Lipschitz constant. Hence, both the output magnitude and the perturbation sensitivity of the recognizer are reduced to explicit deterministic constants governed by the representation index $\upsilon$, the bounded layer parameters, and the activation regularity. These bounded score functions, posterior mappings, and hypothesis classes are combined with the domain, loss, and risk definitions in Subsection II-E.\par

\subsection{Domain and Risk Definitions}
To place the bounded TWR HAR recognizer into a source-to-target statistical learning framework, the physical variables established in the previous subsections are aggregated into a single domain descriptor. This formulation follows recent information-theoretic transfer analyses \cite{TransferLearningIT2024} and risk characterizations for Gibbs and Bayesian learning \cite{GibbsGeneralization2024,MinimumExcessRisk2022}. Define the admissible domain space as
\begin{equation}
\mathfrak{D}
=
\mathcal{P}_{\mathrm{sub}}
\times
\mathcal{B}_{\mathrm{view}}
\times
\mathcal{X}_{\mathrm{wall}},
\end{equation}
where $\mathcal{P}_{\mathrm{sub}}$ denotes the admissible subject-index set, $\mathcal{B}_{\mathrm{view}}$ denotes the admissible observation-view set, and $\mathcal{X}_{\mathrm{wall}}\subseteq\mathbb{R}^{d_{\xi}}$ denotes the admissible wall-parameter set. Each domain element is written as
\begin{equation}
\mathfrak{d}
=
(p,\beta,\boldsymbol{\xi})
\in
\mathfrak{D},
\end{equation}
where $p$ retains the subject-index meaning introduced in Subsection II-A, $\beta$ retains the observation-view meaning introduced in Subsection II-B, and $\boldsymbol{\xi}$ retains the wall-parameter meaning introduced in Subsection II-B. Hence, cross-person, cross-view, and cross-wall shifts are all represented by changes of $\mathfrak{d}$ along the subject, view, and wall coordinates, respectively.\par
For statistical learning, the representation vector generated under a fixed domain descriptor is regarded as a random vector, since repeated acquisitions still contain activity-phase variability, residual clutter fluctuation, and receiver-noise randomness. Therefore, for each representation index $\upsilon\in\mathcal{V}$ and each activity label $a\in\mathcal{A}$, the activity-conditional feature law is defined as
\begin{equation}
\mathbb{P}_{\mathfrak{d},\upsilon}^{(a)}
:=
\operatorname{Law}\!\left(
\boldsymbol{\varphi}_{p,a}^{(\beta,\boldsymbol{\xi},\upsilon)}
\right),
\end{equation}
where $\mathbb{P}_{\mathfrak{d},\upsilon}^{(a)}$ is a probability measure on $\mathbb{R}^{d_{\upsilon}}$. Recall from the input-radius bound in Subsection II-D that the admissible representation-input set is
\begin{equation}
\mathfrak{X}_{\upsilon}
=
\left\{
\boldsymbol{\varphi}\in\mathbb{R}^{d_{\upsilon}}
:
\|\boldsymbol{\varphi}\|_{2}\leq B_{\varphi,\upsilon}
\right\},
\end{equation}
where $B_{\varphi,\upsilon}>0$ is the representation-dependent input-radius constant. Since every admissible representation output belongs to $\mathfrak{X}_{\upsilon}$, one has
\begin{equation}
\operatorname{supp}\!\left(
\mathbb{P}_{\mathfrak{d},\upsilon}^{(a)}
\right)
\subseteq
\mathfrak{X}_{\upsilon},
\end{equation}
where $\operatorname{supp}(\cdot)$ denotes distribution support. Let
\begin{equation}
\alpha_{a}^{(\mathfrak{d})}
:=
\mathbb{P}(A=a\mid \mathfrak{d}),\qquad a\in\mathcal{A},
\end{equation}
denote the domain-dependent label prior, where $A$ is the random activity label and $\sum_{a=1}^{C}\alpha_{a}^{(\mathfrak{d})}=1$. The resulting representation-dependent joint law on feature-label pairs is denoted by $\mathbb{P}_{\mathfrak{d},\upsilon}$. Equivalently, for every bounded measurable test function $\varpi:\mathbb{R}^{d_{\upsilon}}\times\mathcal{A}\rightarrow\mathbb{R}$,
\begin{equation}
\mathbb{E}_{(\boldsymbol{\varphi},A)\sim\mathbb{P}_{\mathfrak{d},\upsilon}}
\big[
\varpi(\boldsymbol{\varphi},A)
\big]
=
\sum_{a=1}^{C}
\alpha_{a}^{(\mathfrak{d})}
\mathbb{E}_{\boldsymbol{\varphi}\sim\mathbb{P}_{\mathfrak{d},\upsilon}^{(a)}}
\big[
\varpi(\boldsymbol{\varphi},a)
\big],
\end{equation}
where $\boldsymbol{\varphi}\in\mathbb{R}^{d_{\upsilon}}$ denotes the random classifier input. Fix now one source-domain descriptor and one target-domain descriptor as
\begin{equation}
\mathfrak{d}_{\mathrm{src}}
=
\left(
p_{\mathrm{src}},
\beta_{\mathrm{src}},
\boldsymbol{\xi}_{\mathrm{src}}
\right),\qquad
\mathfrak{d}_{\mathrm{tar}}
=
\left(
p_{\mathrm{tar}},
\beta_{\mathrm{tar}},
\boldsymbol{\xi}_{\mathrm{tar}}
\right),
\end{equation}
and abbreviate
\begin{equation}
\mathbb{P}_{\mathrm{src},\upsilon}
:=
\mathbb{P}_{\mathfrak{d}_{\mathrm{src}},\upsilon},\qquad
\mathbb{P}_{\mathrm{tar},\upsilon}
:=
\mathbb{P}_{\mathfrak{d}_{\mathrm{tar}},\upsilon}.
\end{equation}
\par
This notation covers the three canonical TWR shifts in a unified manner: cross-person transfer corresponds to $p_{\mathrm{src}}\neq p_{\mathrm{tar}}$ with $\beta_{\mathrm{src}}=\beta_{\mathrm{tar}}$ and $\boldsymbol{\xi}_{\mathrm{src}}=\boldsymbol{\xi}_{\mathrm{tar}}$, cross-view transfer corresponds to $\beta_{\mathrm{src}}\neq\beta_{\mathrm{tar}}$ with $p_{\mathrm{src}}=p_{\mathrm{tar}}$ and $\boldsymbol{\xi}_{\mathrm{src}}=\boldsymbol{\xi}_{\mathrm{tar}}$, and cross-wall transfer corresponds to $\boldsymbol{\xi}_{\mathrm{src}}\neq\boldsymbol{\xi}_{\mathrm{tar}}$ with $p_{\mathrm{src}}=p_{\mathrm{tar}}$ and $\beta_{\mathrm{src}}=\beta_{\mathrm{tar}}$.\par
For a fixed representation index $\upsilon\in\mathcal{V}$, source-domain training is based on the sample
\begin{equation}
\mathcal{S}_{\mathrm{src},\upsilon}
=
\left\{
\left(
\boldsymbol{\varphi}_{i}^{\mathrm{src}},
a_{i}^{\mathrm{src}}
\right)
:
i=1,\ldots,N_{\mathrm{src}}
\right\},
\end{equation}
where $N_{\mathrm{src}}\geq 1$ is the source sample size and the pairs $\{(\boldsymbol{\varphi}_{i}^{\mathrm{src}},a_{i}^{\mathrm{src}})\}_{i=1}^{N_{\mathrm{src}}}$ are assumed to be independently and identically distributed according to $\mathbb{P}_{\mathrm{src},\upsilon}$. For every $h\in\mathcal{H}_{\upsilon}$, the multiclass cross-entropy loss is defined by
\begin{equation}
\ell_{\mathrm{ce}}
\big(
h(\boldsymbol{\varphi}),
a
\big)
=
-\log
\left[
h(\boldsymbol{\varphi})
\right]_{a},
\end{equation}
where $\boldsymbol{\varphi}\in\mathbb{R}^{d_{\upsilon}}$ and $a\in\mathcal{A}$. In particular, when $h=h_{\Theta_{\upsilon},\upsilon}\in\mathcal{H}_{\upsilon}$ is induced by the bounded-weight neural network in Subsection II-D,
\begin{equation}
\ell_{\mathrm{ce}}
\big(
h_{\Theta_{\upsilon},\upsilon}(\boldsymbol{\varphi}),
a
\big)
=
-
\left[
\mathbf{s}_{\Theta_{\upsilon},\upsilon}(\boldsymbol{\varphi})
\right]_{a}
+
\log
\left(
\sum_{\bar{a}=1}^{C}
\exp\!\left(
\left[
\mathbf{s}_{\Theta_{\upsilon},\upsilon}(\boldsymbol{\varphi})
\right]_{\bar{a}}
\right)
\right),
\end{equation}
where $\mathbf{s}_{\Theta_{\upsilon},\upsilon}(\boldsymbol{\varphi})\in\mathbb{R}^{C}$ is the score vector defined in Subsection II-D. For compactness, define the score-domain loss
\begin{equation}
\ell_{\mathrm{ce}}(\mathbf{c},a)
:=
-c_{a}
+
\operatorname{lse}(\mathbf{c}),\qquad
\operatorname{lse}(\mathbf{c})
:=
\log\!\left(
\sum_{\bar{a}=1}^{C}e^{c_{\bar{a}}}
\right),
\end{equation}
where $\mathbf{c}=[c_{1},\ldots,c_{C}]^{\top}\in\mathbb{R}^{C}$ is a generic score vector.\par
The bounded-score property established in Subsection II-D immediately implies that the source and target losses are uniformly bounded. Indeed, if $\mathbf{c}=\mathbf{s}_{\Theta_{\upsilon},\upsilon}(\boldsymbol{\varphi})$ for some admissible $h_{\Theta_{\upsilon},\upsilon}\in\mathcal{H}_{\upsilon}$ and some $\boldsymbol{\varphi}\in\mathfrak{X}_{\upsilon}$, then $|c_{\bar{a}}|\leq B_{s,\upsilon}$ for every $\bar{a}\in\mathcal{A}$. Hence,
\begin{equation}
\begin{aligned}
\ell_{\mathrm{ce}}(\mathbf{c},a)
=
\log\!\left(
\frac{\sum_{\bar{a}=1}^{C}e^{c_{\bar{a}}}}{e^{c_{a}}}
\right)
=
\log\!\left(
1+\sum_{\bar{a}\neq a}e^{c_{\bar{a}}-c_{a}}
\right)
\leq
\log\!\left(
1+(C-1)e^{2B_{s,\upsilon}}
\right)
=:
B_{\ell,\upsilon},
\end{aligned}
\end{equation}
where $B_{\ell,\upsilon}>0$ depends only on the class number $C$ and the score-envelope constant $B_{s,\upsilon}$. Since $\sum_{\bar{a}=1}^{C}e^{c_{\bar{a}}}\geq e^{c_{a}}$, the same identity also yields
\begin{equation}
0
\leq
\ell_{\mathrm{ce}}(\mathbf{c},a)
\leq
B_{\ell,\upsilon}.
\end{equation}
\par
The same loss is also Lipschitz continuous with respect to the score vector. Since
\begin{equation}
\nabla\operatorname{lse}(\mathbf{c})
=
\begin{bmatrix}
\dfrac{e^{c_{1}}}{\sum_{\bar{a}=1}^{C}e^{c_{\bar{a}}}} &
\cdots &
\dfrac{e^{c_{C}}}{\sum_{\bar{a}=1}^{C}e^{c_{\bar{a}}}}
\end{bmatrix}^{\top}
\in
\Delta_{C},
\end{equation}
the mean-value theorem implies that, for any $\mathbf{c},\mathbf{c}'\in\mathbb{R}^{C}$, there exists a vector $\widetilde{\mathbf{c}}$ on the line segment joining $\mathbf{c}$ and $\mathbf{c}'$ such that
\begin{equation}
\begin{aligned}
\big|
\operatorname{lse}(\mathbf{c})
-
\operatorname{lse}(\mathbf{c}')
\big|
=
\big|
\nabla\operatorname{lse}(\widetilde{\mathbf{c}})^{\top}
(\mathbf{c}-\mathbf{c}')
\big|
\leq
\left\|
\nabla\operatorname{lse}(\widetilde{\mathbf{c}})
\right\|_{2}
\left\|
\mathbf{c}-\mathbf{c}'
\right\|_{2}
\leq
\left\|
\nabla\operatorname{lse}(\widetilde{\mathbf{c}})
\right\|_{1}
\left\|
\mathbf{c}-\mathbf{c}'
\right\|_{2}
=
\left\|
\mathbf{c}-\mathbf{c}'
\right\|_{2},
\end{aligned}
\end{equation}
where $\|\nabla\operatorname{lse}(\widetilde{\mathbf{c}})\|_{1}=1$ follows from $\nabla\operatorname{lse}(\widetilde{\mathbf{c}})\in\Delta_{C}$. Therefore,
\begin{equation}
\big|
\ell_{\mathrm{ce}}(\mathbf{c},a)
-
\ell_{\mathrm{ce}}(\mathbf{c}',a)
\big|
\leq
\left\|
\mathbf{c}-\mathbf{c}'
\right\|_{2}
+
\left\|
\mathbf{c}-\mathbf{c}'
\right\|_{2}
=
2
\left\|
\mathbf{c}-\mathbf{c}'
\right\|_{2},
\end{equation}
where the inequality $|c_{a}-c_{a}'|\leq\|\mathbf{c}-\mathbf{c}'\|_{2}$ has been used. Combining this result with the network Lipschitz constant $J_{\upsilon}$ from Subsection II-D gives
\begin{equation}
\big|
\ell_{\mathrm{ce}}
\big(
h_{\Theta_{\upsilon},\upsilon}(\boldsymbol{\varphi}),
a
\big)
-
\ell_{\mathrm{ce}}
\big(
h_{\Theta_{\upsilon},\upsilon}(\boldsymbol{\varphi}'),
a
\big)
\big|
\leq
2J_{\upsilon}
\left\|
\boldsymbol{\varphi}
-
\boldsymbol{\varphi}'
\right\|_{2},
\end{equation}
where $\boldsymbol{\varphi},\boldsymbol{\varphi}'\in\mathbb{R}^{d_{\upsilon}}$ are arbitrary. By further composition with the representation map, one also has
\begin{equation}
\big|
\ell_{\mathrm{ce}}
\big(
h_{\Theta_{\upsilon},\upsilon}(\mathcal{R}_{\upsilon}(\widehat{\mathbf{Z}})),
a
\big)
-
\ell_{\mathrm{ce}}
\big(
h_{\Theta_{\upsilon},\upsilon}(\mathcal{R}_{\upsilon}(\widehat{\mathbf{Z}}')),
a
\big)
\big|
\leq
2G_{\upsilon}
\left\|
\widehat{\mathbf{Z}}
-
\widehat{\mathbf{Z}}'
\right\|_{F},
\end{equation}
where $\widehat{\mathbf{Z}},\widehat{\mathbf{Z}}'\in\mathfrak{Z}$ are admissible clutter-suppressed observations and $G_{\upsilon}=J_{\upsilon}L_{\upsilon}$ was defined in Subsection II-D. Consequently, the TWR learning loss inherits both a uniform envelope and an explicit end-to-end perturbation sensitivity from the bounded representation-classifier pipeline.\par
The source-domain population risk, target-domain population risk, and source-domain empirical risk are now defined, respectively, by
\begin{equation}
\mathcal{L}_{\mathrm{src},\upsilon}(h)
:=
\mathbb{E}_{(\boldsymbol{\varphi}^{\mathrm{src}},A^{\mathrm{src}})\sim\mathbb{P}_{\mathrm{src},\upsilon}}
\big[
\ell_{\mathrm{ce}}
\big(
h(\boldsymbol{\varphi}^{\mathrm{src}}),
A^{\mathrm{src}}
\big)
\big],
\end{equation}
\begin{equation}
\mathcal{L}_{\mathrm{tar},\upsilon}(h)
:=
\mathbb{E}_{(\boldsymbol{\varphi}^{\mathrm{tar}},A^{\mathrm{tar}})\sim\mathbb{P}_{\mathrm{tar},\upsilon}}
\big[
\ell_{\mathrm{ce}}
\big(
h(\boldsymbol{\varphi}^{\mathrm{tar}}),
A^{\mathrm{tar}}
\big)
\big],
\end{equation}
and
\begin{equation}
\widehat{\mathcal{L}}_{\mathrm{src},\upsilon}(h)
:=
\frac{1}{N_{\mathrm{src}}}
\sum_{i=1}^{N_{\mathrm{src}}}
\ell_{\mathrm{ce}}
\big(
h(\boldsymbol{\varphi}_{i}^{\mathrm{src}}),
a_{i}^{\mathrm{src}}
\big),
\end{equation}
where $h\in\mathcal{H}_{\upsilon}$ is arbitrary. The corresponding source empirical risk minimizer is defined as
\begin{equation}
\widehat{h}_{\mathrm{src},\upsilon}
\in
\arg\min_{h\in\mathcal{H}_{\upsilon}}
\widehat{\mathcal{L}}_{\mathrm{src},\upsilon}(h).
\end{equation}
Since $\mathfrak{P}_{\upsilon}$ is a closed and bounded subset of a finite-dimensional Euclidean space and the mapping $\Theta_{\upsilon}\mapsto \widehat{\mathcal{L}}_{\mathrm{src},\upsilon}(h_{\Theta_{\upsilon},\upsilon})$ is continuous, the above argmin set is nonempty.\par
The central source-to-target generalization quantity is the domain-dependent risk gap
\begin{equation}
\varepsilon_{\mathrm{gen},\upsilon}(h)
:=
\mathcal{L}_{\mathrm{tar},\upsilon}(h)
-
\widehat{\mathcal{L}}_{\mathrm{src},\upsilon}(h),
\end{equation}
which admits the exact decomposition
\begin{equation}
\varepsilon_{\mathrm{gen},\upsilon}(h)
=
\Bigl[
\mathcal{L}_{\mathrm{src},\upsilon}(h)
-
\widehat{\mathcal{L}}_{\mathrm{src},\upsilon}(h)
\Bigr]
+
\Bigl[
\mathcal{L}_{\mathrm{tar},\upsilon}(h)
-
\mathcal{L}_{\mathrm{src},\upsilon}(h)
\Bigr].
\end{equation}
\par
The first bracket measures source-domain estimation error, whereas the second bracket measures source-to-target domain shift under the fixed representation index $\upsilon$.\par
To separate the irreducible target-domain difficulty from the training-induced error, define the optimal target risk within $\mathcal{H}_{\upsilon}$ as
\begin{equation}
\mathcal{L}_{\mathrm{tar},\upsilon}^{\star}
:=
\inf_{h\in\mathcal{H}_{\upsilon}}
\mathcal{L}_{\mathrm{tar},\upsilon}(h),
\end{equation}
and the excess target risk as
\begin{equation}
\mathcal{E}_{\mathrm{tar},\upsilon}(h)
:=
\mathcal{L}_{\mathrm{tar},\upsilon}(h)
-
\mathcal{L}_{\mathrm{tar},\upsilon}^{\star}.
\end{equation}
\par
In particular, for the source-trained hypothesis $\widehat{h}_{\mathrm{src},\upsilon}$,
\begin{equation}
\mathcal{E}_{\mathrm{tar},\upsilon}
\big(
\widehat{h}_{\mathrm{src},\upsilon}
\big)
=
\widehat{\mathcal{L}}_{\mathrm{src},\upsilon}
\big(
\widehat{h}_{\mathrm{src},\upsilon}
\big)
+
\Bigl[
\mathcal{L}_{\mathrm{src},\upsilon}
\big(
\widehat{h}_{\mathrm{src},\upsilon}
\big)
-
\widehat{\mathcal{L}}_{\mathrm{src},\upsilon}
\big(
\widehat{h}_{\mathrm{src},\upsilon}
\big)
\Bigr]
+
\Bigl[
\mathcal{L}_{\mathrm{tar},\upsilon}
\big(
\widehat{h}_{\mathrm{src},\upsilon}
\big)
-
\mathcal{L}_{\mathrm{src},\upsilon}
\big(
\widehat{h}_{\mathrm{src},\upsilon}
\big)
\Bigr]
-
\mathcal{L}_{\mathrm{tar},\upsilon}^{\star},
\end{equation}
so that the later theory can control the source empirical fit term, the estimation term, the structured domain-shift term, and the target optimal term within one unified expression. Although the present formulation is written for a single source domain, the same definitions extend directly to the multi-source case by replacing $\widehat{\mathcal{L}}_{\mathrm{src},\upsilon}$ with a weighted sum of source empirical risks, which is the setting treated later in Subsection III-D.\par

\section{Generalization Error Bounds under Structured TWR Domain Shifts}

\subsection{Unified Source-to-Target Generalization Bound}
Fix one representation index $\upsilon\in\mathcal{V}$. The bounded-weight hypothesis class $\mathcal{H}_{\upsilon}$, the source-domain and target-domain laws $\mathbb{P}_{\mathrm{src},\upsilon}$ and $\mathbb{P}_{\mathrm{tar},\upsilon}$, the empirical risk $\widehat{\mathcal{L}}_{\mathrm{src},\upsilon}$, the population risks $\mathcal{L}_{\mathrm{src},\upsilon}$ and $\mathcal{L}_{\mathrm{tar},\upsilon}$, and the source empirical risk minimizer $\widehat{h}_{\mathrm{src},\upsilon}$ retain the definitions introduced in Subsections II-D and II-E. The present derivation is aligned with fast-rate and excess-risk analyses \cite{FastRateIT2025,MinimumExcessRisk2022,MinimaxExcessRisk2023}, transfer and meta-learning bounds \cite{TransferMetaLearning2022,TransferLearningIT2024}, and bounded-network information-theoretic generalization results \cite{LossSurfaceIR2023,GibbsGeneralization2024,InfoGenDNN2025}. For compactness, define the loss induced by $h\in\mathcal{H}_{\upsilon}$ as
\begin{equation}
\ell_{h}(\boldsymbol{\varphi},a)
:=
\ell_{\mathrm{ce}}\big(h(\boldsymbol{\varphi}),a\big),
\end{equation}
where $(\boldsymbol{\varphi},a)\in\mathfrak{X}_{\upsilon}\times\mathcal{A}$. Recall from Subsection II-E that
\begin{equation}
0
\leq
\ell_{h}(\boldsymbol{\varphi},a)
\leq
B_{\ell,\upsilon},
\end{equation}
and
\begin{equation}
\big|
\ell_{h}(\boldsymbol{\varphi},a)
-
\ell_{h}(\boldsymbol{\varphi}',a)
\big|
\leq
2J_{\upsilon}
\left\|
\boldsymbol{\varphi}
-
\boldsymbol{\varphi}'
\right\|_{2},
\end{equation}
where $B_{\ell,\upsilon}=\log(1+(C-1)e^{2B_{s,\upsilon}})$ is the cross-entropy envelope constant and $J_{\upsilon}>0$ is the network Lipschitz constant. Hence the classwise source-to-target analysis can be carried out directly on the feature-label product space $\mathfrak{X}_{\upsilon}\times\mathcal{A}$.\par
Let
\begin{equation}
\mathcal{F}_{\upsilon}
:=
\left\{
(\boldsymbol{\varphi},a)\mapsto \ell_{h}(\boldsymbol{\varphi},a)
:
h\in\mathcal{H}_{\upsilon}
\right\},
\end{equation}
denote the loss class induced by $\mathcal{H}_{\upsilon}$. For the source sample
\begin{equation}
\mathcal{S}_{\mathrm{src},\upsilon}
=
\left\{
\left(
\boldsymbol{\varphi}_{i}^{\mathrm{src}},
a_{i}^{\mathrm{src}}
\right)
:
i=1,\ldots,N_{\mathrm{src}}
\right\},
\end{equation}
whose elements are independently and identically distributed according to $\mathbb{P}_{\mathrm{src},\upsilon}$, the source Rademacher complexity is defined by
\begin{equation}
\mathfrak{R}_{N_{\mathrm{src}}}(\mathcal{F}_{\upsilon})
:=
\mathbb{E}_{\mathcal{S}_{\mathrm{src},\upsilon},\boldsymbol{\varepsilon}}
\left[
\sup_{h\in\mathcal{H}_{\upsilon}}
\frac{1}{N_{\mathrm{src}}}
\sum_{i=1}^{N_{\mathrm{src}}}
\varepsilon_{i}
\ell_{h}
\big(
\boldsymbol{\varphi}_{i}^{\mathrm{src}},
a_{i}^{\mathrm{src}}
\big)
\right],
\end{equation}
where $\boldsymbol{\varepsilon}=[\varepsilon_{1},\ldots,\varepsilon_{N_{\mathrm{src}}}]^{\top}$ is a vector of independent Rademacher variables satisfying $\mathbb{P}(\varepsilon_{i}=1)=\mathbb{P}(\varepsilon_{i}=-1)=1/2$. To isolate the irreducible mismatch between the source and target domains within the restricted hypothesis class, define
\begin{equation}
\lambda_{\upsilon}^{\star}
:=
\inf_{h\in\mathcal{H}_{\upsilon}}
\left(
\mathcal{L}_{\mathrm{src},\upsilon}(h)
+
\mathcal{L}_{\mathrm{tar},\upsilon}(h)
\right),
\end{equation}
where $\lambda_{\upsilon}^{\star}\geq 0$ is the joint source-target approximation term.\par
The source-to-target discrepancy is measured through pairwise loss oscillations. For arbitrary probability measures $\mathbb{P}$ and $\mathbb{Q}$ on $\mathfrak{X}_{\upsilon}\times\mathcal{A}$, define
\begin{equation}
\operatorname{disc}_{\upsilon}^{\mathrm{pair}}(\mathbb{P},\mathbb{Q})
:=
\sup_{h,h'\in\mathcal{H}_{\upsilon}}
\left|
\mathbb{E}_{(\boldsymbol{\varphi},a)\sim\mathbb{P}}
\left[
\left|
\ell_{h}(\boldsymbol{\varphi},a)
-
\ell_{h'}(\boldsymbol{\varphi},a)
\right|
\right]
-
\mathbb{E}_{(\boldsymbol{\varphi},a)\sim\mathbb{Q}}
\left[
\left|
\ell_{h}(\boldsymbol{\varphi},a)
-
\ell_{h'}(\boldsymbol{\varphi},a)
\right|
\right]
\right|.
\end{equation}
\par
To connect this discrepancy to the structured TWR geometry, introduce the transport cost
\begin{equation}
\mathfrak{c}_{\upsilon}
\big(
(\boldsymbol{\varphi},a),
(\boldsymbol{\varphi}',a')
\big)
:=
4J_{\upsilon}
\left\|
\boldsymbol{\varphi}
-
\boldsymbol{\varphi}'
\right\|_{2}
+
2B_{\ell,\upsilon}\mathbf{1}\{a\neq a'\},
\end{equation}
where $\mathbf{1}\{\cdot\}$ denotes the indicator function. The corresponding $1$-Wasserstein distance is
\begin{equation}
W_{1,\mathfrak{c}_{\upsilon}}(\mathbb{P},\mathbb{Q})
:=
\inf_{\pi\in\Pi(\mathbb{P},\mathbb{Q})}
\mathbb{E}_{\big((\boldsymbol{\varphi},a),(\boldsymbol{\varphi}',a')\big)\sim\pi}
\left[
\mathfrak{c}_{\upsilon}
\big(
(\boldsymbol{\varphi},a),
(\boldsymbol{\varphi}',a')
\big)
\right],
\end{equation}
where $\Pi(\mathbb{P},\mathbb{Q})$ denotes the set of all couplings whose first marginal is $\mathbb{P}$ and whose second marginal is $\mathbb{Q}$. The next lemma shows that the loss discrepancy is dominated by this transport distance.\par
\begin{lemma}\label{lem:pairwise_transport}
For every $h,h'\in\mathcal{H}_{\upsilon}$, the pairwise-loss function
\begin{equation}
\Delta_{h,h'}(\boldsymbol{\varphi},a)
:=
\left|
\ell_{h}(\boldsymbol{\varphi},a)
-
\ell_{h'}(\boldsymbol{\varphi},a)
\right|,
\end{equation}
is $1$-Lipschitz with respect to the cost $\mathfrak{c}_{\upsilon}$ on $\mathfrak{X}_{\upsilon}\times\mathcal{A}$, namely
\begin{equation}
\big|
\Delta_{h,h'}(\boldsymbol{\varphi},a)
-
\Delta_{h,h'}(\boldsymbol{\varphi}',a')
\big|
\leq
\mathfrak{c}_{\upsilon}
\big(
(\boldsymbol{\varphi},a),
(\boldsymbol{\varphi}',a')
\big).
\end{equation}
\par
Consequently,
\begin{equation}
\operatorname{disc}_{\upsilon}^{\mathrm{pair}}
\big(
\mathbb{P}_{\mathrm{src},\upsilon},
\mathbb{P}_{\mathrm{tar},\upsilon}
\big)
\leq
W_{1,\mathfrak{c}_{\upsilon}}
\big(
\mathbb{P}_{\mathrm{src},\upsilon},
\mathbb{P}_{\mathrm{tar},\upsilon}
\big).
\end{equation}
\end{lemma}
\begin{proof}
Let
\begin{equation}
(\boldsymbol{\varphi},a),(\boldsymbol{\varphi}',a')
\in
\mathfrak{X}_{\upsilon}\times\mathcal{A},
\end{equation}
and introduce an intermediate comparison while the feature arguments are held fixed. By the reverse triangle inequality,
\begin{equation}
\begin{aligned}
\big|
\Delta_{h,h'}(\boldsymbol{\varphi},a)
-
\Delta_{h,h'}(\boldsymbol{\varphi}',a')
\big|
&\leq
\big|
\ell_{h}(\boldsymbol{\varphi},a)
-
\ell_{h}(\boldsymbol{\varphi}',a')
\big|
+
\big|
\ell_{h'}(\boldsymbol{\varphi},a)
-
\ell_{h'}(\boldsymbol{\varphi}',a')
\big|.
\end{aligned}
\end{equation}
For the first term, one has
\begin{equation}
\begin{aligned}
\big|
\ell_{h}(\boldsymbol{\varphi},a)
-
\ell_{h}(\boldsymbol{\varphi}',a')
\big|
&\leq
\big|
\ell_{h}(\boldsymbol{\varphi},a)
-
\ell_{h}(\boldsymbol{\varphi},a')
\big|
+
\big|
\ell_{h}(\boldsymbol{\varphi},a')
-
\ell_{h}(\boldsymbol{\varphi}',a')
\big|\\
&\leq
B_{\ell,\upsilon}\mathbf{1}\{a\neq a'\}
+
2J_{\upsilon}
\left\|
\boldsymbol{\varphi}
-
\boldsymbol{\varphi}'
\right\|_{2},
\end{aligned}
\end{equation}
where the first addend follows from $0\leq\ell_{h}\leq B_{\ell,\upsilon}$ and vanishes when $a=a'$, whereas the second addend follows from the feature-Lipschitz property of the loss. The same argument yields
\begin{equation}
\big|
\ell_{h'}(\boldsymbol{\varphi},a)
-
\ell_{h'}(\boldsymbol{\varphi}',a')
\big|
\leq
B_{\ell,\upsilon}\mathbf{1}\{a\neq a'\}
+
2J_{\upsilon}
\left\|
\boldsymbol{\varphi}
-
\boldsymbol{\varphi}'
\right\|_{2}.
\end{equation}
\par
Therefore,
\begin{equation}
\big|
\Delta_{h,h'}(\boldsymbol{\varphi},a)
-
\Delta_{h,h'}(\boldsymbol{\varphi}',a')
\big|
\leq
4J_{\upsilon}
\left\|
\boldsymbol{\varphi}
-
\boldsymbol{\varphi}'
\right\|_{2}
+
2B_{\ell,\upsilon}\mathbf{1}\{a\neq a'\}
=
\mathfrak{c}_{\upsilon}
\big(
(\boldsymbol{\varphi},a),
(\boldsymbol{\varphi}',a')
\big).
\end{equation}
\par
Hence, for any coupling $\pi\in\Pi(\mathbb{P}_{\mathrm{src},\upsilon},\mathbb{P}_{\mathrm{tar},\upsilon})$,
\begin{equation}
\begin{aligned}
\left|
\mathbb{E}_{\mathbb{P}_{\mathrm{src},\upsilon}}
\big[
\Delta_{h,h'}(\boldsymbol{\varphi},a)
\big]
-
\mathbb{E}_{\mathbb{P}_{\mathrm{tar},\upsilon}}
\big[
\Delta_{h,h'}(\boldsymbol{\varphi},a)
\big]
\right|
&=
\left|
\mathbb{E}_{\pi}
\big[
\Delta_{h,h'}(\boldsymbol{\varphi},a)
-
\Delta_{h,h'}(\boldsymbol{\varphi}',a')
\big]
\right|\\
&\leq
\mathbb{E}_{\pi}
\left[
\left|
\Delta_{h,h'}(\boldsymbol{\varphi},a)
-
\Delta_{h,h'}(\boldsymbol{\varphi}',a')
\right|
\right]\\
&\leq
\mathbb{E}_{\pi}
\left[
\mathfrak{c}_{\upsilon}
\big(
(\boldsymbol{\varphi},a),
(\boldsymbol{\varphi}',a')
\big)
\right].
\end{aligned}
\end{equation}
\par
Taking the infimum over all such couplings yields, for every $h,h'\in\mathcal{H}_{\upsilon}$,
\begin{equation}
\left|
\mathbb{E}_{\mathbb{P}_{\mathrm{src},\upsilon}}
\big[
\Delta_{h,h'}(\boldsymbol{\varphi},a)
\big]
-
\mathbb{E}_{\mathbb{P}_{\mathrm{tar},\upsilon}}
\big[
\Delta_{h,h'}(\boldsymbol{\varphi},a)
\big]
\right|
\leq
W_{1,\mathfrak{c}_{\upsilon}}
\big(
\mathbb{P}_{\mathrm{src},\upsilon},
\mathbb{P}_{\mathrm{tar},\upsilon}
\big).
\end{equation}
\par
Taking the supremum over $h,h'\in\mathcal{H}_{\upsilon}$ proves the desired discrepancy bound.\par
\end{proof}
The bounded-weight network also yields a norm-controlled complexity specialization. Let $Q_{\mathrm{net}}$ denote the number of affine layers introduced in Subsection II-D. By the $2$-Lipschitz continuity of $\ell_{\mathrm{ce}}(\mathbf{c},a)$ with respect to the score vector $\mathbf{c}$, by vector contraction, and by standard norm-based Rademacher estimates for bounded-input feedforward networks under the spectral/Frobenius constraints in $\mathfrak{P}_{\upsilon}$, there exists a finite deterministic constant $\Psi_{\upsilon}^{\mathrm{cmp}}>0$ such that
\begin{equation}
\mathfrak{R}_{N_{\mathrm{src}}}(\mathcal{F}_{\upsilon})
\leq
\frac{\Psi_{\upsilon}^{\mathrm{cmp}}}{\sqrt{N_{\mathrm{src}}}},
\label{eq:complexity_specialization}
\end{equation}
where $\Psi_{\upsilon}^{\mathrm{cmp}}$ depends only on the class number $C$, the input-radius constant $B_{\varphi,\upsilon}$, the activation Lipschitz constants $\{\Lambda_{r}\}_{r=1}^{Q_{\mathrm{net}}-1}$, and the layerwise norm bounds $\{S_{r,\upsilon},F_{r,\upsilon}\}_{r=1}^{Q_{\mathrm{net}}}$. The next theorem combines empirical fit, complexity, structured domain shift, and joint approximation into one source-to-target target risk bound.\par
\begin{theorem}\label{thm:unified_bound}
Let $\delta\in(0,1)$. Then, with probability at least $1-\delta$ over the draw of $\mathcal{S}_{\mathrm{src},\upsilon}$, the inequality
\begin{equation}
\mathcal{L}_{\mathrm{tar},\upsilon}(h)
\leq
\widehat{\mathcal{L}}_{\mathrm{src},\upsilon}(h)
+
2\mathfrak{R}_{N_{\mathrm{src}}}(\mathcal{F}_{\upsilon})
+
B_{\ell,\upsilon}\sqrt{\frac{\log(1/\delta)}{2N_{\mathrm{src}}}}
+
\operatorname{disc}_{\upsilon}^{\mathrm{pair}}
\big(
\mathbb{P}_{\mathrm{src},\upsilon},
\mathbb{P}_{\mathrm{tar},\upsilon}
\big)
+
\lambda_{\upsilon}^{\star},
\label{eq:unified_target_risk_bound}
\end{equation}
holds simultaneously for all $h\in\mathcal{H}_{\upsilon}$. In particular,
\begin{equation}
\mathcal{L}_{\mathrm{tar},\upsilon}(h)
\leq
\widehat{\mathcal{L}}_{\mathrm{src},\upsilon}(h)
+
\frac{2\Psi_{\upsilon}^{\mathrm{cmp}}}{\sqrt{N_{\mathrm{src}}}}
+
B_{\ell,\upsilon}\sqrt{\frac{\log(1/\delta)}{2N_{\mathrm{src}}}}
+
\operatorname{disc}_{\upsilon}^{\mathrm{pair}}
\big(
\mathbb{P}_{\mathrm{src},\upsilon},
\mathbb{P}_{\mathrm{tar},\upsilon}
\big)
+
\lambda_{\upsilon}^{\star},
\label{eq:unified_target_risk_bound_specialized}
\end{equation}
where \eqref{eq:unified_target_risk_bound_specialized} follows by substituting \eqref{eq:complexity_specialization} into \eqref{eq:unified_target_risk_bound}.\par
\end{theorem}
\begin{proof}
Define the source uniform deviation functional
\begin{equation}
\Xi(\mathcal{S}_{\mathrm{src},\upsilon})
:=
\sup_{h\in\mathcal{H}_{\upsilon}}
\left(
\mathcal{L}_{\mathrm{src},\upsilon}(h)
-
\widehat{\mathcal{L}}_{\mathrm{src},\upsilon}(h)
\right).
\end{equation}
\par
If the $i$th sample in $\mathcal{S}_{\mathrm{src},\upsilon}$ is replaced by an independent copy, the empirical average associated with every $h\in\mathcal{H}_{\upsilon}$ changes by at most $B_{\ell,\upsilon}/N_{\mathrm{src}}$, because the loss is bounded between $0$ and $B_{\ell,\upsilon}$. Therefore, McDiarmid's inequality gives
\begin{equation}
\Xi(\mathcal{S}_{\mathrm{src},\upsilon})
\leq
\mathbb{E}\big[\Xi(\mathcal{S}_{\mathrm{src},\upsilon})\big]
+
B_{\ell,\upsilon}
\sqrt{
\frac{\log(1/\delta)}{2N_{\mathrm{src}}}
},
\end{equation}
with probability at least $1-\delta$. It remains to bound the expectation. Let $\widetilde{\mathcal{S}}_{\mathrm{src},\upsilon}=\{(\widetilde{\boldsymbol{\varphi}}_{i}^{\mathrm{src}},\widetilde{a}_{i}^{\mathrm{src}})\}_{i=1}^{N_{\mathrm{src}}}$ be an independent ghost sample drawn from the same source law, and let the inner expectation in the next display be taken with respect to $\widetilde{\mathcal{S}}_{\mathrm{src},\upsilon}$. Then
\begin{equation}
\begin{aligned}
\mathbb{E}\big[\Xi(\mathcal{S}_{\mathrm{src},\upsilon})\big]
&=
\mathbb{E}
\left[
\sup_{h\in\mathcal{H}_{\upsilon}}
\frac{1}{N_{\mathrm{src}}}
\sum_{i=1}^{N_{\mathrm{src}}}
\left(
\mathbb{E}\big[\ell_{h}(\widetilde{\boldsymbol{\varphi}}_{i}^{\mathrm{src}},\widetilde{a}_{i}^{\mathrm{src}})\big]
-
\ell_{h}(\boldsymbol{\varphi}_{i}^{\mathrm{src}},a_{i}^{\mathrm{src}})
\right)
\right]\\
&\leq
\mathbb{E}
\left[
\sup_{h\in\mathcal{H}_{\upsilon}}
\frac{1}{N_{\mathrm{src}}}
\sum_{i=1}^{N_{\mathrm{src}}}
\left(
\ell_{h}(\widetilde{\boldsymbol{\varphi}}_{i}^{\mathrm{src}},\widetilde{a}_{i}^{\mathrm{src}})
-
\ell_{h}(\boldsymbol{\varphi}_{i}^{\mathrm{src}},a_{i}^{\mathrm{src}})
\right)
\right]\\
&=
\mathbb{E}
\left[
\sup_{h\in\mathcal{H}_{\upsilon}}
\frac{1}{N_{\mathrm{src}}}
\sum_{i=1}^{N_{\mathrm{src}}}
\varepsilon_{i}
\left(
\ell_{h}(\widetilde{\boldsymbol{\varphi}}_{i}^{\mathrm{src}},\widetilde{a}_{i}^{\mathrm{src}})
-
\ell_{h}(\boldsymbol{\varphi}_{i}^{\mathrm{src}},a_{i}^{\mathrm{src}})
\right)
\right]\\
&\leq
\mathbb{E}
\left[
\sup_{h\in\mathcal{H}_{\upsilon}}
\frac{1}{N_{\mathrm{src}}}
\sum_{i=1}^{N_{\mathrm{src}}}
\varepsilon_{i}
\ell_{h}(\widetilde{\boldsymbol{\varphi}}_{i}^{\mathrm{src}},\widetilde{a}_{i}^{\mathrm{src}})
\right]
+
\mathbb{E}
\left[
\sup_{h\in\mathcal{H}_{\upsilon}}
\frac{1}{N_{\mathrm{src}}}
\sum_{i=1}^{N_{\mathrm{src}}}
\left(
-\varepsilon_{i}
\right)
\ell_{h}(\boldsymbol{\varphi}_{i}^{\mathrm{src}},a_{i}^{\mathrm{src}})
\right]\\
&=
2\mathfrak{R}_{N_{\mathrm{src}}}(\mathcal{F}_{\upsilon}),
\end{aligned}
\end{equation}
where the third line uses the symmetry of the Rademacher variables and the last line uses the fact that $\varepsilon_{i}$ and $-\varepsilon_{i}$ have the same distribution. Consequently, with probability at least $1-\delta$,
\begin{equation}
\sup_{h\in\mathcal{H}_{\upsilon}}
\left(
\mathcal{L}_{\mathrm{src},\upsilon}(h)
-
\widehat{\mathcal{L}}_{\mathrm{src},\upsilon}(h)
\right)
\leq
2\mathfrak{R}_{N_{\mathrm{src}}}(\mathcal{F}_{\upsilon})
+
B_{\ell,\upsilon}
\sqrt{
\frac{\log(1/\delta)}{2N_{\mathrm{src}}}
}.
\end{equation}
\par
Now fix $\epsilon>0$. By the definition of $\lambda_{\upsilon}^{\star}$, there exists $h_{\epsilon}\in\mathcal{H}_{\upsilon}$ such that
\begin{equation}
\mathcal{L}_{\mathrm{src},\upsilon}(h_{\epsilon})
+
\mathcal{L}_{\mathrm{tar},\upsilon}(h_{\epsilon})
\leq
\lambda_{\upsilon}^{\star}
+
\epsilon.
\end{equation}
\par
For an arbitrary $h\in\mathcal{H}_{\upsilon}$,
\begin{equation}
\begin{aligned}
\mathcal{L}_{\mathrm{tar},\upsilon}(h)
&=
\mathbb{E}_{\mathbb{P}_{\mathrm{tar},\upsilon}}
\big[
\ell_{h}(\boldsymbol{\varphi},a)
\big]
\leq
\mathbb{E}_{\mathbb{P}_{\mathrm{tar},\upsilon}}
\big[
\ell_{h_{\epsilon}}(\boldsymbol{\varphi},a)
\big]
+
\mathbb{E}_{\mathbb{P}_{\mathrm{tar},\upsilon}}
\big[
\big|
\ell_{h}(\boldsymbol{\varphi},a)
-
\ell_{h_{\epsilon}}(\boldsymbol{\varphi},a)
\big|
\big]\\
&\leq
\mathcal{L}_{\mathrm{tar},\upsilon}(h_{\epsilon})
+
\mathbb{E}_{\mathbb{P}_{\mathrm{src},\upsilon}}
\big[
\big|
\ell_{h}(\boldsymbol{\varphi},a)
-
\ell_{h_{\epsilon}}(\boldsymbol{\varphi},a)
\big|
\big]
+
\operatorname{disc}_{\upsilon}^{\mathrm{pair}}
\big(
\mathbb{P}_{\mathrm{src},\upsilon},
\mathbb{P}_{\mathrm{tar},\upsilon}
\big)\\
&\leq
\mathcal{L}_{\mathrm{tar},\upsilon}(h_{\epsilon})
+
\mathcal{L}_{\mathrm{src},\upsilon}(h)
+
\mathcal{L}_{\mathrm{src},\upsilon}(h_{\epsilon})
+
\operatorname{disc}_{\upsilon}^{\mathrm{pair}}
\big(
\mathbb{P}_{\mathrm{src},\upsilon},
\mathbb{P}_{\mathrm{tar},\upsilon}
\big),
\end{aligned}
\end{equation}
where the second inequality follows from the definition of $\operatorname{disc}_{\upsilon}^{\mathrm{pair}}$ and the third inequality uses $|u-v|\leq u+v$ for nonnegative numbers $u$ and $v$. Therefore,
\begin{equation}
\mathcal{L}_{\mathrm{tar},\upsilon}(h)
\leq
\mathcal{L}_{\mathrm{src},\upsilon}(h)
+
\operatorname{disc}_{\upsilon}^{\mathrm{pair}}
\big(
\mathbb{P}_{\mathrm{src},\upsilon},
\mathbb{P}_{\mathrm{tar},\upsilon}
\big)
+
\lambda_{\upsilon}^{\star}
+
\epsilon.
\end{equation}
\par
Intersecting this deterministic inequality with the previously derived uniform source-deviation event yields, simultaneously for all $h\in\mathcal{H}_{\upsilon}$,
\begin{equation}
\mathcal{L}_{\mathrm{tar},\upsilon}(h)
\leq
\widehat{\mathcal{L}}_{\mathrm{src},\upsilon}(h)
+
2\mathfrak{R}_{N_{\mathrm{src}}}(\mathcal{F}_{\upsilon})
+
B_{\ell,\upsilon}\sqrt{\frac{\log(1/\delta)}{2N_{\mathrm{src}}}}
+
\operatorname{disc}_{\upsilon}^{\mathrm{pair}}
\big(
\mathbb{P}_{\mathrm{src},\upsilon},
\mathbb{P}_{\mathrm{tar},\upsilon}
\big)
+
\lambda_{\upsilon}^{\star}
+
\epsilon.
\end{equation}
\par
Sending $\epsilon\downarrow 0$ proves the first inequality in the theorem. The second inequality follows immediately from $\mathfrak{R}_{N_{\mathrm{src}}}(\mathcal{F}_{\upsilon})\leq\Psi_{\upsilon}^{\mathrm{cmp}}/\sqrt{N_{\mathrm{src}}}$.\par
\end{proof}
\begin{corollary}\label{cor:unified_bound_erm}
Under the same probability event as in Theorem~\ref{thm:unified_bound}, the source empirical risk minimizer satisfies
\begin{equation}
\mathcal{L}_{\mathrm{tar},\upsilon}
\big(
\widehat{h}_{\mathrm{src},\upsilon}
\big)
\leq
\widehat{\mathcal{L}}_{\mathrm{src},\upsilon}
\big(
\widehat{h}_{\mathrm{src},\upsilon}
\big)
+
\frac{2\Psi_{\upsilon}^{\mathrm{cmp}}}{\sqrt{N_{\mathrm{src}}}}
+
B_{\ell,\upsilon}\sqrt{\frac{\log(1/\delta)}{2N_{\mathrm{src}}}}
+
W_{1,\mathfrak{c}_{\upsilon}}
\big(
\mathbb{P}_{\mathrm{src},\upsilon},
\mathbb{P}_{\mathrm{tar},\upsilon}
\big)
+
\lambda_{\upsilon}^{\star}.
\end{equation}
\par
Since $\mathcal{L}_{\mathrm{tar},\upsilon}^{\star}\geq 0$, the same right-hand side also upper-bounds the excess target risk $\mathcal{E}_{\mathrm{tar},\upsilon}(\widehat{h}_{\mathrm{src},\upsilon})$.\par
\end{corollary}
\begin{proof}
Theorem~\ref{thm:unified_bound} is applied with $h=\widehat{h}_{\mathrm{src},\upsilon}$.
Lemma~\ref{lem:pairwise_transport} yields
\begin{equation}
\operatorname{disc}_{\upsilon}^{\mathrm{pair}}
\big(
\mathbb{P}_{\mathrm{src},\upsilon},
\mathbb{P}_{\mathrm{tar},\upsilon}
\big)
\leq
W_{1,\mathfrak{c}_{\upsilon}}
\big(
\mathbb{P}_{\mathrm{src},\upsilon},
\mathbb{P}_{\mathrm{tar},\upsilon}
\big).
\end{equation}
\par
The excess-risk claim follows from
\begin{equation}
\mathcal{E}_{\mathrm{tar},\upsilon}
\big(
\widehat{h}_{\mathrm{src},\upsilon}
\big)
=
\mathcal{L}_{\mathrm{tar},\upsilon}
\big(
\widehat{h}_{\mathrm{src},\upsilon}
\big)
-
\mathcal{L}_{\mathrm{tar},\upsilon}^{\star}
\leq
\mathcal{L}_{\mathrm{tar},\upsilon}
\big(
\widehat{h}_{\mathrm{src},\upsilon}
\big).
\end{equation}
\par
\end{proof}
The unified bound therefore separates the target-domain error into four transparent terms: $\widehat{\mathcal{L}}_{\mathrm{src},\upsilon}$ measures source empirical fit, $\Psi_{\upsilon}^{\mathrm{cmp}}/\sqrt{N_{\mathrm{src}}}$ together with the concentration factor measures estimation complexity, $W_{1,\mathfrak{c}_{\upsilon}}(\mathbb{P}_{\mathrm{src},\upsilon},\mathbb{P}_{\mathrm{tar},\upsilon})$ measures structured source-to-target shift on the feature-label space, and $\lambda_{\upsilon}^{\star}$ measures the irreducible joint approximation burden that cannot be removed without enlarging the hypothesis class or changing the representation. This representation-level Wasserstein term will be decomposed in Subsection III-B into cross-person, cross-view, and cross-wall contributions induced by human kinematics, radial-velocity projection, and wall propagation mismatch, respectively.\par

\subsection{Cross-Person, Cross-View, and Cross-Wall Generalization Bound}
The unified source-to-target result in Subsection III-A leaves the structured shift term in the compact Wasserstein form
\begin{equation}
W_{1,\mathfrak{c}_{\upsilon}}
\big(
\mathbb{P}_{\mathrm{src},\upsilon},
\mathbb{P}_{\mathrm{tar},\upsilon}
\big).
\end{equation}\par
The purpose of the present subsection is to open this term along the physically meaningful path ``person $\rightarrow$ view $\rightarrow$ wall'' and to retain, in explicit form, the activity-prior mismatch and the residual clutter/noise discrepancy that do not belong to the three principal physical coordinates. Such a decomposition is consistent with recent domain-generalization analyses \cite{SahliGeneralization2021} and with TWR HAR experiments that exhibit structured source-target mismatch \cite{GEBTWRHARConf2024}. Throughout this subsection, the constants $\mathfrak{c}_{\upsilon}$, $B_{\ell,\upsilon}$, and $J_{\upsilon}$ from Subsection III-A, the row-wise clutter-suppression operator $\mathcal{P}_{\mathrm{cs}}$ and the representation Lipschitz constant $L_{\upsilon}$ from Subsection II-C, and the complexity terms $\Psi_{\upsilon}^{\mathrm{cmp}}$ and $\lambda_{\upsilon}^{\star}$ from Subsection III-A retain their earlier definitions and are not redefined.\par
Define the intermediate domain descriptors
\begin{equation}
\mathfrak{d}_{0}
=
\left(
p_{\mathrm{src}},
\beta_{\mathrm{src}},
\boldsymbol{\xi}_{\mathrm{src}}
\right),\quad
\mathfrak{d}_{1}
=
\left(
p_{\mathrm{tar}},
\beta_{\mathrm{src}},
\boldsymbol{\xi}_{\mathrm{src}}
\right),\quad
\mathfrak{d}_{2}
=
\left(
p_{\mathrm{tar}},
\beta_{\mathrm{tar}},
\boldsymbol{\xi}_{\mathrm{src}}
\right),\quad
\mathfrak{d}_{3}
=
\left(
p_{\mathrm{tar}},
\beta_{\mathrm{tar}},
\boldsymbol{\xi}_{\mathrm{tar}}
\right),
\end{equation}
and abbreviate
\begin{equation}
\mathbb{P}_{q,\upsilon}
:=
\mathbb{P}_{\mathfrak{d}_{q},\upsilon},\qquad
q\in\{0,1,2,3\},
\end{equation}
where $\mathbb{P}_{0,\upsilon}=\mathbb{P}_{\mathrm{src},\upsilon}$ and $\mathbb{P}_{3,\upsilon}=\mathbb{P}_{\mathrm{tar},\upsilon}$. For each activity label $a\in\mathcal{A}$, let
\begin{equation}
\widetilde{\mathbf{Z}}_{\mathfrak{d},a}
:=
\mathcal{P}_{\mathrm{cs}}
\left(
\mathcal{G}_{\mathrm{echo}}
\big(
\{\mathbf{r}_{p,a,m},\mathbf{v}_{p,a,m},\varrho_{p,a,m}\}_{m=1}^{M};
\beta,\boldsymbol{\xi}
\big)
\right),
\end{equation}
where $\mathfrak{d}=(p,\beta,\boldsymbol{\xi})$, and $\mathcal{P}_{\mathrm{cs}}$ denotes the row-wise clutter-suppression operator introduced in Subsection II-C. The quantity $\widetilde{\mathbf{Z}}_{\mathfrak{d},a}$ is the ideal clutter-suppressed observation that keeps the moving-target echo but removes the nuisance clutter/noise addends. The induced ideal feature law is denoted by
\begin{equation}
\widetilde{\mathbb{P}}_{\mathfrak{d},\upsilon}^{(a)}
:=
\operatorname{Law}
\left(
\mathcal{R}_{\upsilon}
\big(
\widetilde{\mathbf{Z}}_{\mathfrak{d},a}
\big)
\right),
\end{equation}
and the corresponding activity-conditional laws along the path are abbreviated by
\begin{equation}
\widetilde{\mathbb{P}}_{q,\upsilon}^{(a)}
:=
\widetilde{\mathbb{P}}_{\mathfrak{d}_{q},\upsilon}^{(a)},\qquad
q\in\{0,1,2,3\}.
\end{equation}
\par
The associated ideal feature-label joint law is written as
\begin{equation}
\widetilde{\mathbb{P}}_{q,\upsilon}
:=
\sum_{a=1}^{C}
\alpha_{a}^{(\mathfrak{d}_{q})}
\widetilde{\mathbb{P}}_{q,\upsilon}^{(a)}
\otimes
\delta_{a},
\end{equation}
where $\delta_{a}$ denotes the point mass at label $a$ and $\alpha_{a}^{(\mathfrak{d}_{q})}=\mathbb{P}(A=a\mid \mathfrak{d}_{q})$ is the domain-dependent activity prior introduced in Subsection II-E.\par
For the later path decomposition, define the three activity-prior total-variation fragments
\begin{equation}
\mathfrak{A}_{\mathrm{pri}}^{(q,q+1)}
:=
\frac{1}{2}
\sum_{a=1}^{C}
\left|
\alpha_{a}^{(\mathfrak{d}_{q})}
-
\alpha_{a}^{(\mathfrak{d}_{q+1})}
\right|,\qquad
q\in\{0,1,2\},
\end{equation}
and their aggregate
\begin{equation}
\mathfrak{A}_{\mathrm{pri}}
:=
\mathfrak{A}_{\mathrm{pri}}^{(0,1)}
+
\mathfrak{A}_{\mathrm{pri}}^{(1,2)}
+
\mathfrak{A}_{\mathrm{pri}}^{(2,3)}.
\end{equation}
\par
The later prior-mismatch control also uses
\begin{equation}
\mathfrak{C}_{\mathrm{pri},\upsilon}
:=
8J_{\upsilon}B_{\varphi,\upsilon}
+
2B_{\ell,\upsilon},
\end{equation}
where the feature-diameter contribution $8J_{\upsilon}B_{\varphi,\upsilon}$ follows from $\|\boldsymbol{\varphi}-\boldsymbol{\varphi}'\|_{2}\leq 2B_{\varphi,\upsilon}$ on $\mathfrak{X}_{\upsilon}$.\par
To complement the cross-person kinematic discrepancy from Subsection II-A with a reflectivity mismatch descriptor, define the phase-aligned scattering coefficient
\begin{equation}
\widetilde{\varrho}_{p,a,m}(\phi)
:=
\varrho_{p,a,m}(T_{a}\phi),\qquad
\phi\in[0,1],
\label{eq:phase_aligned_reflectivity}
\end{equation}
and
\begin{equation}
\Sigma_{p,p'}^{(\mathrm{ref})}
:=
\sum_{a=1}^{C}
\rho_{a}
\frac{1}{M}
\sum_{m=1}^{M}
\sup_{\phi\in[0,1]}
\left|
\widetilde{\varrho}_{p,a,m}(\phi)
-
\widetilde{\varrho}_{p',a,m}(\phi)
\right|,
\label{eq:reflectivity_mismatch}
\end{equation}
where $\{\rho_{a}\}_{a=1}^{C}$ are the activity weights used in $D_{p,p'}^{(\mathrm{cp})}$ and $\sum_{a=1}^{C}\rho_{a}=1$. For a fixed coefficient $\omega_{\varrho}>0$, define the extended cross-person discrepancy
\begin{equation}
\overline{D}_{p,p'}^{(\mathrm{cp})}
:=
D_{p,p'}^{(\mathrm{cp})}
+
\omega_{\varrho}\Sigma_{p,p'}^{(\mathrm{ref})}.
\label{eq:extended_cross_person_discrepancy}
\end{equation}
\par
The cross-view and cross-wall discrepancy coordinates are fixed as
\begin{equation}
\mathfrak{V}_{\beta,\beta'}
:=
\left\|
\mathbf{q}_{\beta}
-
\mathbf{q}_{\beta'}
\right\|_{2}
+
\left\|
\mathbf{u}_{\beta}
-
\mathbf{u}_{\beta'}
\right\|_{2}
+
\left|
\tau_{0,\beta}
-
\tau_{0,\beta'}
\right|,
\end{equation}
and
\begin{equation}
\mathfrak{W}_{\boldsymbol{\xi},\boldsymbol{\xi}'}^{(\beta)}
:=
\left|
A_{\mathrm{w}}(\omega_{c};\beta,\boldsymbol{\xi})
-
A_{\mathrm{w}}(\omega_{c};\beta,\boldsymbol{\xi}')
\right|
+
\left|
\tau_{\mathrm{w}}(\beta,\boldsymbol{\xi})
-
\tau_{\mathrm{w}}(\beta,\boldsymbol{\xi}')
\right|,
\end{equation}
respectively. Since the same wall can be interrogated from different views in the step $\mathfrak{d}_{1}\rightarrow\mathfrak{d}_{2}$, the carrier-frequency wall response is additionally assumed to be locally Lipschitz in the view coordinate at fixed wall descriptor, namely
\begin{equation}
\left|
A_{\mathrm{w}}(\omega_{c};\beta,\boldsymbol{\xi})
-
A_{\mathrm{w}}(\omega_{c};\beta',\boldsymbol{\xi})
\right|
\leq
L_{A,\beta}
\mathfrak{V}_{\beta,\beta'},
\end{equation}
\begin{equation}
\left|
\tau_{\mathrm{w}}(\beta,\boldsymbol{\xi})
-
\tau_{\mathrm{w}}(\beta',\boldsymbol{\xi})
\right|
\leq
L_{\tau,\beta}
\mathfrak{V}_{\beta,\beta'},
\end{equation}
for all admissible $\beta$, $\beta'$, and $\boldsymbol{\xi}$, where $L_{A,\beta}>0$ and $L_{\tau,\beta}>0$ are view-Lipschitz wall constants. Finally, let $c_{\mathrm{sg}}>0$ denote a universal first-moment conversion constant for the sub-Gaussian noise envelope, and define the nuisance residual
\begin{equation}
\mathfrak{N}_{\mathrm{res}}
:=
\sqrt{N_{\tau}L_{t}}
\sup_{\tau,t}
\left|
b_{z}^{(\beta_{\mathrm{src}},\boldsymbol{\xi}_{\mathrm{src}})}(\tau,t)
-
b_{z}^{(\beta_{\mathrm{tar}},\boldsymbol{\xi}_{\mathrm{tar}})}(\tau,t)
\right|
+
c_{\mathrm{sg}}\sqrt{N_{\tau}L_{t}}
\big(
\varsigma_{\mathrm{n}}(\beta_{\mathrm{src}},\boldsymbol{\xi}_{\mathrm{src}})
+
\varsigma_{\mathrm{n}}(\beta_{\mathrm{tar}},\boldsymbol{\xi}_{\mathrm{tar}})
\big),
\end{equation}
where the first addend is the clutter-residual Frobenius envelope induced by the entrywise bound after stacking the $N_{\tau}L_{t}$ observation grid.\par
The later constants also use
\begin{equation}
\rho_{\min}
:=
\min_{1\leq a\leq C}\rho_{a},
\qquad
T_{\max}
:=
\max_{1\leq a\leq C}T_{a},
\qquad
B_{\mathrm{pc}}
:=
\sup_{\beta}
\left\|
\mathbf{q}_{\beta}
\right\|_{2},
\qquad
d_{\max}
:=
2(B_{r}+B_{\mathrm{pc}}),
\end{equation}
where $\rho_{\min}>0$ is assumed. Moreover, the waveform autocorrelation is assumed to be Lipschitz continuous on the relevant delay interval, namely
\begin{equation}
\left|
\chi_{u}(\Delta)
-
\chi_{u}(\Delta')
\right|
\leq
L_{\chi}
\left|
\Delta
-
\Delta'
\right|,
\end{equation}
where $L_{\chi}>0$ is the autocorrelation Lipschitz constant.\par
\begin{lemma}\label{lem:path_decomp}
Define, for each adjacent pair $(q,q+1)\in\{(0,1),(1,2),(2,3)\}$,
\begin{equation}
\mathfrak{T}_{\upsilon}^{(q,q+1)}
:=
\sup_{a\in\mathcal{A}}
W_{1,\|\cdot\|_{2}}
\big(
\widetilde{\mathbb{P}}_{q,\upsilon}^{(a)},
\widetilde{\mathbb{P}}_{q+1,\upsilon}^{(a)}
\big),
\end{equation}
where $W_{1,\|\cdot\|_{2}}$ denotes the $1$-Wasserstein distance on the feature space with ground cost $\|\boldsymbol{\varphi}-\boldsymbol{\varphi}'\|_{2}$. Then
\begin{equation}
W_{1,\mathfrak{c}_{\upsilon}}
\big(
\widetilde{\mathbb{P}}_{0,\upsilon},
\widetilde{\mathbb{P}}_{3,\upsilon}
\big)
\leq
4J_{\upsilon}
\left(
\mathfrak{T}_{\upsilon}^{(0,1)}
+
\mathfrak{T}_{\upsilon}^{(1,2)}
+
\mathfrak{T}_{\upsilon}^{(2,3)}
\right)
+
\mathfrak{C}_{\mathrm{pri},\upsilon}\mathfrak{A}_{\mathrm{pri}}.
\end{equation}
\end{lemma}
\begin{proof}
Fix one adjacent pair $(q,q+1)\in\{(0,1),(1,2),(2,3)\}$ and abbreviate
\begin{equation}
\alpha_{a}^{-}
:=
\min
\left\{
\alpha_{a}^{(\mathfrak{d}_{q})},
\alpha_{a}^{(\mathfrak{d}_{q+1})}
\right\},
\end{equation}
and let $\pi_{a}^{\star}$ be an optimal coupling between $\widetilde{\mathbb{P}}_{q,\upsilon}^{(a)}$ and $\widetilde{\mathbb{P}}_{q+1,\upsilon}^{(a)}$ under the feature cost $\|\cdot\|_{2}$. Since
\begin{equation}
\sum_{a=1}^{C}\alpha_{a}^{-}
=
1-\mathfrak{A}_{\mathrm{pri}}^{(q,q+1)},
\end{equation}
the common label mass equals the complement of the activity-prior total variation. A coupling between the two joint laws is now constructed by transporting the common label mass through the same-label couplings $\{\pi_{a}^{\star}\}_{a=1}^{C}$ and by transporting the remaining mass through an arbitrary residual coupling on mismatched labels. Under this construction, the label mismatch probability equals $\mathfrak{A}_{\mathrm{pri}}^{(q,q+1)}$.\par
By the definition of $\mathfrak{c}_{\upsilon}$, every same-label transported pair contributes only the feature cost term
\begin{equation}
\mathfrak{c}_{\upsilon}
\big(
(\boldsymbol{\varphi},a),
(\boldsymbol{\varphi}',a)
\big)
=
4J_{\upsilon}
\left\|
\boldsymbol{\varphi}
-
\boldsymbol{\varphi}'
\right\|_{2},
\end{equation}
whereas every mismatched-label pair contributes at most
\begin{equation}
\mathfrak{c}_{\upsilon}
\big(
(\boldsymbol{\varphi},a),
(\boldsymbol{\varphi}',a')
\big)
\leq
8J_{\upsilon}B_{\varphi,\upsilon}
+
2B_{\ell,\upsilon}
=
\mathfrak{C}_{\mathrm{pri},\upsilon},\qquad a\neq a',
\end{equation}
where the diameter bound on $\mathfrak{X}_{\upsilon}$ has been used. Hence the residual mismatched-label contribution is bounded by $\mathfrak{C}_{\mathrm{pri},\upsilon}\mathfrak{A}_{\mathrm{pri}}^{(q,q+1)}$, whereas the same-label part is controlled by the optimal feature couplings. Consequently,
\begin{equation}
W_{1,\mathfrak{c}_{\upsilon}}
\big(
\widetilde{\mathbb{P}}_{q,\upsilon},
\widetilde{\mathbb{P}}_{q+1,\upsilon}
\big)
\leq
4J_{\upsilon}
\sum_{a=1}^{C}
\alpha_{a}^{-}
W_{1,\|\cdot\|_{2}}
\big(
\widetilde{\mathbb{P}}_{q,\upsilon}^{(a)},
\widetilde{\mathbb{P}}_{q+1,\upsilon}^{(a)}
\big)
+
\mathfrak{C}_{\mathrm{pri},\upsilon}\mathfrak{A}_{\mathrm{pri}}^{(q,q+1)}
\leq
4J_{\upsilon}
\mathfrak{T}_{\upsilon}^{(q,q+1)}
+
\mathfrak{C}_{\mathrm{pri},\upsilon}\mathfrak{A}_{\mathrm{pri}}^{(q,q+1)},
\end{equation}
where $\sum_{a=1}^{C}\alpha_{a}^{-}\leq 1$ has been used in the second step.\par
Applying the triangle inequality of the Wasserstein distance along the path $\mathfrak{d}_{0}\rightarrow\mathfrak{d}_{1}\rightarrow\mathfrak{d}_{2}\rightarrow\mathfrak{d}_{3}$ gives
\begin{equation}
W_{1,\mathfrak{c}_{\upsilon}}
\big(
\widetilde{\mathbb{P}}_{0,\upsilon},
\widetilde{\mathbb{P}}_{3,\upsilon}
\big)
\leq
W_{1,\mathfrak{c}_{\upsilon}}
\big(
\widetilde{\mathbb{P}}_{0,\upsilon},
\widetilde{\mathbb{P}}_{1,\upsilon}
\big)
+
W_{1,\mathfrak{c}_{\upsilon}}
\big(
\widetilde{\mathbb{P}}_{1,\upsilon},
\widetilde{\mathbb{P}}_{2,\upsilon}
\big)
+
W_{1,\mathfrak{c}_{\upsilon}}
\big(
\widetilde{\mathbb{P}}_{2,\upsilon},
\widetilde{\mathbb{P}}_{3,\upsilon}
\big).
\end{equation}
\par
Substituting the adjacent-pair bound three times and using
\begin{equation}
\mathfrak{A}_{\mathrm{pri}}
=
\mathfrak{A}_{\mathrm{pri}}^{(0,1)}
+
\mathfrak{A}_{\mathrm{pri}}^{(1,2)}
+
\mathfrak{A}_{\mathrm{pri}}^{(2,3)},
\end{equation}
proves the claim.\par
\end{proof}
The next three lemmas upper-bound the adjacent ideal feature transports in terms of cross-person, cross-view, and cross-wall coordinates, respectively. Their proofs all use the same propagation chain: matched-filter echo perturbation $\rightarrow$ row-wise clutter suppression $\rightarrow$ representation map.\par
\begin{lemma}\label{lem:cross_person_transport}
Let
\begin{equation}
\gamma_{\mathrm{per},r}
:=
4|\kappa_{0}|B_{\varrho}B_{\mathrm{w}}B_{\chi}d_{\max}d_{\min}^{-4}
+
\frac{4\pi B_{\Gamma}B_{\chi}}{\lambda_{c}}
+
\frac{2B_{\Gamma}L_{\chi}}{c},
\end{equation}
\begin{equation}
\gamma_{\mathrm{per},v}
:=
\frac{4\pi B_{\Gamma}B_{\chi}T_{\max}}{\lambda_{c}},
\end{equation}
and
\begin{equation}
\gamma_{\mathrm{per},\varrho}
:=
|\kappa_{0}|B_{\mathrm{w}}B_{\chi}d_{\min}^{-2}.
\end{equation}
\par
Define
\begin{equation}
\Gamma_{\mathrm{per}}
:=
M\sqrt{N_{\tau}L_{t}}
\left(
\frac{\gamma_{\mathrm{per},r}}{\rho_{\min}\omega_{r}}
+
\frac{\gamma_{\mathrm{per},v}}{\rho_{\min}\omega_{v}}
+
\frac{\gamma_{\mathrm{per},\varrho}}{\rho_{\min}\omega_{\varrho}}
\right),
\end{equation}
where $\omega_{r}>0$, $\omega_{v}>0$, and $\omega_{\varrho}>0$ are the position, velocity, and reflectivity weights appearing in $\overline{D}_{p,p'}^{(\mathrm{cp})}$. The state weight $\omega_{s}$ is allowed to remain merely nonnegative, because the later reduction uses only the monotone inequalities $\delta_{p,p'}^{(\mathrm{cp})}(a)\geq \omega_{r}\overline{\Delta}_{r,p,p'}(a)$ and $\delta_{p,p'}^{(\mathrm{cp})}(a)\geq \omega_{v}\overline{\Delta}_{v,p,p'}(a)$. Then
\begin{equation}
\mathfrak{T}_{\upsilon}^{(0,1)}
\leq
L_{\upsilon}\Gamma_{\mathrm{per}}
\overline{D}_{p_{\mathrm{src}},p_{\mathrm{tar}}}^{(\mathrm{cp})}.
\end{equation}
\end{lemma}
\begin{proof}
Fix one activity label $a\in\mathcal{A}$. The two adjacent domains $\mathfrak{d}_{0}$ and $\mathfrak{d}_{1}$ share the same view $\beta_{\mathrm{src}}$ and the same wall parameter $\boldsymbol{\xi}_{\mathrm{src}}$, so only the subject index changes. Let
\begin{equation}
\delta_{a}^{\mathrm{per}}
:=
\delta_{p_{\mathrm{src}},p_{\mathrm{tar}}}^{(\mathrm{cp})}(a).
\end{equation}
\par
Since $\rho_{\min}\leq \rho_{a}$ and
\begin{equation}
D_{p_{\mathrm{src}},p_{\mathrm{tar}}}^{(\mathrm{cp})}
=
\sum_{\bar{a}=1}^{C}
\rho_{\bar{a}}
\delta_{p_{\mathrm{src}},p_{\mathrm{tar}}}^{(\mathrm{cp})}(\bar{a}),
\end{equation}
one has
\begin{equation}
\delta_{a}^{\mathrm{per}}
\leq
\frac{D_{p_{\mathrm{src}},p_{\mathrm{tar}}}^{(\mathrm{cp})}}{\rho_{\min}}
\leq
\frac{\overline{D}_{p_{\mathrm{src}},p_{\mathrm{tar}}}^{(\mathrm{cp})}}{\rho_{\min}}.
\end{equation}
\par
Therefore, by \eqref{eq:cross_person_discrepancy} and \eqref{eq:reflectivity_mismatch},
\begin{equation}
\overline{\Delta}_{r,p_{\mathrm{src}},p_{\mathrm{tar}}}(a)
\leq
\frac{\delta_{a}^{\mathrm{per}}}{\omega_{r}}
\leq
\frac{\overline{D}_{p_{\mathrm{src}},p_{\mathrm{tar}}}^{(\mathrm{cp})}}{\rho_{\min}\omega_{r}},
\end{equation}
\begin{equation}
\overline{\Delta}_{v,p_{\mathrm{src}},p_{\mathrm{tar}}}(a)
\leq
\frac{\delta_{a}^{\mathrm{per}}}{\omega_{v}}
\leq
\frac{\overline{D}_{p_{\mathrm{src}},p_{\mathrm{tar}}}^{(\mathrm{cp})}}{\rho_{\min}\omega_{v}},
\end{equation}
and
\begin{equation}
\frac{1}{M}
\sum_{m=1}^{M}
\sup_{\phi\in[0,1]}
\left|
\widetilde{\varrho}_{p_{\mathrm{src}},a,m}(\phi)
-
\widetilde{\varrho}_{p_{\mathrm{tar}},a,m}(\phi)
\right|
\leq
\frac{\Sigma_{p_{\mathrm{src}},p_{\mathrm{tar}}}^{(\mathrm{ref})}}{\rho_{\min}}
\leq
\frac{\overline{D}_{p_{\mathrm{src}},p_{\mathrm{tar}}}^{(\mathrm{cp})}}{\rho_{\min}\omega_{\varrho}}.
\end{equation}
\par
For every slow time $t\in[0,T_{a}]$, let $\phi=t/T_{a}$. By \eqref{eq:uniform_activity_duration} and \eqref{eq:phase_aligned_human_state},
\begin{equation}
\left\|
\mathbf{r}_{p_{\mathrm{src}},a,m}(t)
-
\mathbf{r}_{p_{\mathrm{tar}},a,m}(t)
\right\|_{2}
\leq
\sup_{\phi'\in[0,1]}
\left\|
\widetilde{\mathbf{r}}_{p_{\mathrm{src}},a,m}(\phi')
-
\widetilde{\mathbf{r}}_{p_{\mathrm{tar}},a,m}(\phi')
\right\|_{2},
\end{equation}
\begin{equation}
\left\|
\mathbf{v}_{p_{\mathrm{src}},a,m}(t)
-
\mathbf{v}_{p_{\mathrm{tar}},a,m}(t)
\right\|_{2}
\leq
\sup_{\phi'\in[0,1]}
\left\|
\widetilde{\mathbf{v}}_{p_{\mathrm{src}},a,m}(\phi')
-
\widetilde{\mathbf{v}}_{p_{\mathrm{tar}},a,m}(\phi')
\right\|_{2},
\end{equation}
and
\begin{equation}
\left|
\varrho_{p_{\mathrm{src}},a,m}(t)
-
\varrho_{p_{\mathrm{tar}},a,m}(t)
\right|
\leq
\sup_{\phi'\in[0,1]}
\left|
\widetilde{\varrho}_{p_{\mathrm{src}},a,m}(\phi')
-
\widetilde{\varrho}_{p_{\mathrm{tar}},a,m}(\phi')
\right|.
\end{equation}
\par
Consider now the matched-filtered moving-target echo at a fixed sample pair $(\tau,t)$. Using the notation of Subsection II-B, write the signal-only component as
\begin{equation}
s_{p,a}^{(\beta,\boldsymbol{\xi})}(\tau,t)
:=
\sum_{m=1}^{M}
\Gamma_{p,a,m}^{(\beta,\boldsymbol{\xi})}(t)
\chi_{u}
\big(
\tau
-
\tau_{\mathrm{w}}(\beta,\boldsymbol{\xi})
-
\tau_{p,a,m}^{(\beta)}(t)
\big).
\end{equation}
\par
For the same view-wall pair $(\beta_{\mathrm{src}},\boldsymbol{\xi}_{\mathrm{src}})$, add and subtract
\begin{equation}
\Gamma_{p_{\mathrm{tar}},a,m}^{(\beta_{\mathrm{src}},\boldsymbol{\xi}_{\mathrm{src}})}(t)
\chi_{u}
\big(
\tau
-
\tau_{\mathrm{w}}(\beta_{\mathrm{src}},\boldsymbol{\xi}_{\mathrm{src}})
-
\tau_{p_{\mathrm{src}},a,m}^{(\beta_{\mathrm{src}})}(t)
\big),
\end{equation}
inside the difference. Then
\begin{equation}
\begin{aligned}
\big|
s_{p_{\mathrm{src}},a}^{(\beta_{\mathrm{src}},\boldsymbol{\xi}_{\mathrm{src}})}(\tau,t)
-
s_{p_{\mathrm{tar}},a}^{(\beta_{\mathrm{src}},\boldsymbol{\xi}_{\mathrm{src}})}(\tau,t)
\big|
&\leq
\sum_{m=1}^{M}
\big|
\Gamma_{p_{\mathrm{src}},a,m}^{(\beta_{\mathrm{src}},\boldsymbol{\xi}_{\mathrm{src}})}(t)
-
\Gamma_{p_{\mathrm{tar}},a,m}^{(\beta_{\mathrm{src}},\boldsymbol{\xi}_{\mathrm{src}})}(t)
\big|
B_{\chi}\\
&\quad+
\sum_{m=1}^{M}
B_{\Gamma}
\big|
\chi_{u}
\big(
\tau-\tau_{\mathrm{w}}-\tau_{p_{\mathrm{src}},a,m}^{(\beta_{\mathrm{src}})}(t)
\big)
-
\chi_{u}
\big(
\tau-\tau_{\mathrm{w}}-\tau_{p_{\mathrm{tar}},a,m}^{(\beta_{\mathrm{src}})}(t)
\big)
\big|,
\end{aligned}
\end{equation}
where $\tau_{\mathrm{w}}=\tau_{\mathrm{w}}(\beta_{\mathrm{src}},\boldsymbol{\xi}_{\mathrm{src}})$ is used for compactness.\par
The autocorrelation term is bounded by the delay difference. Since
\begin{equation}
\tau_{p,a,m}^{(\beta_{\mathrm{src}})}(t)
=
\frac{d_{p,a,m}^{(\beta_{\mathrm{src}})}(t)}{c},
\end{equation}
and
\begin{equation}
\left|
d_{p_{\mathrm{src}},a,m}^{(\beta_{\mathrm{src}})}(t)
-
d_{p_{\mathrm{tar}},a,m}^{(\beta_{\mathrm{src}})}(t)
\right|
\leq
2
\left\|
\mathbf{r}_{p_{\mathrm{src}},a,m}(t)
-
\mathbf{r}_{p_{\mathrm{tar}},a,m}(t)
\right\|_{2},
\end{equation}
one has
\begin{equation}
\big|
\chi_{u}
\big(
\tau-\tau_{\mathrm{w}}-\tau_{p_{\mathrm{src}},a,m}^{(\beta_{\mathrm{src}})}(t)
\big)
-
\chi_{u}
\big(
\tau-\tau_{\mathrm{w}}-\tau_{p_{\mathrm{tar}},a,m}^{(\beta_{\mathrm{src}})}(t)
\big)
\big|
\leq
\frac{2L_{\chi}}{c}
\left\|
\mathbf{r}_{p_{\mathrm{src}},a,m}(t)
-
\mathbf{r}_{p_{\mathrm{tar}},a,m}(t)
\right\|_{2}.
\end{equation}
\par
It remains to control the coefficient difference. Since the wall factor is identical here, the carrier phase is decomposed as in Subsection II-B into an initial projected-position factor and an accumulated micro-Doppler factor. By adding and subtracting four intermediate terms, by using $\big|e^{-jx}-e^{-jy}\big|\leq |x-y|$, and by using
\begin{equation}
\left|
d^{-2}
-
d'^{-2}
\right|
=
\frac{|d+d'||d-d'|}{d^{2}d'^{2}}
\leq
2d_{\max}d_{\min}^{-4}|d-d'|,
\end{equation}
one obtains
\begin{equation}
\begin{aligned}
\big|
\Gamma_{p_{\mathrm{src}},a,m}^{(\beta_{\mathrm{src}},\boldsymbol{\xi}_{\mathrm{src}})}(t)
-
\Gamma_{p_{\mathrm{tar}},a,m}^{(\beta_{\mathrm{src}},\boldsymbol{\xi}_{\mathrm{src}})}(t)
\big|
&\leq
|\kappa_{0}|B_{\mathrm{w}}d_{\min}^{-2}
\big|
\varrho_{p_{\mathrm{src}},a,m}(t)
-
\varrho_{p_{\mathrm{tar}},a,m}(t)
\big|\\
&\quad+
4|\kappa_{0}|B_{\varrho}B_{\mathrm{w}}d_{\max}d_{\min}^{-4}
\left\|
\mathbf{r}_{p_{\mathrm{src}},a,m}(t)
-
\mathbf{r}_{p_{\mathrm{tar}},a,m}(t)
\right\|_{2}\\
&\quad+
\frac{4\pi B_{\Gamma}}{\lambda_{c}}
\left\|
\mathbf{r}_{p_{\mathrm{src}},a,m}(0)
-
\mathbf{r}_{p_{\mathrm{tar}},a,m}(0)
\right\|_{2}\\
&\quad+
2\pi B_{\Gamma}
\left|
\int_{0}^{t}
\big(
f_{p_{\mathrm{src}},a,m}^{(\beta_{\mathrm{src}})}(\zeta)
-
f_{p_{\mathrm{tar}},a,m}^{(\beta_{\mathrm{src}})}(\zeta)
\big)
d\zeta
\right|,
\end{aligned}
\end{equation}
where the third line comes from the initial projected-position phase and the last line uses the accumulated micro-Doppler phase representation from Subsection II-B. Because
\begin{equation}
f_{p,a,m}^{(\beta_{\mathrm{src}})}(\zeta)
\approx
\frac{2}{\lambda_{c}}
\mathbf{u}_{\beta_{\mathrm{src}}}^{\top}
\mathbf{v}_{p,a,m}(\zeta),
\end{equation}
the radial-velocity difference satisfies
\begin{equation}
\left|
f_{p_{\mathrm{src}},a,m}^{(\beta_{\mathrm{src}})}(\zeta)
-
f_{p_{\mathrm{tar}},a,m}^{(\beta_{\mathrm{src}})}(\zeta)
\right|
\leq
\frac{2}{\lambda_{c}}
\left\|
\mathbf{v}_{p_{\mathrm{src}},a,m}(\zeta)
-
\mathbf{v}_{p_{\mathrm{tar}},a,m}(\zeta)
\right\|_{2},
\end{equation}
and therefore, with the change of variables $\zeta=T_{a}\phi$ together with $t\leq T_{a}\leq T_{\max}$,
\begin{equation}
\begin{aligned}
\left|
\int_{0}^{t}
\big(
f_{p_{\mathrm{src}},a,m}^{(\beta_{\mathrm{src}})}(\zeta)
-
f_{p_{\mathrm{tar}},a,m}^{(\beta_{\mathrm{src}})}(\zeta)
\big)
d\zeta
\right|
&\leq
\frac{2T_{a}}{\lambda_{c}}
\sup_{\phi\in[0,t/T_{a}]}
\left\|
\widetilde{\mathbf{v}}_{p_{\mathrm{src}},a,m}(\phi)
-
\widetilde{\mathbf{v}}_{p_{\mathrm{tar}},a,m}(\phi)
\right\|_{2}\\
&\leq
\frac{2T_{\max}}{\lambda_{c}}
\sup_{\phi\in[0,1]}
\left\|
\widetilde{\mathbf{v}}_{p_{\mathrm{src}},a,m}(\phi)
-
\widetilde{\mathbf{v}}_{p_{\mathrm{tar}},a,m}(\phi)
\right\|_{2}.
\end{aligned}
\end{equation}
\par
Moreover, since $t=0$ corresponds to $\phi=0$ under the common activity duration $T_{a}$,
\begin{equation}
\left\|
\mathbf{r}_{p_{\mathrm{src}},a,m}(0)
-
\mathbf{r}_{p_{\mathrm{tar}},a,m}(0)
\right\|_{2}
\leq
\sup_{\phi\in[0,1]}
\left\|
\widetilde{\mathbf{r}}_{p_{\mathrm{src}},a,m}(\phi)
-
\widetilde{\mathbf{r}}_{p_{\mathrm{tar}},a,m}(\phi)
\right\|_{2}.
\end{equation}
\par
Substituting these estimates into the signal difference yields
\begin{equation}
\begin{aligned}
\big|
s_{p_{\mathrm{src}},a}^{(\beta_{\mathrm{src}},\boldsymbol{\xi}_{\mathrm{src}})}(\tau,t)
-
s_{p_{\mathrm{tar}},a}^{(\beta_{\mathrm{src}},\boldsymbol{\xi}_{\mathrm{src}})}(\tau,t)
\big|
&\leq
M\gamma_{\mathrm{per},r}
\overline{\Delta}_{r,p_{\mathrm{src}},p_{\mathrm{tar}}}(a)
+
M\gamma_{\mathrm{per},v}
\overline{\Delta}_{v,p_{\mathrm{src}},p_{\mathrm{tar}}}(a)\\
&\quad+
M\gamma_{\mathrm{per},\varrho}
\frac{1}{M}
\sum_{m=1}^{M}
\sup_{\phi\in[0,1]}
\left|
\widetilde{\varrho}_{p_{\mathrm{src}},a,m}(\phi)
-
\widetilde{\varrho}_{p_{\mathrm{tar}},a,m}(\phi)
\right|\\
&\leq
M
\left(
\frac{\gamma_{\mathrm{per},r}}{\rho_{\min}\omega_{r}}
+
\frac{\gamma_{\mathrm{per},v}}{\rho_{\min}\omega_{v}}
+
\frac{\gamma_{\mathrm{per},\varrho}}{\rho_{\min}\omega_{\varrho}}
\right)
\overline{D}_{p_{\mathrm{src}},p_{\mathrm{tar}}}^{(\mathrm{cp})}.
\end{aligned}
\end{equation}
\par
For compactness, define the signal-only matched-filtered matrix
\begin{equation}
\mathbf{S}_{p,a}^{(\beta,\boldsymbol{\xi})}
:=
\mathcal{G}_{\mathrm{echo}}
\big(
\{\mathbf{r}_{p,a,m},\mathbf{v}_{p,a,m},\varrho_{p,a,m}\}_{m=1}^{M};
\beta,\boldsymbol{\xi}
\big),
\end{equation}
where the clutter and noise terms are excluded by construction.\par
After sampling over the $N_{\tau}L_{t}$ range-slow-time grid, the Frobenius norm of the signal-only matrix difference is bounded by
\begin{equation}
\left\|
\mathbf{S}_{p_{\mathrm{src}},a}^{(\beta_{\mathrm{src}},\boldsymbol{\xi}_{\mathrm{src}})}
-
\mathbf{S}_{p_{\mathrm{tar}},a}^{(\beta_{\mathrm{src}},\boldsymbol{\xi}_{\mathrm{src}})}
\right\|_{F}
\leq
\Gamma_{\mathrm{per}}
\overline{D}_{p_{\mathrm{src}},p_{\mathrm{tar}}}^{(\mathrm{cp})}.
\end{equation}
\par
Because $\mathcal{P}_{\mathrm{cs}}$ is an orthogonal projection applied row-wise, it is non-expansive in Frobenius norm. Therefore,
\begin{equation}
\left\|
\widetilde{\mathbf{Z}}_{\mathfrak{d}_{0},a}
-
\widetilde{\mathbf{Z}}_{\mathfrak{d}_{1},a}
\right\|_{F}
\leq
\Gamma_{\mathrm{per}}
\overline{D}_{p_{\mathrm{src}},p_{\mathrm{tar}}}^{(\mathrm{cp})}.
\end{equation}
\par
Applying the representation Lipschitz property finally gives
\begin{equation}
W_{1,\|\cdot\|_{2}}
\big(
\widetilde{\mathbb{P}}_{0,\upsilon}^{(a)},
\widetilde{\mathbb{P}}_{1,\upsilon}^{(a)}
\big)
\leq
L_{\upsilon}\Gamma_{\mathrm{per}}
\overline{D}_{p_{\mathrm{src}},p_{\mathrm{tar}}}^{(\mathrm{cp})}.
\end{equation}
\par
Taking the supremum over $a\in\mathcal{A}$ proves the claim.\par
\end{proof}
\begin{lemma}\label{lem:cross_view_transport}
Let
\begin{equation}
\gamma_{\mathrm{view},q}
:=
4|\kappa_{0}|B_{\varrho}B_{\mathrm{w}}B_{\chi}d_{\max}d_{\min}^{-4},
\end{equation}
\begin{equation}
\gamma_{\mathrm{view},w}
:=
|\kappa_{0}|B_{\varrho}B_{\chi}d_{\min}^{-2}L_{A,\beta}
+
\left(
2\pi f_{c}B_{\Gamma}B_{\chi}
+
B_{\Gamma}L_{\chi}
\right)
L_{\tau,\beta},
\end{equation}
\begin{equation}
\gamma_{\mathrm{view},u}
:=
\frac{4\pi B_{\Gamma}B_{\chi}B_{r}}{\lambda_{c}}
+
\frac{4\pi B_{\Gamma}B_{\chi}T_{\max}B_{v}}{\lambda_{c}}
+
\frac{2B_{\Gamma}L_{\chi}B_{r}}{c},
\end{equation}
and
\begin{equation}
\gamma_{\mathrm{view},\tau}
:=
2\pi f_{c}B_{\Gamma}B_{\chi}
+
B_{\Gamma}L_{\chi}.
\end{equation}
\par
Define
\begin{equation}
\Gamma_{\mathrm{view}}
:=
M\sqrt{N_{\tau}L_{t}}
\big(
\gamma_{\mathrm{view},q}
+
\gamma_{\mathrm{view},w}
+
\gamma_{\mathrm{view},u}
+
\gamma_{\mathrm{view},\tau}
\big).
\end{equation}
\par
Then
\begin{equation}
\mathfrak{T}_{\upsilon}^{(1,2)}
\leq
L_{\upsilon}\Gamma_{\mathrm{view}}
\mathfrak{V}_{\beta_{\mathrm{src}},\beta_{\mathrm{tar}}}.
\end{equation}
\end{lemma}
\begin{proof}
Fix one activity label $a\in\mathcal{A}$. The two adjacent domains $\mathfrak{d}_{1}$ and $\mathfrak{d}_{2}$ share the same subject $p_{\mathrm{tar}}$ and the same wall parameter $\boldsymbol{\xi}_{\mathrm{src}}$, whereas the view changes from $\beta_{\mathrm{src}}$ to $\beta_{\mathrm{tar}}$. The wall descriptor is fixed, but the effective wall response still changes with the view through $A_{\mathrm{w}}(\omega_{c};\beta,\boldsymbol{\xi}_{\mathrm{src}})$ and $\tau_{\mathrm{w}}(\beta,\boldsymbol{\xi}_{\mathrm{src}})$. For one scatterer index $m$ and one slow-time instant $t$, the two-way distance difference satisfies
\begin{equation}
\left|
d_{p_{\mathrm{tar}},a,m}^{(\beta_{\mathrm{src}})}(t)
-
d_{p_{\mathrm{tar}},a,m}^{(\beta_{\mathrm{tar}})}(t)
\right|
\leq
2
\left\|
\mathbf{q}_{\beta_{\mathrm{src}}}
-
\mathbf{q}_{\beta_{\mathrm{tar}}}
\right\|_{2},
\end{equation}
because only the radar phase-center location changes in the exact expression $d_{p,a,m}^{(\beta)}(t)=2\|\mathbf{r}_{p,a,m}(t)-\mathbf{q}_{\beta}\|_{2}$. The far-field delay linearization from Subsection II-B yields
\begin{equation}
\tau_{p_{\mathrm{tar}},a,m}^{(\beta)}(t)
\approx
\tau_{0,\beta}
-
\frac{2}{c}
\mathbf{u}_{\beta}^{\top}
\mathbf{r}_{p_{\mathrm{tar}},a,m}(t),
\end{equation}
so that
\begin{equation}
\begin{aligned}
\big|
\tau_{p_{\mathrm{tar}},a,m}^{(\beta_{\mathrm{src}})}(t)
-
\tau_{p_{\mathrm{tar}},a,m}^{(\beta_{\mathrm{tar}})}(t)
\big|
&\leq
\left|
\tau_{0,\beta_{\mathrm{src}}}
-
\tau_{0,\beta_{\mathrm{tar}}}
\right|
+
\frac{2}{c}
\left|
\big(
\mathbf{u}_{\beta_{\mathrm{src}}}
-
\mathbf{u}_{\beta_{\mathrm{tar}}}
\big)^{\top}
\mathbf{r}_{p_{\mathrm{tar}},a,m}(t)
\right|\\
&\leq
\left|
\tau_{0,\beta_{\mathrm{src}}}
-
\tau_{0,\beta_{\mathrm{tar}}}
\right|
+
\frac{2B_{r}}{c}
\left\|
\mathbf{u}_{\beta_{\mathrm{src}}}
-
\mathbf{u}_{\beta_{\mathrm{tar}}}
\right\|_{2}.
\end{aligned}
\label{eq:view_delay_difference}
\end{equation}
\par
Likewise, the micro-Doppler frequency relation
\begin{equation}
f_{p_{\mathrm{tar}},a,m}^{(\beta)}(t)
\approx
\frac{2}{\lambda_{c}}
\mathbf{u}_{\beta}^{\top}
\mathbf{v}_{p_{\mathrm{tar}},a,m}(t),
\end{equation}
gives the explicit radial-velocity projection bound
\begin{equation}
\big|
f_{p_{\mathrm{tar}},a,m}^{(\beta_{\mathrm{src}})}(t)
-
f_{p_{\mathrm{tar}},a,m}^{(\beta_{\mathrm{tar}})}(t)
\big|
\leq
\frac{2}{\lambda_{c}}
\left|
\big(
\mathbf{u}_{\beta_{\mathrm{src}}}
-
\mathbf{u}_{\beta_{\mathrm{tar}}}
\big)^{\top}
\mathbf{v}_{p_{\mathrm{tar}},a,m}(t)
\right|
\leq
\frac{2B_{v}}{\lambda_{c}}
\left\|
\mathbf{u}_{\beta_{\mathrm{src}}}
-
\mathbf{u}_{\beta_{\mathrm{tar}}}
\right\|_{2}.
\label{eq:view_radial_velocity_projection}
\end{equation}
\par
The view-Lipschitz wall assumptions yield
\begin{equation}
\left|
A_{\mathrm{w}}(\omega_{c};\beta_{\mathrm{src}},\boldsymbol{\xi}_{\mathrm{src}})
-
A_{\mathrm{w}}(\omega_{c};\beta_{\mathrm{tar}},\boldsymbol{\xi}_{\mathrm{src}})
\right|
\leq
L_{A,\beta}
\mathfrak{V}_{\beta_{\mathrm{src}},\beta_{\mathrm{tar}}},
\label{eq:view_wall_amplitude_difference}
\end{equation}
\begin{equation}
\left|
\tau_{\mathrm{w}}(\beta_{\mathrm{src}},\boldsymbol{\xi}_{\mathrm{src}})
-
\tau_{\mathrm{w}}(\beta_{\mathrm{tar}},\boldsymbol{\xi}_{\mathrm{src}})
\right|
\leq
L_{\tau,\beta}
\mathfrak{V}_{\beta_{\mathrm{src}},\beta_{\mathrm{tar}}}.
\label{eq:view_wall_delay_difference}
\end{equation}
\par
The matched-filtered moving-target signal is compared exactly as in the previous lemma. Adding and subtracting an intermediate term yields
\begin{equation}
\begin{aligned}
\big|
s_{p_{\mathrm{tar}},a}^{(\beta_{\mathrm{src}},\boldsymbol{\xi}_{\mathrm{src}})}(\tau,t)
-
s_{p_{\mathrm{tar}},a}^{(\beta_{\mathrm{tar}},\boldsymbol{\xi}_{\mathrm{src}})}(\tau,t)
\big|
&\leq
\sum_{m=1}^{M}
\big|
\Gamma_{p_{\mathrm{tar}},a,m}^{(\beta_{\mathrm{src}},\boldsymbol{\xi}_{\mathrm{src}})}(t)
-
\Gamma_{p_{\mathrm{tar}},a,m}^{(\beta_{\mathrm{tar}},\boldsymbol{\xi}_{\mathrm{src}})}(t)
\big|
B_{\chi}\\
&\quad+
\sum_{m=1}^{M}
B_{\Gamma}
\big|
\chi_{u}
\big(
\tau-\tau_{\mathrm{w}}^{\mathrm{src}}-\tau_{p_{\mathrm{tar}},a,m}^{(\beta_{\mathrm{src}})}(t)
\big)
-
\chi_{u}
\big(
\tau-\tau_{\mathrm{w}}^{\mathrm{tar}}-\tau_{p_{\mathrm{tar}},a,m}^{(\beta_{\mathrm{tar}})}(t)
\big)
\big|,
\end{aligned}
\end{equation}
where $\tau_{\mathrm{w}}^{\mathrm{src}}=\tau_{\mathrm{w}}(\beta_{\mathrm{src}},\boldsymbol{\xi}_{\mathrm{src}})$ and $\tau_{\mathrm{w}}^{\mathrm{tar}}=\tau_{\mathrm{w}}(\beta_{\mathrm{tar}},\boldsymbol{\xi}_{\mathrm{src}})$. The second term is bounded by
\begin{equation}
\big|
\chi_{u}
\big(
\tau-\tau_{\mathrm{w}}^{\mathrm{src}}-\tau_{p_{\mathrm{tar}},a,m}^{(\beta_{\mathrm{src}})}(t)
\big)
-
\chi_{u}
\big(
\tau-\tau_{\mathrm{w}}^{\mathrm{tar}}-\tau_{p_{\mathrm{tar}},a,m}^{(\beta_{\mathrm{tar}})}(t)
\big)
\big|
\leq
L_{\chi}
\left(
\big|
\tau_{\mathrm{w}}^{\mathrm{src}}
-
\tau_{\mathrm{w}}^{\mathrm{tar}}
\big|
+
\big|
\tau_{p_{\mathrm{tar}},a,m}^{(\beta_{\mathrm{src}})}(t)
-
\tau_{p_{\mathrm{tar}},a,m}^{(\beta_{\mathrm{tar}})}(t)
\big|
\right).
\label{eq:view_kernel_difference}
\end{equation}
\par
To control the coefficient difference, note that the scattering coefficient remains unchanged, whereas the view-dependent wall factor, the spherical-spreading factor, the wall-delay phase, the nominal reference delay phase, the initial projected-position phase, and the accumulated radial-velocity phase all vary with the view. By adding and subtracting suitable intermediate terms,
\begin{equation}
\begin{aligned}
\big|
\Gamma_{p_{\mathrm{tar}},a,m}^{(\beta_{\mathrm{src}},\boldsymbol{\xi}_{\mathrm{src}})}(t)
-
\Gamma_{p_{\mathrm{tar}},a,m}^{(\beta_{\mathrm{tar}},\boldsymbol{\xi}_{\mathrm{src}})}(t)
\big|
&\leq
|\kappa_{0}|B_{\varrho}d_{\min}^{-2}
\left|
A_{\mathrm{w}}(\omega_{c};\beta_{\mathrm{src}},\boldsymbol{\xi}_{\mathrm{src}})
-
A_{\mathrm{w}}(\omega_{c};\beta_{\mathrm{tar}},\boldsymbol{\xi}_{\mathrm{src}})
\right|\\
&\quad+
4|\kappa_{0}|B_{\varrho}B_{\mathrm{w}}d_{\max}d_{\min}^{-4}
\left\|
\mathbf{q}_{\beta_{\mathrm{src}}}
-
\mathbf{q}_{\beta_{\mathrm{tar}}}
\right\|_{2}\\
&\quad+
2\pi f_{c}B_{\Gamma}
\left|
\tau_{\mathrm{w}}^{\mathrm{src}}
-
\tau_{\mathrm{w}}^{\mathrm{tar}}
\right|\\
&\quad+
\frac{4\pi B_{\Gamma}B_{r}}{\lambda_{c}}
\left\|
\mathbf{u}_{\beta_{\mathrm{src}}}
-
\mathbf{u}_{\beta_{\mathrm{tar}}}
\right\|_{2}
+
2\pi B_{\Gamma}f_{c}
\left|
\tau_{0,\beta_{\mathrm{src}}}
-
\tau_{0,\beta_{\mathrm{tar}}}
\right|\\
&\quad+
2\pi B_{\Gamma}
\left|
\int_{0}^{t}
\big(
f_{p_{\mathrm{tar}},a,m}^{(\beta_{\mathrm{src}})}(\zeta)
-
f_{p_{\mathrm{tar}},a,m}^{(\beta_{\mathrm{tar}})}(\zeta)
\big)
d\zeta
\right|,
\end{aligned}
\label{eq:view_coefficient_difference}
\end{equation}
where the second and third lines collect the wall-delay, initial projected-position, and nominal reference-delay phases, and the last line is controlled by the radial-velocity projection bound:
\begin{equation}
\left|
\int_{0}^{t}
\big(
f_{p_{\mathrm{tar}},a,m}^{(\beta_{\mathrm{src}})}(\zeta)
-
f_{p_{\mathrm{tar}},a,m}^{(\beta_{\mathrm{tar}})}(\zeta)
\big)
d\zeta
\right|
\leq
\frac{2T_{\max}B_{v}}{\lambda_{c}}
\left\|
\mathbf{u}_{\beta_{\mathrm{src}}}
-
\mathbf{u}_{\beta_{\mathrm{tar}}}
\right\|_{2}.
\label{eq:view_phase_integral_difference}
\end{equation}
\par
Hence, by combining \eqref{eq:view_delay_difference}, \eqref{eq:view_radial_velocity_projection}, \eqref{eq:view_wall_amplitude_difference}, \eqref{eq:view_wall_delay_difference}, \eqref{eq:view_kernel_difference}, \eqref{eq:view_coefficient_difference}, and \eqref{eq:view_phase_integral_difference} with the definition of $\mathfrak{V}_{\beta_{\mathrm{src}},\beta_{\mathrm{tar}}}$,
\begin{equation}
\big|
s_{p_{\mathrm{tar}},a}^{(\beta_{\mathrm{src}},\boldsymbol{\xi}_{\mathrm{src}})}(\tau,t)
-
s_{p_{\mathrm{tar}},a}^{(\beta_{\mathrm{tar}},\boldsymbol{\xi}_{\mathrm{src}})}(\tau,t)
\big|
\leq
M
\big(
\gamma_{\mathrm{view},q}
+
\gamma_{\mathrm{view},w}
+
\gamma_{\mathrm{view},u}
+
\gamma_{\mathrm{view},\tau}
\big)
\mathfrak{V}_{\beta_{\mathrm{src}},\beta_{\mathrm{tar}}}.
\end{equation}
\par
Therefore,
\begin{equation}
\left\|
\widetilde{\mathbf{Z}}_{\mathfrak{d}_{1},a}
-
\widetilde{\mathbf{Z}}_{\mathfrak{d}_{2},a}
\right\|_{F}
\leq
\Gamma_{\mathrm{view}}
\mathfrak{V}_{\beta_{\mathrm{src}},\beta_{\mathrm{tar}}},
\end{equation}
where the non-expansiveness of $\mathcal{P}_{\mathrm{cs}}$ has again been used. Applying the Lipschitz continuity of $\mathcal{R}_{\upsilon}$ yields
\begin{equation}
W_{1,\|\cdot\|_{2}}
\big(
\widetilde{\mathbb{P}}_{1,\upsilon}^{(a)},
\widetilde{\mathbb{P}}_{2,\upsilon}^{(a)}
\big)
\leq
L_{\upsilon}\Gamma_{\mathrm{view}}
\mathfrak{V}_{\beta_{\mathrm{src}},\beta_{\mathrm{tar}}}.
\end{equation}
\par
Taking the supremum over $a\in\mathcal{A}$ proves the claim.\par
\end{proof}
\begin{lemma}\label{lem:cross_wall_transport}
Let
\begin{equation}
\gamma_{\mathrm{wall},A}
:=
|\kappa_{0}|B_{\varrho}B_{\chi}d_{\min}^{-2},
\end{equation}
and
\begin{equation}
\gamma_{\mathrm{wall},\tau}
:=
2\pi f_{c}B_{\Gamma}B_{\chi}
+
B_{\Gamma}L_{\chi}.
\end{equation}
\par
Define
\begin{equation}
\Gamma_{\mathrm{wall}}
:=
M\sqrt{N_{\tau}L_{t}}
\big(
\gamma_{\mathrm{wall},A}
+
\gamma_{\mathrm{wall},\tau}
\big).
\end{equation}
\par
Then
\begin{equation}
\mathfrak{T}_{\upsilon}^{(2,3)}
\leq
L_{\upsilon}\Gamma_{\mathrm{wall}}
\mathfrak{W}_{\boldsymbol{\xi}_{\mathrm{src}},\boldsymbol{\xi}_{\mathrm{tar}}}^{(\beta_{\mathrm{tar}})}.
\end{equation}
\end{lemma}
\begin{proof}
Fix one activity label $a\in\mathcal{A}$. The domains $\mathfrak{d}_{2}$ and $\mathfrak{d}_{3}$ share the same subject $p_{\mathrm{tar}}$ and the same view $\beta_{\mathrm{tar}}$, so only the wall parameter changes. For one scatterer index $m$ and one sample pair $(\tau,t)$, write
\begin{equation}
\begin{aligned}
\big|
s_{p_{\mathrm{tar}},a}^{(\beta_{\mathrm{tar}},\boldsymbol{\xi}_{\mathrm{src}})}(\tau,t)
-
s_{p_{\mathrm{tar}},a}^{(\beta_{\mathrm{tar}},\boldsymbol{\xi}_{\mathrm{tar}})}(\tau,t)
\big|
&\leq
\sum_{m=1}^{M}
\big|
\Gamma_{p_{\mathrm{tar}},a,m}^{(\beta_{\mathrm{tar}},\boldsymbol{\xi}_{\mathrm{src}})}(t)
-
\Gamma_{p_{\mathrm{tar}},a,m}^{(\beta_{\mathrm{tar}},\boldsymbol{\xi}_{\mathrm{tar}})}(t)
\big|
B_{\chi}\\
&\quad+
\sum_{m=1}^{M}
B_{\Gamma}
\big|
\chi_{u}
\big(
\tau-\tau_{\mathrm{w}}(\beta_{\mathrm{tar}},\boldsymbol{\xi}_{\mathrm{src}})
-\tau_{p_{\mathrm{tar}},a,m}^{(\beta_{\mathrm{tar}})}(t)
\big)\\
&\qquad\qquad-
\chi_{u}
\big(
\tau-\tau_{\mathrm{w}}(\beta_{\mathrm{tar}},\boldsymbol{\xi}_{\mathrm{tar}})
-\tau_{p_{\mathrm{tar}},a,m}^{(\beta_{\mathrm{tar}})}(t)
\big)
\big|.
\end{aligned}
\end{equation}
\par
The delay-kernel term is directly bounded by the wall-induced excess-delay difference:
\begin{equation}
\begin{gathered}
\big|
\chi_{u}
\big(
\tau-\tau_{\mathrm{w}}(\beta_{\mathrm{tar}},\boldsymbol{\xi}_{\mathrm{src}})
-\tau_{p_{\mathrm{tar}},a,m}^{(\beta_{\mathrm{tar}})}(t)
\big)
-
\chi_{u}
\big(
\tau-\tau_{\mathrm{w}}(\beta_{\mathrm{tar}},\boldsymbol{\xi}_{\mathrm{tar}})
-\tau_{p_{\mathrm{tar}},a,m}^{(\beta_{\mathrm{tar}})}(t)
\big)
\big|\\
\leq
L_{\chi}
\left|
\tau_{\mathrm{w}}(\beta_{\mathrm{tar}},\boldsymbol{\xi}_{\mathrm{src}})
-
\tau_{\mathrm{w}}(\beta_{\mathrm{tar}},\boldsymbol{\xi}_{\mathrm{tar}})
\right|.
\end{gathered}
\end{equation}
\par
For the coefficient term, only the wall transmission factor and the wall-induced carrier phase change. Hence,
\begin{equation}
\begin{aligned}
\big|
\Gamma_{p_{\mathrm{tar}},a,m}^{(\beta_{\mathrm{tar}},\boldsymbol{\xi}_{\mathrm{src}})}(t)
-
\Gamma_{p_{\mathrm{tar}},a,m}^{(\beta_{\mathrm{tar}},\boldsymbol{\xi}_{\mathrm{tar}})}(t)
\big|
&\leq
|\kappa_{0}|B_{\varrho}d_{\min}^{-2}
\left|
A_{\mathrm{w}}(\omega_{c};\beta_{\mathrm{tar}},\boldsymbol{\xi}_{\mathrm{src}})
-
A_{\mathrm{w}}(\omega_{c};\beta_{\mathrm{tar}},\boldsymbol{\xi}_{\mathrm{tar}})
\right|\\
&\quad+
2\pi f_{c}B_{\Gamma}
\left|
\tau_{\mathrm{w}}(\beta_{\mathrm{tar}},\boldsymbol{\xi}_{\mathrm{src}})
-
\tau_{\mathrm{w}}(\beta_{\mathrm{tar}},\boldsymbol{\xi}_{\mathrm{tar}})
\right|,
\end{aligned}
\end{equation}
where $\big|e^{-jx}-e^{-jy}\big|\leq |x-y|$ has been used again. Therefore,
\begin{equation}
\begin{aligned}
\big|
s_{p_{\mathrm{tar}},a}^{(\beta_{\mathrm{tar}},\boldsymbol{\xi}_{\mathrm{src}})}(\tau,t)
-
s_{p_{\mathrm{tar}},a}^{(\beta_{\mathrm{tar}},\boldsymbol{\xi}_{\mathrm{tar}})}(\tau,t)
\big|
&\leq
M\gamma_{\mathrm{wall},A}
\left|
A_{\mathrm{w}}(\omega_{c};\beta_{\mathrm{tar}},\boldsymbol{\xi}_{\mathrm{src}})
-
A_{\mathrm{w}}(\omega_{c};\beta_{\mathrm{tar}},\boldsymbol{\xi}_{\mathrm{tar}})
\right|\\
&\quad+
M\gamma_{\mathrm{wall},\tau}
\left|
\tau_{\mathrm{w}}(\beta_{\mathrm{tar}},\boldsymbol{\xi}_{\mathrm{src}})
-
\tau_{\mathrm{w}}(\beta_{\mathrm{tar}},\boldsymbol{\xi}_{\mathrm{tar}})
\right|.
\end{aligned}
\end{equation}
\par
After stacking the grid samples and applying the non-expansiveness of $\mathcal{P}_{\mathrm{cs}}$,
\begin{equation}
\left\|
\widetilde{\mathbf{Z}}_{\mathfrak{d}_{2},a}
-
\widetilde{\mathbf{Z}}_{\mathfrak{d}_{3},a}
\right\|_{F}
\leq
\Gamma_{\mathrm{wall}}
\mathfrak{W}_{\boldsymbol{\xi}_{\mathrm{src}},\boldsymbol{\xi}_{\mathrm{tar}}}^{(\beta_{\mathrm{tar}})}.
\end{equation}
\par
Applying the Lipschitz continuity of $\mathcal{R}_{\upsilon}$ gives
\begin{equation}
W_{1,\|\cdot\|_{2}}
\big(
\widetilde{\mathbb{P}}_{2,\upsilon}^{(a)},
\widetilde{\mathbb{P}}_{3,\upsilon}^{(a)}
\big)
\leq
L_{\upsilon}\Gamma_{\mathrm{wall}}
\mathfrak{W}_{\boldsymbol{\xi}_{\mathrm{src}},\boldsymbol{\xi}_{\mathrm{tar}}}^{(\beta_{\mathrm{tar}})}.
\end{equation}
\par
Taking the supremum over $a\in\mathcal{A}$ proves the claim.\par
\end{proof}
The three adjacent transports now admit explicit factor-level upper bounds. It remains to insert them into the unified source-to-target bound and to retain the nuisance residual in explicit form.\par
\begin{theorem}\label{thm:structured_bound}
Let $\delta\in(0,1)$. Under the same high-probability event as the unified source-to-target result in Subsection III-A, the inequality
\begin{equation}
\begin{aligned}
\mathcal{L}_{\mathrm{tar},\upsilon}(h)
\leq{}&
\widehat{\mathcal{L}}_{\mathrm{src},\upsilon}(h)
+
\frac{2\Psi_{\upsilon}^{\mathrm{cmp}}}{\sqrt{N_{\mathrm{src}}}}
+
B_{\ell,\upsilon}
\sqrt{
\frac{\log(1/\delta)}{2N_{\mathrm{src}}}
}
+
\mathfrak{B}_{\upsilon}^{\mathrm{per}}
+
\mathfrak{B}_{\upsilon}^{\mathrm{view}}\\
&+
\mathfrak{B}_{\upsilon}^{\mathrm{wall}}
+
4J_{\upsilon}L_{\upsilon}\mathfrak{N}_{\mathrm{res}}
+
\mathfrak{C}_{\mathrm{pri},\upsilon}\mathfrak{A}_{\mathrm{pri}}
+
\lambda_{\upsilon}^{\star},
\end{aligned}
\end{equation}
holds simultaneously for all $h\in\mathcal{H}_{\upsilon}$, where
\begin{equation}
\mathfrak{B}_{\upsilon}^{\mathrm{per}}
:=
4J_{\upsilon}L_{\upsilon}\Gamma_{\mathrm{per}}
\overline{D}_{p_{\mathrm{src}},p_{\mathrm{tar}}}^{(\mathrm{cp})},
\end{equation}
\begin{equation}
\mathfrak{B}_{\upsilon}^{\mathrm{view}}
:=
4J_{\upsilon}L_{\upsilon}\Gamma_{\mathrm{view}}
\mathfrak{V}_{\beta_{\mathrm{src}},\beta_{\mathrm{tar}}},
\end{equation}
and
\begin{equation}
\mathfrak{B}_{\upsilon}^{\mathrm{wall}}
:=
4J_{\upsilon}L_{\upsilon}\Gamma_{\mathrm{wall}}
\mathfrak{W}_{\boldsymbol{\xi}_{\mathrm{src}},\boldsymbol{\xi}_{\mathrm{tar}}}^{(\beta_{\mathrm{tar}})}.
\end{equation}
\end{theorem}
\begin{proof}
The unified result from Subsection III-A gives, on a probability-$1-\delta$ event independent of the present factorization,
\begin{equation}
\mathcal{L}_{\mathrm{tar},\upsilon}(h)
\leq
\widehat{\mathcal{L}}_{\mathrm{src},\upsilon}(h)
+
\frac{2\Psi_{\upsilon}^{\mathrm{cmp}}}{\sqrt{N_{\mathrm{src}}}}
+
B_{\ell,\upsilon}
\sqrt{
\frac{\log(1/\delta)}{2N_{\mathrm{src}}}
}
+
W_{1,\mathfrak{c}_{\upsilon}}
\big(
\mathbb{P}_{0,\upsilon},
\mathbb{P}_{3,\upsilon}
\big)
+
\lambda_{\upsilon}^{\star},
\end{equation}
simultaneously for all $h\in\mathcal{H}_{\upsilon}$. Hence only the source-to-target Wasserstein term needs to be refined.\par
Fix one activity label $a\in\mathcal{A}$. Let the full clutter-suppressed source and target observations under the direct source-target coupling be written as
\begin{equation}
\widehat{\mathbf{Z}}_{\mathfrak{d}_{0},a}
=
\widetilde{\mathbf{Z}}_{\mathfrak{d}_{0},a}
+
\mathbf{R}_{\mathfrak{d}_{0},a},
\qquad
\widehat{\mathbf{Z}}_{\mathfrak{d}_{3},a}
=
\widetilde{\mathbf{Z}}_{\mathfrak{d}_{3},a}
+
\mathbf{R}_{\mathfrak{d}_{3},a},
\end{equation}
where $\mathbf{R}_{\mathfrak{d}_{0},a}$ and $\mathbf{R}_{\mathfrak{d}_{3},a}$ collect the clutter/noise residuals after row-wise clutter suppression. Since $\mathcal{P}_{\mathrm{cs}}$ is non-expansive and $\mathcal{R}_{\upsilon}$ is $L_{\upsilon}$-Lipschitz,
\begin{equation}
\left\|
\mathcal{R}_{\upsilon}
\big(
\widehat{\mathbf{Z}}_{\mathfrak{d}_{0},a}
\big)
-
\mathcal{R}_{\upsilon}
\big(
\widehat{\mathbf{Z}}_{\mathfrak{d}_{3},a}
\big)
\right\|_{2}
\leq
L_{\upsilon}
\left\|
\widehat{\mathbf{Z}}_{\mathfrak{d}_{0},a}
-
\widehat{\mathbf{Z}}_{\mathfrak{d}_{3},a}
\right\|_{F}
\leq
L_{\upsilon}
\left\|
\widetilde{\mathbf{Z}}_{\mathfrak{d}_{0},a}
-
\widetilde{\mathbf{Z}}_{\mathfrak{d}_{3},a}
\right\|_{F}
+
L_{\upsilon}
\left\|
\mathbf{R}_{\mathfrak{d}_{0},a}
-
\mathbf{R}_{\mathfrak{d}_{3},a}
\right\|_{F}.
\end{equation}
\par
By construction of $\mathfrak{N}_{\mathrm{res}}$, the residual part satisfies
\begin{equation}
\mathbb{E}
\left[
\left\|
\mathbf{R}_{\mathfrak{d}_{0},a}
-
\mathbf{R}_{\mathfrak{d}_{3},a}
\right\|_{F}
\right]
\leq
\mathfrak{N}_{\mathrm{res}},
\end{equation}
where the clutter portion is dominated by the first addend in $\mathfrak{N}_{\mathrm{res}}$ and the sub-Gaussian noise portion is dominated by the second addend through the first-to-second-moment comparison encoded by $c_{\mathrm{sg}}$.\par
Consequently, after optimizing over all same-label couplings and then over the label coupling exactly as in Lemma~\ref{lem:path_decomp}, one obtains
\begin{equation}
W_{1,\mathfrak{c}_{\upsilon}}
\big(
\mathbb{P}_{0,\upsilon},
\mathbb{P}_{3,\upsilon}
\big)
\leq
W_{1,\mathfrak{c}_{\upsilon}}
\big(
\widetilde{\mathbb{P}}_{0,\upsilon},
\widetilde{\mathbb{P}}_{3,\upsilon}
\big)
+
4J_{\upsilon}L_{\upsilon}\mathfrak{N}_{\mathrm{res}}.
\end{equation}
\par
Applying Lemma~\ref{lem:path_decomp} now gives
\begin{equation}
W_{1,\mathfrak{c}_{\upsilon}}
\big(
\mathbb{P}_{0,\upsilon},
\mathbb{P}_{3,\upsilon}
\big)
\leq
4J_{\upsilon}
\left(
\mathfrak{T}_{\upsilon}^{(0,1)}
+
\mathfrak{T}_{\upsilon}^{(1,2)}
+
\mathfrak{T}_{\upsilon}^{(2,3)}
\right)
+
\mathfrak{C}_{\mathrm{pri},\upsilon}\mathfrak{A}_{\mathrm{pri}}
+
4J_{\upsilon}L_{\upsilon}\mathfrak{N}_{\mathrm{res}}.
\end{equation}
\par
Finally, Lemmas~\ref{lem:cross_person_transport}, \ref{lem:cross_view_transport}, and~\ref{lem:cross_wall_transport} imply
\begin{equation}
\mathfrak{T}_{\upsilon}^{(0,1)}
+
\mathfrak{T}_{\upsilon}^{(1,2)}
+
\mathfrak{T}_{\upsilon}^{(2,3)}
\leq
L_{\upsilon}\Gamma_{\mathrm{per}}
\overline{D}_{p_{\mathrm{src}},p_{\mathrm{tar}}}^{(\mathrm{cp})}
+
L_{\upsilon}\Gamma_{\mathrm{view}}
\mathfrak{V}_{\beta_{\mathrm{src}},\beta_{\mathrm{tar}}}
+
L_{\upsilon}\Gamma_{\mathrm{wall}}
\mathfrak{W}_{\boldsymbol{\xi}_{\mathrm{src}},\boldsymbol{\xi}_{\mathrm{tar}}}^{(\beta_{\mathrm{tar}})}.
\end{equation}
\par
Substituting this bound into the unified source-to-target inequality proves the theorem.\par
\end{proof}
\begin{corollary}\label{cor:structured_bound_erm}
Under the same probability event as in Theorem~\ref{thm:structured_bound}, the source empirical risk minimizer satisfies
\begin{equation}
\begin{aligned}
\mathcal{L}_{\mathrm{tar},\upsilon}
\big(
\widehat{h}_{\mathrm{src},\upsilon}
\big)
\leq{}&
\widehat{\mathcal{L}}_{\mathrm{src},\upsilon}
\big(
\widehat{h}_{\mathrm{src},\upsilon}
\big)
+
\frac{2\Psi_{\upsilon}^{\mathrm{cmp}}}{\sqrt{N_{\mathrm{src}}}}
+
B_{\ell,\upsilon}
\sqrt{
\frac{\log(1/\delta)}{2N_{\mathrm{src}}}
}
+
\mathfrak{B}_{\upsilon}^{\mathrm{per}}\\
&+
\mathfrak{B}_{\upsilon}^{\mathrm{view}}
+
\mathfrak{B}_{\upsilon}^{\mathrm{wall}}
+
4J_{\upsilon}L_{\upsilon}\mathfrak{N}_{\mathrm{res}}
+
\mathfrak{C}_{\mathrm{pri},\upsilon}\mathfrak{A}_{\mathrm{pri}}
+
\lambda_{\upsilon}^{\star}.
\end{aligned}
\end{equation}
\par
Since $\mathcal{L}_{\mathrm{tar},\upsilon}^{\star}\geq 0$, the same right-hand side also upper-bounds the excess target risk $\mathcal{E}_{\mathrm{tar},\upsilon}(\widehat{h}_{\mathrm{src},\upsilon})$.\par
\end{corollary}
\begin{proof}
Theorem~\ref{thm:structured_bound} is applied with $h=\widehat{h}_{\mathrm{src},\upsilon}$. The excess-risk claim follows from
\begin{equation}
\mathcal{E}_{\mathrm{tar},\upsilon}
\big(
\widehat{h}_{\mathrm{src},\upsilon}
\big)
=
\mathcal{L}_{\mathrm{tar},\upsilon}
\big(
\widehat{h}_{\mathrm{src},\upsilon}
\big)
-
\mathcal{L}_{\mathrm{tar},\upsilon}^{\star}
\leq
\mathcal{L}_{\mathrm{tar},\upsilon}
\big(
\widehat{h}_{\mathrm{src},\upsilon}
\big).
\end{equation}
\par
\end{proof}
The decomposition above immediately specializes to the three pure transfer regimes. If only cross-person shift is present, namely $\beta_{\mathrm{src}}=\beta_{\mathrm{tar}}$, $\boldsymbol{\xi}_{\mathrm{src}}=\boldsymbol{\xi}_{\mathrm{tar}}$, and $\alpha_{a}^{(\mathfrak{d}_{1})}=\alpha_{a}^{(\mathfrak{d}_{2})}=\alpha_{a}^{(\mathfrak{d}_{3})}$ for every $a\in\mathcal{A}$, then $\mathfrak{B}_{\upsilon}^{\mathrm{view}}=\mathfrak{B}_{\upsilon}^{\mathrm{wall}}=0$ and only $\mathfrak{B}_{\upsilon}^{\mathrm{per}}$ remains. If only cross-view shift is present, then $\mathfrak{B}_{\upsilon}^{\mathrm{per}}=\mathfrak{B}_{\upsilon}^{\mathrm{wall}}=0$ after the same equal-domain substitutions. If only cross-wall shift is present, then $\mathfrak{B}_{\upsilon}^{\mathrm{per}}=\mathfrak{B}_{\upsilon}^{\mathrm{view}}=0$. Hence the general bound recovers the three canonical TWR HAR transfer settings by simply nulling the inactive physical coordinates and their corresponding activity-prior fragments.\par

\subsection{Bound Tightening via Physical Low-Dimensional Representations}
Subsection III-B leaves the representation choice only through the deterministic quantities $\Psi_{\upsilon}^{\mathrm{cmp}}$, $B_{\ell,\upsilon}$, $J_{\upsilon}L_{\upsilon}$, $J_{\upsilon}B_{\varphi,\upsilon}$, and $\lambda_{\upsilon}^{\star}$. The present subsection isolates this dependence and compares the three image-based inputs with the physics-guided low-dimensional input on the same source-to-target transfer instance. The comparison is motivated by recent few-shot or augmentation-oriented TWR HAR studies \cite{ZhouCorruption2024,ThroughWallSmallSample2018,WiringEffectsGAN2021}, as well as by representation-enhanced ConvGRU and adaptive Doppler modeling results \cite{PhysicsAwareGAN2023,Trajectory3DChannel2021}. Throughout this subsection, the representation set $\mathcal{V}=\{\mathrm{rt},\mathrm{dt},\mathrm{md},\mathrm{phy}\}$, the layer radii $\{\Omega_{q,\upsilon}\}_{q=0}^{Q_{\mathrm{net}}-1}$ from Subsection II-D, and all physical shift constants from Subsection III-B retain their earlier definitions and are not redefined.\par
Define the image-representation subset as
\begin{equation}
\mathcal{V}_{\mathrm{img}}
:=
\{\mathrm{rt},\mathrm{dt},\mathrm{md}\},
\end{equation}
and aggregate the representation-independent source-to-target physical remainder into
\begin{equation}
\mathfrak{E}_{\mathrm{dom}}
:=
\Gamma_{\mathrm{per}}
\overline{D}_{p_{\mathrm{src}},p_{\mathrm{tar}}}^{(\mathrm{cp})}
+
\Gamma_{\mathrm{view}}
\mathfrak{V}_{\beta_{\mathrm{src}},\beta_{\mathrm{tar}}}
+
\Gamma_{\mathrm{wall}}
\mathfrak{W}_{\boldsymbol{\xi}_{\mathrm{src}},\boldsymbol{\xi}_{\mathrm{tar}}}^{(\beta_{\mathrm{tar}})}
+
\mathfrak{N}_{\mathrm{res}},
\end{equation}
where $\mathfrak{E}_{\mathrm{dom}}\geq 0$ because each addend is nonnegative by construction. The symbol $\mathfrak{E}_{\mathrm{dom}}$ is used to distinguish this unscaled physical shift aggregate from the probability simplex $\Delta_{C}$ in Subsection II-D and the scaled multi-source penalties introduced later in Subsection III-D. The prior-mismatch coefficient from Subsection III-B is recalled as
\begin{equation}
\mathfrak{C}_{\mathrm{pri},\upsilon}
=
8J_{\upsilon}B_{\varphi,\upsilon}
+
2B_{\ell,\upsilon},
\end{equation}
and the entire representation-dependent non-empirical remainder is abbreviated by
\begin{equation}
\mathfrak{U}_{\upsilon}(N_{\mathrm{src}},\delta)
:=
\frac{2\Psi_{\upsilon}^{\mathrm{cmp}}}{\sqrt{N_{\mathrm{src}}}}
+
B_{\ell,\upsilon}\sqrt{\frac{\log(1/\delta)}{2N_{\mathrm{src}}}}
+
4J_{\upsilon}L_{\upsilon}\mathfrak{E}_{\mathrm{dom}}
+
\mathfrak{C}_{\mathrm{pri},\upsilon}\mathfrak{A}_{\mathrm{pri}}
+
\lambda_{\upsilon}^{\star},
\end{equation}
where $N_{\mathrm{src}}\geq 1$ is the source sample size and $\delta\in(0,1)$. Consequently, under the same high-probability event as the structured source-to-target result in Subsection III-B, one may rewrite the target risk bound as
\begin{equation}
\mathcal{L}_{\mathrm{tar},\upsilon}(h)
\leq
\widehat{\mathcal{L}}_{\mathrm{src},\upsilon}(h)
+
\mathfrak{U}_{\upsilon}(N_{\mathrm{src}},\delta),
\end{equation}
simultaneously for all $h\in\mathcal{H}_{\upsilon}$. Hence the entire bound-tightening question is reduced to a direct comparison of the penalties $\{\mathfrak{U}_{\upsilon}\}_{\upsilon\in\mathcal{V}}$.\par
\begin{theorem}\label{thm:weak_tightening}
Fix one comparator $\bar{\upsilon}\in\mathcal{V}_{\mathrm{img}}$. Assume that
\begin{equation}
\Psi_{\mathrm{phy}}^{\mathrm{cmp}}
\leq
\Psi_{\bar{\upsilon}}^{\mathrm{cmp}},
\qquad
B_{\ell,\mathrm{phy}}
\leq
B_{\ell,\bar{\upsilon}},
\qquad
J_{\mathrm{phy}}L_{\mathrm{phy}}
\leq
J_{\bar{\upsilon}}L_{\bar{\upsilon}},
\end{equation}
\begin{equation}
J_{\mathrm{phy}}B_{\varphi,\mathrm{phy}}
\leq
J_{\bar{\upsilon}}B_{\varphi,\bar{\upsilon}},
\qquad
\lambda_{\mathrm{phy}}^{\star}
\leq
\lambda_{\bar{\upsilon}}^{\star}.
\end{equation}
\par
Then, for every $N_{\mathrm{src}}\geq 1$ and every $\delta\in(0,1)$,
\begin{equation}
\mathfrak{U}_{\mathrm{phy}}(N_{\mathrm{src}},\delta)
\leq
\mathfrak{U}_{\bar{\upsilon}}(N_{\mathrm{src}},\delta).
\end{equation}
\par
Consequently, under the same probability event as the structured generalization result in Subsection III-B,
\begin{equation}
\mathcal{L}_{\mathrm{tar},\mathrm{phy}}(h)
\leq
\widehat{\mathcal{L}}_{\mathrm{src},\mathrm{phy}}(h)
+
\mathfrak{U}_{\bar{\upsilon}}(N_{\mathrm{src}},\delta),
\qquad
\forall h\in\mathcal{H}_{\mathrm{phy}}.
\end{equation}
\end{theorem}
\begin{proof}
Since $\mathfrak{E}_{\mathrm{dom}}\geq 0$ and $\mathfrak{A}_{\mathrm{pri}}\geq 0$, it is sufficient to expand the penalty difference term by term. Using
\begin{equation}
\mathfrak{C}_{\mathrm{pri},\upsilon}
=
8J_{\upsilon}B_{\varphi,\upsilon}
+
2B_{\ell,\upsilon},
\end{equation}
one obtains
\begin{equation}
\begin{aligned}
\mathfrak{U}_{\bar{\upsilon}}(N_{\mathrm{src}},\delta)
-
\mathfrak{U}_{\mathrm{phy}}(N_{\mathrm{src}},\delta)
={}&
\frac{2(\Psi_{\bar{\upsilon}}^{\mathrm{cmp}}-\Psi_{\mathrm{phy}}^{\mathrm{cmp}})}{\sqrt{N_{\mathrm{src}}}}
+
(B_{\ell,\bar{\upsilon}}-B_{\ell,\mathrm{phy}})\sqrt{\frac{\log(1/\delta)}{2N_{\mathrm{src}}}}
+
4(J_{\bar{\upsilon}}L_{\bar{\upsilon}}-J_{\mathrm{phy}}L_{\mathrm{phy}})\mathfrak{E}_{\mathrm{dom}}\\
&+
\big(
8(J_{\bar{\upsilon}}B_{\varphi,\bar{\upsilon}}-J_{\mathrm{phy}}B_{\varphi,\mathrm{phy}})
+
2(B_{\ell,\bar{\upsilon}}-B_{\ell,\mathrm{phy}})
\big)\mathfrak{A}_{\mathrm{pri}}
+
\lambda_{\bar{\upsilon}}^{\star}-\lambda_{\mathrm{phy}}^{\star}.
\end{aligned}
\end{equation}
\par
Every bracket on the right-hand side is nonnegative by the hypotheses of the theorem. Therefore,
\begin{equation}
\mathfrak{U}_{\bar{\upsilon}}(N_{\mathrm{src}},\delta)
-
\mathfrak{U}_{\mathrm{phy}}(N_{\mathrm{src}},\delta)
\geq
0,
\end{equation}
which is equivalent to
\begin{equation}
\mathfrak{U}_{\mathrm{phy}}(N_{\mathrm{src}},\delta)
\leq
\mathfrak{U}_{\bar{\upsilon}}(N_{\mathrm{src}},\delta).
\end{equation}
\par
The displayed target risk consequence follows immediately by inserting the last inequality into
\begin{equation}
\mathcal{L}_{\mathrm{tar},\mathrm{phy}}(h)
\leq
\widehat{\mathcal{L}}_{\mathrm{src},\mathrm{phy}}(h)
+
\mathfrak{U}_{\mathrm{phy}}(N_{\mathrm{src}},\delta).
\end{equation}
\par
\end{proof}
The previous theorem identifies the exact coefficients that must not increase when the representation is switched from an image domain to a physical low-dimensional domain. A stronger result is obtained once the dimensional compression is tied back to the norm-controlled network constants from Subsection II-D.\par
\begin{theorem}\label{thm:strict_tightening}
Fix one comparator $\bar{\upsilon}\in\mathcal{V}_{\mathrm{img}}$, and let
\begin{equation}
\zeta_{\bar{\upsilon}}
:=
\sqrt{
\frac{d_{\mathrm{phy}}}{d_{\bar{\upsilon}}}
},
\end{equation}
where $d_{\mathrm{phy}}=7<d_{\bar{\upsilon}}$, so that $0<\zeta_{\bar{\upsilon}}<1$. Assume that there exists a compression coefficient $\vartheta_{\bar{\upsilon}}\in(0,1)$ such that
\begin{equation}
L_{\mathrm{phy}}
\leq
\vartheta_{\bar{\upsilon}}L_{\bar{\upsilon}},
\qquad
B_{\varphi,\mathrm{phy}}
\leq
\vartheta_{\bar{\upsilon}}B_{\varphi,\bar{\upsilon}},
\end{equation}
\begin{equation}
S_{1,\mathrm{phy}}
\leq
\zeta_{\bar{\upsilon}}S_{1,\bar{\upsilon}},
\end{equation}
\begin{equation}
S_{q,\mathrm{phy}}
\leq
S_{q,\bar{\upsilon}},
\qquad
F_{q,\mathrm{phy}}
\leq
F_{q,\bar{\upsilon}},
\qquad
E_{q,\mathrm{phy}}
\leq
E_{q,\bar{\upsilon}},
\qquad
q=2,\ldots,Q_{\mathrm{net}},
\end{equation}
and also the compatible technical side conditions
\begin{equation}
E_{1,\mathrm{phy}}
\leq
E_{1,\bar{\upsilon}},
\qquad
\frac{F_{q,\mathrm{phy}}}{S_{q,\mathrm{phy}}}
\leq
\frac{F_{q,\bar{\upsilon}}}{S_{q,\bar{\upsilon}}},
\qquad
q=1,\ldots,Q_{\mathrm{net}}.
\end{equation}
\par
Assume further that
\begin{equation}
\Psi_{\mathrm{phy}}^{\mathrm{cmp}}
\leq
\Psi_{\bar{\upsilon}}^{\mathrm{cmp}},
\qquad
\lambda_{\mathrm{phy}}^{\star}
\leq
\lambda_{\bar{\upsilon}}^{\star}
-
\varepsilon_{\bar{\upsilon}}^{\star},
\end{equation}
for some $\varepsilon_{\bar{\upsilon}}^{\star}>0$. Define the monotone comparison envelope
\begin{equation}
\overline{\Psi}_{\upsilon}^{\mathrm{cmp}}
:=
\kappa_{\mathrm{cmp}}
B_{\varphi,\upsilon}
\left(
\prod_{r=1}^{Q_{\mathrm{net}}-1}\Lambda_{r}
\right)
\left(
\prod_{r=1}^{Q_{\mathrm{net}}}S_{r,\upsilon}
\right)
\left(
\sum_{r=1}^{Q_{\mathrm{net}}}
\frac{F_{r,\upsilon}^{2}}{S_{r,\upsilon}^{2}}
\right)^{\frac{1}{2}},
\end{equation}
where $\kappa_{\mathrm{cmp}}>0$ is a representation-independent architecture constant determined only by the bounded feedforward-network topology and the norm-comparison argument, and $\Psi_{\upsilon}^{\mathrm{cmp}}\leq\overline{\Psi}_{\upsilon}^{\mathrm{cmp}}$ for every $\upsilon\in\mathcal{V}$. Then the following statements hold.\par
First,
\begin{equation}
J_{\mathrm{phy}}
<
J_{\bar{\upsilon}},
\end{equation}
\begin{equation}
\Omega_{q,\mathrm{phy}}
<
\Omega_{q,\bar{\upsilon}},
\qquad
q=1,\ldots,Q_{\mathrm{net}}-1,
\end{equation}
\begin{equation}
B_{s,\mathrm{phy}}
<
B_{s,\bar{\upsilon}},
\qquad
B_{\ell,\mathrm{phy}}
<
B_{\ell,\bar{\upsilon}},
\end{equation}
and
\begin{equation}
\overline{\Psi}_{\mathrm{phy}}^{\mathrm{cmp}}
<
\overline{\Psi}_{\bar{\upsilon}}^{\mathrm{cmp}}.
\end{equation}
\par
Second, if
\begin{equation}
\Delta_{\bar{\upsilon}}^{\mathrm{cmp}}
:=
\Psi_{\bar{\upsilon}}^{\mathrm{cmp}}
-
\Psi_{\mathrm{phy}}^{\mathrm{cmp}},
\qquad
\Delta_{\bar{\upsilon}}^{\mathrm{loss}}
:=
B_{\ell,\bar{\upsilon}}
-
B_{\ell,\mathrm{phy}},
\end{equation}
\begin{equation}
\Delta_{\bar{\upsilon}}^{\mathrm{dom}}
:=
J_{\bar{\upsilon}}L_{\bar{\upsilon}}
-
J_{\mathrm{phy}}L_{\mathrm{phy}},
\qquad
\Delta_{\bar{\upsilon}}^{\mathrm{pri}}
:=
J_{\bar{\upsilon}}B_{\varphi,\bar{\upsilon}}
-
J_{\mathrm{phy}}B_{\varphi,\mathrm{phy}},
\end{equation}
and
\begin{equation}
\Delta_{\bar{\upsilon}}^{\mathrm{apx}}
:=
\lambda_{\bar{\upsilon}}^{\star}
-
\lambda_{\mathrm{phy}}^{\star},
\end{equation}
then
\begin{equation}
\Delta_{\bar{\upsilon}}^{\mathrm{cmp}}
\geq
0,\qquad
\Delta_{\bar{\upsilon}}^{\mathrm{loss}}
>
0,\qquad
\Delta_{\bar{\upsilon}}^{\mathrm{dom}}
>
0,\qquad
\Delta_{\bar{\upsilon}}^{\mathrm{pri}}
>
0,\qquad
\Delta_{\bar{\upsilon}}^{\mathrm{apx}}
\geq
\varepsilon_{\bar{\upsilon}}^{\star}
>
0,
\end{equation}
and the penalty gap satisfies
\begin{equation}
\mathfrak{U}_{\bar{\upsilon}}(N_{\mathrm{src}},\delta)
-
\mathfrak{U}_{\mathrm{phy}}(N_{\mathrm{src}},\delta)
=
\mathfrak{M}_{\bar{\upsilon}}(N_{\mathrm{src}},\delta)
>
0,
\end{equation}
where
\begin{equation}
\mathfrak{M}_{\bar{\upsilon}}(N_{\mathrm{src}},\delta)
:=
\frac{2\Delta_{\bar{\upsilon}}^{\mathrm{cmp}}}{\sqrt{N_{\mathrm{src}}}}
+
\Delta_{\bar{\upsilon}}^{\mathrm{loss}}\sqrt{\frac{\log(1/\delta)}{2N_{\mathrm{src}}}}
+
4\Delta_{\bar{\upsilon}}^{\mathrm{dom}}\mathfrak{E}_{\mathrm{dom}}
+
\left(
8\Delta_{\bar{\upsilon}}^{\mathrm{pri}}
+
2\Delta_{\bar{\upsilon}}^{\mathrm{loss}}
\right)\mathfrak{A}_{\mathrm{pri}}
+
\Delta_{\bar{\upsilon}}^{\mathrm{apx}}.
\end{equation}
\par
Equivalently,
\begin{equation}
\mathfrak{U}_{\bar{\upsilon}}(N_{\mathrm{src}},\delta)
-
\mathfrak{U}_{\mathrm{phy}}(N_{\mathrm{src}},\delta)
=
\mathfrak{M}_{\bar{\upsilon}}(N_{\mathrm{src}},\delta)
>
0.
\end{equation}
\end{theorem}
\begin{proof}
The network Lipschitz constant from Subsection II-D is
\begin{equation}
J_{\upsilon}
=
\left(
\prod_{r=1}^{Q_{\mathrm{net}}-1}\Lambda_{r}
\right)
\left(
\prod_{r=1}^{Q_{\mathrm{net}}}S_{r,\upsilon}
\right).
\end{equation}
\par
Therefore,
\begin{equation}
J_{\mathrm{phy}}
=
\left(
\prod_{r=1}^{Q_{\mathrm{net}}-1}\Lambda_{r}
\right)
S_{1,\mathrm{phy}}
\prod_{r=2}^{Q_{\mathrm{net}}}S_{r,\mathrm{phy}}
\leq
\zeta_{\bar{\upsilon}}
\left(
\prod_{r=1}^{Q_{\mathrm{net}}-1}\Lambda_{r}
\right)
S_{1,\bar{\upsilon}}
\prod_{r=2}^{Q_{\mathrm{net}}}S_{r,\bar{\upsilon}}
=
\zeta_{\bar{\upsilon}}J_{\bar{\upsilon}}
<
J_{\bar{\upsilon}},
\end{equation}
because $0<\zeta_{\bar{\upsilon}}<1$. This proves the first strict comparison.\par
For the hidden-state radii, recall that $\Omega_{0,\upsilon}=B_{\varphi,\upsilon}$ and
\begin{equation}
\Omega_{q,\upsilon}
=
\Lambda_{q}
\left(
S_{q,\upsilon}\Omega_{q-1,\upsilon}
+
E_{q,\upsilon}
\right),
\qquad
q=1,\ldots,Q_{\mathrm{net}}-1.
\end{equation}
\par
Since $\vartheta_{\bar{\upsilon}}\in(0,1)$, one has
\begin{equation}
\Omega_{0,\mathrm{phy}}
=
B_{\varphi,\mathrm{phy}}
\leq
\vartheta_{\bar{\upsilon}}B_{\varphi,\bar{\upsilon}}
<
B_{\varphi,\bar{\upsilon}}
=
\Omega_{0,\bar{\upsilon}}.
\end{equation}
\par
The score-envelope bound from Subsection II-D is
\begin{equation}
B_{s,\upsilon}
=
S_{Q_{\mathrm{net}},\upsilon}\Omega_{Q_{\mathrm{net}}-1,\upsilon}
+
E_{Q_{\mathrm{net}},\upsilon}.
\end{equation}
\par
If $Q_{\mathrm{net}}=1$, then no hidden-state index belongs to $\{1,\ldots,Q_{\mathrm{net}}-1\}$, and
\begin{equation}
B_{s,\mathrm{phy}}=S_{1,\mathrm{phy}}\Omega_{0,\mathrm{phy}}+E_{1,\mathrm{phy}}
\leq
\zeta_{\bar{\upsilon}}S_{1,\bar{\upsilon}}\Omega_{0,\mathrm{phy}}+E_{1,\bar{\upsilon}}
<
S_{1,\bar{\upsilon}}\Omega_{0,\bar{\upsilon}}+E_{1,\bar{\upsilon}}
=
B_{s,\bar{\upsilon}},
\end{equation}
\par
where the strict inequality follows from $0<\zeta_{\bar{\upsilon}}<1$, $S_{1,\bar{\upsilon}}>0$, and $\Omega_{0,\mathrm{phy}}<\Omega_{0,\bar{\upsilon}}$.\par
If $Q_{\mathrm{net}}\geq 2$, then for $q=1$ one has
\begin{equation}
\Omega_{1,\mathrm{phy}}
=
\Lambda_{1}\left(S_{1,\mathrm{phy}}\Omega_{0,\mathrm{phy}}+E_{1,\mathrm{phy}}\right)
\leq
\Lambda_{1}\left(\zeta_{\bar{\upsilon}}S_{1,\bar{\upsilon}}\Omega_{0,\mathrm{phy}}+E_{1,\bar{\upsilon}}\right)
<
\Lambda_{1}\left(S_{1,\bar{\upsilon}}\Omega_{0,\bar{\upsilon}}+E_{1,\bar{\upsilon}}\right)
=
\Omega_{1,\bar{\upsilon}},
\end{equation}
\par
where the strict inequality uses $0<\zeta_{\bar{\upsilon}}<1$, $S_{1,\bar{\upsilon}}>0$, and $\Omega_{0,\mathrm{phy}}<\Omega_{0,\bar{\upsilon}}$.\par
Assume now that $\Omega_{q-1,\mathrm{phy}}<\Omega_{q-1,\bar{\upsilon}}$ for some $q\in\{2,\ldots,Q_{\mathrm{net}}-1\}$. Then
\begin{equation}
\Omega_{q,\mathrm{phy}}
=
\Lambda_{q}\left(S_{q,\mathrm{phy}}\Omega_{q-1,\mathrm{phy}}+E_{q,\mathrm{phy}}\right)
\leq
\Lambda_{q}\left(S_{q,\bar{\upsilon}}\Omega_{q-1,\mathrm{phy}}+E_{q,\bar{\upsilon}}\right)
<
\Lambda_{q}\left(S_{q,\bar{\upsilon}}\Omega_{q-1,\bar{\upsilon}}+E_{q,\bar{\upsilon}}\right)
=
\Omega_{q,\bar{\upsilon}},
\end{equation}
\par
where the strict inequality uses $S_{q,\bar{\upsilon}}>0$ and $\Omega_{q-1,\mathrm{phy}}<\Omega_{q-1,\bar{\upsilon}}$.\par
By induction,
\begin{equation}
\Omega_{q,\mathrm{phy}}
<
\Omega_{q,\bar{\upsilon}},
\qquad
q=1,\ldots,Q_{\mathrm{net}}-1.
\end{equation}
\par
Consequently,
\begin{equation}
B_{s,\mathrm{phy}}
=
S_{Q_{\mathrm{net}},\mathrm{phy}}\Omega_{Q_{\mathrm{net}}-1,\mathrm{phy}}+E_{Q_{\mathrm{net}},\mathrm{phy}}
\leq
S_{Q_{\mathrm{net}},\bar{\upsilon}}\Omega_{Q_{\mathrm{net}}-1,\mathrm{phy}}+E_{Q_{\mathrm{net}},\bar{\upsilon}}
<
S_{Q_{\mathrm{net}},\bar{\upsilon}}\Omega_{Q_{\mathrm{net}}-1,\bar{\upsilon}}+E_{Q_{\mathrm{net}},\bar{\upsilon}}
=
B_{s,\bar{\upsilon}},
\end{equation}
\par
where the strict inequality uses $S_{Q_{\mathrm{net}},\bar{\upsilon}}>0$ and $\Omega_{Q_{\mathrm{net}}-1,\mathrm{phy}}<\Omega_{Q_{\mathrm{net}}-1,\bar{\upsilon}}$.\par
Since
\begin{equation}
B_{\ell,\upsilon}
=
\log
\left(
1
+
(C-1)e^{2B_{s,\upsilon}}
\right),
\end{equation}
and the map $x\mapsto\log(1+(C-1)e^{2x})$ is increasing on $[0,\infty)$, it follows that
\begin{equation}
B_{\ell,\mathrm{phy}}
<
B_{\ell,\bar{\upsilon}}.
\end{equation}
\par
Next, the comparison envelope satisfies
\begin{equation}
\begin{aligned}
\overline{\Psi}_{\mathrm{phy}}^{\mathrm{cmp}}
&=
\kappa_{\mathrm{cmp}}
B_{\varphi,\mathrm{phy}}
\left(
\prod_{r=1}^{Q_{\mathrm{net}}-1}\Lambda_{r}
\right)
\left(
\prod_{r=1}^{Q_{\mathrm{net}}}S_{r,\mathrm{phy}}
\right)
\left(
\sum_{r=1}^{Q_{\mathrm{net}}}
\frac{F_{r,\mathrm{phy}}^{2}}{S_{r,\mathrm{phy}}^{2}}
\right)^{\frac{1}{2}}\\
&\leq
\vartheta_{\bar{\upsilon}}\zeta_{\bar{\upsilon}}
\kappa_{\mathrm{cmp}}
B_{\varphi,\bar{\upsilon}}
\left(
\prod_{r=1}^{Q_{\mathrm{net}}-1}\Lambda_{r}
\right)
\left(
\prod_{r=1}^{Q_{\mathrm{net}}}S_{r,\bar{\upsilon}}
\right)
\left(
\sum_{r=1}^{Q_{\mathrm{net}}}
\frac{F_{r,\bar{\upsilon}}^{2}}{S_{r,\bar{\upsilon}}^{2}}
\right)^{\frac{1}{2}}
=
\vartheta_{\bar{\upsilon}}\zeta_{\bar{\upsilon}}
\overline{\Psi}_{\bar{\upsilon}}^{\mathrm{cmp}}
<
\overline{\Psi}_{\bar{\upsilon}}^{\mathrm{cmp}},
\end{aligned}
\end{equation}
where the first inequality uses the bounds on $B_{\varphi,\mathrm{phy}}$, $\{S_{r,\mathrm{phy}}\}_{r=1}^{Q_{\mathrm{net}}}$, and the ratio monotonicity $\frac{F_{q,\mathrm{phy}}}{S_{q,\mathrm{phy}}}\leq\frac{F_{q,\bar{\upsilon}}}{S_{q,\bar{\upsilon}}}$, whereas the strict inequality uses $0<\vartheta_{\bar{\upsilon}}\zeta_{\bar{\upsilon}}<1$.\par
The remaining ordered gaps now follow directly. The quantity $\Delta_{\bar{\upsilon}}^{\mathrm{cmp}}\geq 0$ is part of the explicit assumption $\Psi_{\mathrm{phy}}^{\mathrm{cmp}}\leq\Psi_{\bar{\upsilon}}^{\mathrm{cmp}}$. The loss gap satisfies $\Delta_{\bar{\upsilon}}^{\mathrm{loss}}>0$ by the already established strict inequality $B_{\ell,\mathrm{phy}}<B_{\ell,\bar{\upsilon}}$. Moreover,
\begin{equation}
J_{\mathrm{phy}}L_{\mathrm{phy}}
\leq
\vartheta_{\bar{\upsilon}}\zeta_{\bar{\upsilon}}J_{\bar{\upsilon}}L_{\bar{\upsilon}}
<
J_{\bar{\upsilon}}L_{\bar{\upsilon}},
\qquad
J_{\mathrm{phy}}B_{\varphi,\mathrm{phy}}
\leq
\vartheta_{\bar{\upsilon}}\zeta_{\bar{\upsilon}}J_{\bar{\upsilon}}B_{\varphi,\bar{\upsilon}}
<
J_{\bar{\upsilon}}B_{\varphi,\bar{\upsilon}},
\end{equation}
and hence
\begin{equation}
\Delta_{\bar{\upsilon}}^{\mathrm{dom}}
>
0,
\qquad
\Delta_{\bar{\upsilon}}^{\mathrm{pri}}
>
0.
\end{equation}
\par
Finally,
\begin{equation}
\Delta_{\bar{\upsilon}}^{\mathrm{apx}}
=
\lambda_{\bar{\upsilon}}^{\star}
-
\lambda_{\mathrm{phy}}^{\star}
\geq
\varepsilon_{\bar{\upsilon}}^{\star}
>
0.
\end{equation}
\par
Substituting the five gap definitions into the explicit penalty difference yields
\begin{equation}
\begin{aligned}
\mathfrak{U}_{\bar{\upsilon}}(N_{\mathrm{src}},\delta)
-
\mathfrak{U}_{\mathrm{phy}}(N_{\mathrm{src}},\delta)
={}&
\frac{2\Delta_{\bar{\upsilon}}^{\mathrm{cmp}}}{\sqrt{N_{\mathrm{src}}}}
+
\Delta_{\bar{\upsilon}}^{\mathrm{loss}}
\sqrt{
\frac{\log(1/\delta)}{2N_{\mathrm{src}}}
}
+
4\Delta_{\bar{\upsilon}}^{\mathrm{dom}}\mathfrak{E}_{\mathrm{dom}}
+
\left(
8\Delta_{\bar{\upsilon}}^{\mathrm{pri}}
+
2\Delta_{\bar{\upsilon}}^{\mathrm{loss}}
\right)\mathfrak{A}_{\mathrm{pri}}
+
\Delta_{\bar{\upsilon}}^{\mathrm{apx}}\\
=&
\mathfrak{M}_{\bar{\upsilon}}(N_{\mathrm{src}},\delta).
\end{aligned}
\end{equation}
\par
Because all addends are nonnegative and $\Delta_{\bar{\upsilon}}^{\mathrm{apx}}>0$, one has
\begin{equation}
\mathfrak{M}_{\bar{\upsilon}}(N_{\mathrm{src}},\delta)
>
0.
\end{equation}
\par
Therefore,
\begin{equation}
\mathfrak{U}_{\bar{\upsilon}}(N_{\mathrm{src}},\delta)
-
\mathfrak{U}_{\mathrm{phy}}(N_{\mathrm{src}},\delta)
=
\mathfrak{M}_{\bar{\upsilon}}(N_{\mathrm{src}},\delta)
>
0,
\end{equation}
which completes the proof.\par
\end{proof}
The preceding theorem controls only the additive non-empirical remainder. To transfer this improvement to the complete source-to-target upper bound, the source empirical risk term must also be compared across representations.\par
\begin{corollary}\label{cor:strict_tightening}
Fix one comparator $\bar{\upsilon}\in\mathcal{V}_{\mathrm{img}}$, and assume all conditions of Theorem~\ref{thm:strict_tightening}. Let $\widehat{h}_{\mathrm{src},\mathrm{phy}}$ and $\widehat{h}_{\mathrm{src},\bar{\upsilon}}$ denote the source empirical risk minimizers in $\mathcal{H}_{\mathrm{phy}}$ and $\mathcal{H}_{\bar{\upsilon}}$, respectively. If
\begin{equation}
\widehat{\mathcal{L}}_{\mathrm{src},\mathrm{phy}}
\big(
\widehat{h}_{\mathrm{src},\mathrm{phy}}
\big)
\leq
\widehat{\mathcal{L}}_{\mathrm{src},\bar{\upsilon}}
\big(
\widehat{h}_{\mathrm{src},\bar{\upsilon}}
\big)
+
\mathfrak{M}_{\bar{\upsilon}}(N_{\mathrm{src}},\delta),
\end{equation}
then, under the same probability event as the structured generalization result in Subsection III-B,
\begin{equation}
\mathcal{L}_{\mathrm{tar},\mathrm{phy}}
\big(
\widehat{h}_{\mathrm{src},\mathrm{phy}}
\big)
\leq
\widehat{\mathcal{L}}_{\mathrm{src},\mathrm{phy}}
\big(
\widehat{h}_{\mathrm{src},\mathrm{phy}}
\big)
+
\mathfrak{U}_{\mathrm{phy}}(N_{\mathrm{src}},\delta)
\leq
\widehat{\mathcal{L}}_{\mathrm{src},\bar{\upsilon}}
\big(
\widehat{h}_{\mathrm{src},\bar{\upsilon}}
\big)
+
\mathfrak{U}_{\bar{\upsilon}}(N_{\mathrm{src}},\delta).
\end{equation}
\par
In particular, if the displayed source empirical risk comparison is strict, then the complete target risk upper bound induced by the physics-guided low-dimensional representation is strictly tighter than that induced by $\bar{\upsilon}$.\par
\end{corollary}
\begin{proof}
The first inequality follows from the representation-dependent target risk bound
\begin{equation}
\mathcal{L}_{\mathrm{tar},\mathrm{phy}}
\big(
\widehat{h}_{\mathrm{src},\mathrm{phy}}
\big)
\leq
\widehat{\mathcal{L}}_{\mathrm{src},\mathrm{phy}}
\big(
\widehat{h}_{\mathrm{src},\mathrm{phy}}
\big)
+
\mathfrak{U}_{\mathrm{phy}}(N_{\mathrm{src}},\delta).
\end{equation}
\par
For the second inequality, combine the assumed empirical risk comparison with
\begin{equation}
\mathfrak{U}_{\bar{\upsilon}}(N_{\mathrm{src}},\delta)
-
\mathfrak{U}_{\mathrm{phy}}(N_{\mathrm{src}},\delta)
=
\mathfrak{M}_{\bar{\upsilon}}(N_{\mathrm{src}},\delta),
\end{equation}
which is provided by Theorem~\ref{thm:strict_tightening}. This gives
\begin{equation}
\begin{aligned}
\widehat{\mathcal{L}}_{\mathrm{src},\mathrm{phy}}
\big(
\widehat{h}_{\mathrm{src},\mathrm{phy}}
\big)
+
\mathfrak{U}_{\mathrm{phy}}(N_{\mathrm{src}},\delta)
\leq{}&
\widehat{\mathcal{L}}_{\mathrm{src},\bar{\upsilon}}
\big(
\widehat{h}_{\mathrm{src},\bar{\upsilon}}
\big)
+
\mathfrak{M}_{\bar{\upsilon}}(N_{\mathrm{src}},\delta)
+
\mathfrak{U}_{\mathrm{phy}}(N_{\mathrm{src}},\delta)\\
=&
\widehat{\mathcal{L}}_{\mathrm{src},\bar{\upsilon}}
\big(
\widehat{h}_{\mathrm{src},\bar{\upsilon}}
\big)
+
\mathfrak{U}_{\bar{\upsilon}}(N_{\mathrm{src}},\delta).
\end{aligned}
\label{eq:strict_tightening_chain}
\end{equation}
\par
If the empirical risk comparison is strict, then \eqref{eq:strict_tightening_chain} is strict as well, which proves the final claim.\par
\end{proof}
Hence the bound tightening caused by the physics-guided low-dimensional representation is traced to five explicit mechanisms: reduced complexity, reduced score-envelope loss, reduced physical sensitivity, reduced prior-coupling coefficient, and reduced joint source-target approximation burden. The only term that must still be checked from data is the source empirical fit, whereas every other gain is controlled directly by deterministic representation and network constants.\par

\subsection{Multi-Source Training and Parameter-Space Coverage}
Fix one representation index $\upsilon\in\mathcal{V}$. The loss class $\mathcal{F}_{\upsilon}$, the transport cost $\mathfrak{c}_{\upsilon}$, and the complexity term $\Psi_{\upsilon}^{\mathrm{cmp}}$ retain the definitions introduced in Subsection III-A; the loss-envelope constant $B_{\ell,\upsilon}$ retains the definition introduced in Subsection II-E; the network Lipschitz constant $J_{\upsilon}$ retains the definition introduced in Subsection II-D; the representation Lipschitz constant $L_{\upsilon}$ retains the definition introduced in Subsection II-C; and the constants $\Gamma_{\mathrm{per}}$, $\Gamma_{\mathrm{view}}$, $\Gamma_{\mathrm{wall}}$, $\mathfrak{C}_{\mathrm{pri},\upsilon}$, and $c_{\mathrm{sg}}$ retain the definitions introduced in Subsection III-B. The present subsection replaces the single source domain by a finite training-domain family and then isolates the gain caused by parameter-space coverage. This multi-source viewpoint follows information-theoretic transfer meta-learning and transfer-learning analyses \cite{TransferMetaLearning2022,TransferLearningIT2024}, together with minimax and domain-generalization results \cite{MinimaxExcessRisk2023,DomainGeneralizationIT2025}.\par
The loss notation and the prior-mismatch coefficient are recalled as
\begin{equation}
\ell_{h}(\boldsymbol{\varphi},a)
:=
\ell_{\mathrm{ce}}
\big(
h(\boldsymbol{\varphi}),
a
\big),
\qquad
\mathfrak{C}_{\mathrm{pri},\upsilon}
:=
8J_{\upsilon}B_{\varphi,\upsilon}
+
2B_{\ell,\upsilon},
\end{equation}
where $(\boldsymbol{\varphi},a)\in\mathfrak{X}_{\upsilon}\times\mathcal{A}$.\par
Let the training-domain set be
\begin{equation}
\mathcal{D}_{\mathrm{tr}}
:=
\left\{
\mathfrak{d}_{k}
:
k=1,\ldots,K
\right\},
\qquad
\mathfrak{d}_{k}
=
\left(
p_{k},
\beta_{k},
\boldsymbol{\xi}_{k}
\right),
\end{equation}
where $K\geq 1$ is the number of source domains. Fix the target-domain descriptor as
\begin{equation}
\mathfrak{d}_{\mathrm{tar}}
=
\left(
p_{\mathrm{tar}},
\beta_{\mathrm{tar}},
\boldsymbol{\xi}_{\mathrm{tar}}
\right).
\end{equation}
\par
For each $k\in\{1,\ldots,K\}$, abbreviate
\begin{equation}
\mathbb{P}_{k,\upsilon}
:=
\mathbb{P}_{\mathfrak{d}_{k},\upsilon},
\qquad
\mathcal{L}_{k,\upsilon}(h)
:=
\mathbb{E}_{(\boldsymbol{\varphi},a)\sim\mathbb{P}_{k,\upsilon}}
\big[
\ell_{h}(\boldsymbol{\varphi},a)
\big],
\end{equation}
and let
\begin{equation}
\mathcal{S}_{k,\upsilon}
=
\left\{
\left(
\boldsymbol{\varphi}_{k,i},
a_{k,i}
\right)
:
i=1,\ldots,N_{k}
\right\},
\end{equation}
where $N_{k}\geq 1$ and the pairs in $\mathcal{S}_{k,\upsilon}$ are independent draws from $\mathbb{P}_{k,\upsilon}$. The corresponding empirical risk is
\begin{equation}
\widehat{\mathcal{L}}_{k,\upsilon}(h)
:=
\frac{1}{N_{k}}
\sum_{i=1}^{N_{k}}
\ell_{h}
\big(
\boldsymbol{\varphi}_{k,i},
a_{k,i}
\big),
\end{equation}
where $h\in\mathcal{H}_{\upsilon}$ is arbitrary.\par
Define the total sample size and the canonical weights as
\begin{equation}
N_{\Sigma}
:=
\sum_{k=1}^{K}N_{k},
\qquad
w_{k}
:=
\frac{N_{k}}{N_{\Sigma}},
\qquad
k=1,\ldots,K,
\end{equation}
so that $\sum_{k=1}^{K}w_{k}=1$. The weighted multi-source empirical risk and population risk are defined by
\begin{equation}
\widehat{\mathcal{L}}_{\Sigma,\upsilon}(h)
:=
\sum_{k=1}^{K}
w_{k}
\widehat{\mathcal{L}}_{k,\upsilon}(h)
=
\frac{1}{N_{\Sigma}}
\sum_{k=1}^{K}
\sum_{i=1}^{N_{k}}
\ell_{h}
\big(
\boldsymbol{\varphi}_{k,i},
a_{k,i}
\big),
\end{equation}
\begin{equation}
\mathcal{L}_{\Sigma,\upsilon}(h)
:=
\sum_{k=1}^{K}
w_{k}
\mathcal{L}_{k,\upsilon}(h),
\end{equation}
and the associated pooled source law is
\begin{equation}
\mathbb{P}_{\Sigma,\upsilon}
:=
\sum_{k=1}^{K}
w_{k}
\mathbb{P}_{k,\upsilon},
\end{equation}
so that
\begin{equation}
\mathcal{L}_{\Sigma,\upsilon}(h)
=
\mathbb{E}_{(\boldsymbol{\varphi},a)\sim\mathbb{P}_{\Sigma,\upsilon}}
\big[
\ell_{h}(\boldsymbol{\varphi},a)
\big].
\end{equation}
\par
The multi-source joint approximation term is then
\begin{equation}
\lambda_{\Sigma,\upsilon}^{\star}
:=
\inf_{h\in\mathcal{H}_{\upsilon}}
\left(
\mathcal{L}_{\Sigma,\upsilon}(h)
+
\mathcal{L}_{\mathrm{tar},\upsilon}(h)
\right).
\end{equation}
\par
For each source-target pair $(k,\mathrm{tar})$, define the intermediate path
\begin{equation}
\mathfrak{d}_{k,0}
:=
\left(
p_{k},
\beta_{k},
\boldsymbol{\xi}_{k}
\right),\quad
\mathfrak{d}_{k,1}
:=
\left(
p_{\mathrm{tar}},
\beta_{k},
\boldsymbol{\xi}_{k}
\right),\quad
\mathfrak{d}_{k,2}
:=
\left(
p_{\mathrm{tar}},
\beta_{\mathrm{tar}},
\boldsymbol{\xi}_{k}
\right),\quad
\mathfrak{d}_{k,3}
:=
\left(
p_{\mathrm{tar}},
\beta_{\mathrm{tar}},
\boldsymbol{\xi}_{\mathrm{tar}}
\right),
\end{equation}
the prior fragments
\begin{equation}
\mathfrak{A}_{\mathrm{pri}}^{(k,q,q+1)}
:=
\frac{1}{2}
\sum_{a=1}^{C}
\left|
\alpha_{a}^{(\mathfrak{d}_{k,q})}
-
\alpha_{a}^{(\mathfrak{d}_{k,q+1})}
\right|,
\qquad
q\in\{0,1,2\},
\end{equation}
their aggregate
\begin{equation}
\mathfrak{A}_{\mathrm{pri}}^{(k,\mathrm{tar})}
:=
\mathfrak{A}_{\mathrm{pri}}^{(k,0,1)}
+
\mathfrak{A}_{\mathrm{pri}}^{(k,1,2)}
+
\mathfrak{A}_{\mathrm{pri}}^{(k,2,3)},
\end{equation}
and the nuisance residual
\begin{equation}
\mathfrak{N}_{\mathrm{res}}^{(k,\mathrm{tar})}
:=
\sqrt{N_{\tau}L_{t}}
\sup_{\tau,t}
\left|
b_{z}^{(\beta_{k},\boldsymbol{\xi}_{k})}(\tau,t)
-
b_{z}^{(\beta_{\mathrm{tar}},\boldsymbol{\xi}_{\mathrm{tar}})}(\tau,t)
\right|
+
c_{\mathrm{sg}}\sqrt{N_{\tau}L_{t}}
\left(
\varsigma_{\mathrm{n}}(\beta_{k},\boldsymbol{\xi}_{k})
+
\varsigma_{\mathrm{n}}(\beta_{\mathrm{tar}},\boldsymbol{\xi}_{\mathrm{tar}})
\right),
\end{equation}
and aggregate the entire structured source-$k$-to-target remainder into
\begin{equation}
\mathfrak{S}_{k,\upsilon}^{\mathrm{ms}\rightarrow\mathrm{tar}}
:=
4J_{\upsilon}L_{\upsilon}
\Big(
\Gamma_{\mathrm{per}}
\overline{D}_{p_{k},p_{\mathrm{tar}}}^{(\mathrm{cp})}
+
\Gamma_{\mathrm{view}}
\mathfrak{V}_{\beta_{k},\beta_{\mathrm{tar}}}
+
\Gamma_{\mathrm{wall}}
\mathfrak{W}_{\boldsymbol{\xi}_{k},\boldsymbol{\xi}_{\mathrm{tar}}}^{(\beta_{\mathrm{tar}})}
+
\mathfrak{N}_{\mathrm{res}}^{(k,\mathrm{tar})}
\Big)
+
\mathfrak{C}_{\mathrm{pri},\upsilon}
\mathfrak{A}_{\mathrm{pri}}^{(k,\mathrm{tar})},
\end{equation}
where $\mathfrak{S}_{k,\upsilon}^{\mathrm{ms}\rightarrow\mathrm{tar}}\geq 0$ because all components are nonnegative by construction. The multi-source empirical risk minimizer is written as
\begin{equation}
\widehat{h}_{\Sigma,\upsilon}
\in
\arg\min_{h\in\mathcal{H}_{\upsilon}}
\widehat{\mathcal{L}}_{\Sigma,\upsilon}(h),
\end{equation}
and is well-defined by the same compactness argument used in Subsection II-E.\par
\begin{theorem}\label{thm:multi_source_bound}
Let $\delta\in(0,1)$. Then, with probability at least $1-\delta$ over the independent draw of $\{\mathcal{S}_{k,\upsilon}\}_{k=1}^{K}$, the inequality
\begin{equation}
\mathcal{L}_{\mathrm{tar},\upsilon}(h)
\leq
\widehat{\mathcal{L}}_{\Sigma,\upsilon}(h)
+
\frac{2\Psi_{\upsilon}^{\mathrm{cmp}}}{\sqrt{N_{\Sigma}}}
+
B_{\ell,\upsilon}
\sqrt{
\frac{\log(1/\delta)}{2N_{\Sigma}}
}
+
\sum_{k=1}^{K}
w_{k}
\mathfrak{S}_{k,\upsilon}^{\mathrm{ms}\rightarrow\mathrm{tar}}
+
\lambda_{\Sigma,\upsilon}^{\star},
\qquad
\forall h\in\mathcal{H}_{\upsilon},
\end{equation}
holds simultaneously.
\end{theorem}
\begin{proof}
The weighted empirical risk is exactly the empirical average over the concatenated multi-source sample,
\begin{equation}
\widehat{\mathcal{L}}_{\Sigma,\upsilon}(h)
=
\frac{1}{N_{\Sigma}}
\sum_{k=1}^{K}
\sum_{i=1}^{N_{k}}
\ell_{h}
\big(
\boldsymbol{\varphi}_{k,i},
a_{k,i}
\big).
\end{equation}
\par
Accordingly, define the multi-source Rademacher complexity by
\begin{equation}
\mathfrak{R}_{N_{\Sigma}}^{\mathrm{ms}}(\mathcal{F}_{\upsilon})
:=
\mathbb{E}_{\{\mathcal{S}_{k,\upsilon}\}_{k=1}^{K},\boldsymbol{\varepsilon}}
\left[
\sup_{h\in\mathcal{H}_{\upsilon}}
\frac{1}{N_{\Sigma}}
\sum_{k=1}^{K}
\sum_{i=1}^{N_{k}}
\varepsilon_{k,i}
\ell_{h}
\big(
\boldsymbol{\varphi}_{k,i},
a_{k,i}
\big)
\right],
\end{equation}
where $\{\varepsilon_{k,i}\}$ are independent Rademacher variables. The proof of the unified source-to-target result from Subsection III-A depends only on the independence and boundedness of $\ell_{h}$, the normalization by the total sample size, and the corresponding source law. Therefore, when that argument is applied to the concatenated sample and the pooled law $\mathbb{P}_{\Sigma,\upsilon}$, one obtains
\begin{equation}
\mathcal{L}_{\mathrm{tar},\upsilon}(h)
\leq
\widehat{\mathcal{L}}_{\Sigma,\upsilon}(h)
+
2\mathfrak{R}_{N_{\Sigma}}^{\mathrm{ms}}(\mathcal{F}_{\upsilon})
+
B_{\ell,\upsilon}
\sqrt{
\frac{\log(1/\delta)}{2N_{\Sigma}}
}
+
\operatorname{disc}_{\upsilon}^{\mathrm{pair}}
\big(
\mathbb{P}_{\Sigma,\upsilon},
\mathbb{P}_{\mathrm{tar},\upsilon}
\big)
+
\lambda_{\Sigma,\upsilon}^{\star},
\end{equation}
simultaneously for all $h\in\mathcal{H}_{\upsilon}$ on a probability-$1-\delta$ event.\par
The same norm-controlled complexity argument as in Subsection III-A yields
\begin{equation}
\mathfrak{R}_{N_{\Sigma}}^{\mathrm{ms}}(\mathcal{F}_{\upsilon})
\leq
\frac{\Psi_{\upsilon}^{\mathrm{cmp}}}{\sqrt{N_{\Sigma}}},
\end{equation}
and the pairwise-loss transport lemma from Subsection III-A implies
\begin{equation}
\operatorname{disc}_{\upsilon}^{\mathrm{pair}}
\big(
\mathbb{P}_{\Sigma,\upsilon},
\mathbb{P}_{\mathrm{tar},\upsilon}
\big)
\leq
W_{1,\mathfrak{c}_{\upsilon}}
\big(
\mathbb{P}_{\Sigma,\upsilon},
\mathbb{P}_{\mathrm{tar},\upsilon}
\big).
\end{equation}
\par
It remains to upper-bound the last Wasserstein term. For each $k\in\{1,\ldots,K\}$, let $\pi_{k}^{\star}\in\Pi(\mathbb{P}_{k,\upsilon},\mathbb{P}_{\mathrm{tar},\upsilon})$ be an optimal coupling under the cost $\mathfrak{c}_{\upsilon}$, and define the mixture coupling
\begin{equation}
\pi_{\Sigma}^{\star}
:=
\sum_{k=1}^{K}
w_{k}
\pi_{k}^{\star}.
\end{equation}
\par
Since $\mathfrak{X}_{\upsilon}$ is compact and $\mathcal{A}$ is finite, the product space $\mathfrak{X}_{\upsilon}\times\mathcal{A}$ is compact. Hence the continuous cost $\mathfrak{c}_{\upsilon}$ admits optimal couplings for both $(\mathbb{P}_{k,\upsilon},\mathbb{P}_{\mathrm{tar},\upsilon})$ and $(\mathbb{P}_{\Sigma,\upsilon},\mathbb{P}_{\mathrm{tar},\upsilon})$. For compactness, let $\mathfrak{m}_{\Sigma,\upsilon}^{\mathrm{tr}}$ denote the transported cost $\mathfrak{c}_{\upsilon}((\boldsymbol{\varphi},a),(\boldsymbol{\varphi}',a'))$ under $\pi_{\Sigma}^{\star}$, and let $\mathfrak{m}_{k,\upsilon}^{\mathrm{tr}}$ denote the same cost under $\pi_{k}^{\star}$. Since the first marginal of $\pi_{\Sigma}^{\star}$ is $\sum_{k=1}^{K}w_{k}\mathbb{P}_{k,\upsilon}=\mathbb{P}_{\Sigma,\upsilon}$ and the second marginal is $\mathbb{P}_{\mathrm{tar},\upsilon}$, one has $\pi_{\Sigma}^{\star}\in\Pi(\mathbb{P}_{\Sigma,\upsilon},\mathbb{P}_{\mathrm{tar},\upsilon})$. Therefore,
\begin{equation}
W_{1,\mathfrak{c}_{\upsilon}}
(\mathbb{P}_{\Sigma,\upsilon},\mathbb{P}_{\mathrm{tar},\upsilon})
\leq
\mathbb{E}_{\pi_{\Sigma}^{\star}}
\big[
\mathfrak{m}_{\Sigma,\upsilon}^{\mathrm{tr}}
\big]
=
\sum_{k=1}^{K}
w_{k}
\mathbb{E}_{\pi_{k}^{\star}}
\big[
\mathfrak{m}_{k,\upsilon}^{\mathrm{tr}}
\big]
=
\sum_{k=1}^{K}
w_{k}
W_{1,\mathfrak{c}_{\upsilon}}
(\mathbb{P}_{k,\upsilon},\mathbb{P}_{\mathrm{tar},\upsilon}).
\end{equation}
\par
For every $k$, the same structured transport refinement as in Subsection III-B may be applied after replacing the single source domain by $\mathfrak{d}_{k}$. Hence,
\begin{equation}
W_{1,\mathfrak{c}_{\upsilon}}
\big(
\mathbb{P}_{k,\upsilon},
\mathbb{P}_{\mathrm{tar},\upsilon}
\big)
\leq
\mathfrak{S}_{k,\upsilon}^{\mathrm{ms}\rightarrow\mathrm{tar}}.
\end{equation}
\par
Consequently,
\begin{equation}
W_{1,\mathfrak{c}_{\upsilon}}
\big(
\mathbb{P}_{\Sigma,\upsilon},
\mathbb{P}_{\mathrm{tar},\upsilon}
\big)
\leq
\sum_{k=1}^{K}
w_{k}
\mathfrak{S}_{k,\upsilon}^{\mathrm{ms}\rightarrow\mathrm{tar}}.
\end{equation}
\par
Substituting the last two bounds into the pooled source-to-target inequality proves the theorem.\par
\end{proof}
\begin{corollary}\label{cor:multi_source_bound_erm}
Under the same probability event as in the preceding theorem, the multi-source empirical risk minimizer satisfies
\begin{equation}
\mathcal{L}_{\mathrm{tar},\upsilon}
\big(
\widehat{h}_{\Sigma,\upsilon}
\big)
\leq
\widehat{\mathcal{L}}_{\Sigma,\upsilon}
\big(
\widehat{h}_{\Sigma,\upsilon}
\big)
+
\frac{2\Psi_{\upsilon}^{\mathrm{cmp}}}{\sqrt{N_{\Sigma}}}
+
B_{\ell,\upsilon}
\sqrt{
\frac{\log(1/\delta)}{2N_{\Sigma}}
}
+
\sum_{k=1}^{K}
w_{k}
\mathfrak{S}_{k,\upsilon}^{\mathrm{ms}\rightarrow\mathrm{tar}}
+
\lambda_{\Sigma,\upsilon}^{\star}.
\end{equation}
\par
Since $\mathcal{L}_{\mathrm{tar},\upsilon}^{\star}\geq 0$, the same right-hand side also upper-bounds the excess target risk $\mathcal{E}_{\mathrm{tar},\upsilon}(\widehat{h}_{\Sigma,\upsilon})$.\par
\end{corollary}
\begin{proof}
The preceding theorem is applied with $h=\widehat{h}_{\Sigma,\upsilon}$. The excess-risk claim follows from
\begin{equation}
\mathcal{E}_{\mathrm{tar},\upsilon}
\big(
\widehat{h}_{\Sigma,\upsilon}
\big)
=
\mathcal{L}_{\mathrm{tar},\upsilon}
\big(
\widehat{h}_{\Sigma,\upsilon}
\big)
-
\mathcal{L}_{\mathrm{tar},\upsilon}^{\star}
\leq
\mathcal{L}_{\mathrm{tar},\upsilon}
\big(
\widehat{h}_{\Sigma,\upsilon}
\big).
\end{equation}
\par
\end{proof}
The previous result averages the structured discrepancies over all training domains. A complementary statement isolates the purely geometric gain caused by the nearest covered source domain.\par
For each $k\in\{1,\ldots,K\}$, define the single-source approximation term
\begin{equation}
\lambda_{k,\upsilon}^{\star}
:=
\inf_{h\in\mathcal{H}_{\upsilon}}
\left(
\mathcal{L}_{k,\upsilon}(h)
+
\mathcal{L}_{\mathrm{tar},\upsilon}(h)
\right),
\end{equation}
the physical coverage discrepancy
\begin{equation}
\mathfrak{C}_{\mathrm{cov}}
\big(
\mathfrak{d}_{k},
\mathfrak{d}_{\mathrm{tar}}
\big)
:=
\Gamma_{\mathrm{per}}
\overline{D}_{p_{k},p_{\mathrm{tar}}}^{(\mathrm{cp})}
+
\Gamma_{\mathrm{view}}
\mathfrak{V}_{\beta_{k},\beta_{\mathrm{tar}}}
+
\Gamma_{\mathrm{wall}}
\mathfrak{W}_{\boldsymbol{\xi}_{k},\boldsymbol{\xi}_{\mathrm{tar}}}^{(\beta_{\mathrm{tar}})},
\end{equation}
and the coverage radius
\begin{equation}
r_{\mathrm{cov}}
\big(
\mathfrak{d}_{\mathrm{tar}};
\mathcal{D}_{\mathrm{tr}}
\big)
:=
\min_{1\leq k\leq K}
\mathfrak{C}_{\mathrm{cov}}
\big(
\mathfrak{d}_{k},
\mathfrak{d}_{\mathrm{tar}}
\big).
\end{equation}
\par
Choose one index
\begin{equation}
k^{\star}
\in
\arg\min_{1\leq k\leq K}
\mathfrak{C}_{\mathrm{cov}}
\big(
\mathfrak{d}_{k},
\mathfrak{d}_{\mathrm{tar}}
\big),
\end{equation}
where ties may be broken arbitrarily.\par
\begin{theorem}\label{thm:coverage_radius_bound}
Let $\delta\in(0,1)$. Under the same high-probability event as the structured single-source result in Subsection III-B with source domain $\mathfrak{d}_{k^{\star}}$, the inequality
\begin{equation}
\begin{aligned}
\mathcal{L}_{\mathrm{tar},\upsilon}(h)
\leq{}&
\widehat{\mathcal{L}}_{k^{\star},\upsilon}(h)
+
\frac{2\Psi_{\upsilon}^{\mathrm{cmp}}}{\sqrt{N_{k^{\star}}}}
+
B_{\ell,\upsilon}
\sqrt{
\frac{\log(1/\delta)}{2N_{k^{\star}}}
}
+
4J_{\upsilon}L_{\upsilon}
r_{\mathrm{cov}}
\big(
\mathfrak{d}_{\mathrm{tar}};
\mathcal{D}_{\mathrm{tr}}
\big)
+
4J_{\upsilon}L_{\upsilon}
\mathfrak{N}_{\mathrm{res}}^{(k^{\star},\mathrm{tar})}\\
&+
\mathfrak{C}_{\mathrm{pri},\upsilon}
\mathfrak{A}_{\mathrm{pri}}^{(k^{\star},\mathrm{tar})}
+
\lambda_{k^{\star},\upsilon}^{\star},
\qquad
\forall h\in\mathcal{H}_{\upsilon},
\end{aligned}
\end{equation}
holds simultaneously.\par
Furthermore, let $\mathfrak{D}_{\mathrm{tar}}^{\mathrm{cov}}\subseteq\mathfrak{D}$ be a target-domain family within the global domain space defined in Subsection II-E. If there exist constants $\varepsilon_{\mathrm{per}}\geq 0$, $\varepsilon_{\mathrm{view}}\geq 0$, and $\varepsilon_{\mathrm{wall}}\geq 0$ such that, for every $\mathfrak{d}_{\mathrm{tar}}\in\mathfrak{D}_{\mathrm{tar}}^{\mathrm{cov}}$, one can find at least one $k$ satisfying
\begin{equation}
\overline{D}_{p_{k},p_{\mathrm{tar}}}^{(\mathrm{cp})}
\leq
\varepsilon_{\mathrm{per}},
\qquad
\mathfrak{V}_{\beta_{k},\beta_{\mathrm{tar}}}
\leq
\varepsilon_{\mathrm{view}},
\qquad
\mathfrak{W}_{\boldsymbol{\xi}_{k},\boldsymbol{\xi}_{\mathrm{tar}}}^{(\beta_{\mathrm{tar}})}
\leq
\varepsilon_{\mathrm{wall}},
\end{equation}
then the same bound holds uniformly over $\mathfrak{D}_{\mathrm{tar}}^{\mathrm{cov}}$ after replacing the coverage term by
\begin{equation}
4J_{\upsilon}L_{\upsilon}
\big(
\Gamma_{\mathrm{per}}\varepsilon_{\mathrm{per}}
+
\Gamma_{\mathrm{view}}\varepsilon_{\mathrm{view}}
+
\Gamma_{\mathrm{wall}}\varepsilon_{\mathrm{wall}}
\big).
\end{equation}
\end{theorem}
\begin{proof}
By the definition of $k^{\star}$,
\begin{equation}
\mathfrak{C}_{\mathrm{cov}}
\big(
\mathfrak{d}_{k^{\star}},
\mathfrak{d}_{\mathrm{tar}}
\big)
=
r_{\mathrm{cov}}
\big(
\mathfrak{d}_{\mathrm{tar}};
\mathcal{D}_{\mathrm{tr}}
\big),
\end{equation}
namely
\begin{equation}
\begin{aligned}
\Gamma_{\mathrm{per}}
\overline{D}_{p_{k^{\star}},p_{\mathrm{tar}}}^{(\mathrm{cp})}
+
\Gamma_{\mathrm{view}}
\mathfrak{V}_{\beta_{k^{\star}},\beta_{\mathrm{tar}}}
+
\Gamma_{\mathrm{wall}}
\mathfrak{W}_{\boldsymbol{\xi}_{k^{\star}},\boldsymbol{\xi}_{\mathrm{tar}}}^{(\beta_{\mathrm{tar}})}
=
r_{\mathrm{cov}}
\big(
\mathfrak{d}_{\mathrm{tar}};
\mathcal{D}_{\mathrm{tr}}
\big).
\end{aligned}
\end{equation}
\par
The structured single-source theorem from Subsection III-B may now be applied with the source descriptor $(p_{\mathrm{src}},\beta_{\mathrm{src}},\boldsymbol{\xi}_{\mathrm{src}})$ replaced by $\mathfrak{d}_{k^{\star}}$. The resulting physical shift term is
\begin{equation}
4J_{\upsilon}L_{\upsilon}
\Gamma_{\mathrm{per}}
\overline{D}_{p_{k^{\star}},p_{\mathrm{tar}}}^{(\mathrm{cp})}
+
4J_{\upsilon}L_{\upsilon}
\Gamma_{\mathrm{view}}
\mathfrak{V}_{\beta_{k^{\star}},\beta_{\mathrm{tar}}}
+
4J_{\upsilon}L_{\upsilon}
\Gamma_{\mathrm{wall}}
\mathfrak{W}_{\boldsymbol{\xi}_{k^{\star}},\boldsymbol{\xi}_{\mathrm{tar}}}^{(\beta_{\mathrm{tar}})}
=
4J_{\upsilon}L_{\upsilon}
r_{\mathrm{cov}}
\big(
\mathfrak{d}_{\mathrm{tar}};
\mathcal{D}_{\mathrm{tr}}
\big),
\end{equation}
which proves the first claim after substitution.\par
For the uniform statement, the covering assumption guarantees that, for each $\mathfrak{d}_{\mathrm{tar}}\in\mathfrak{D}_{\mathrm{tar}}^{\mathrm{cov}}$,
\begin{equation}
r_{\mathrm{cov}}
\big(
\mathfrak{d}_{\mathrm{tar}};
\mathcal{D}_{\mathrm{tr}}
\big)
\leq
\Gamma_{\mathrm{per}}\varepsilon_{\mathrm{per}}
+
\Gamma_{\mathrm{view}}\varepsilon_{\mathrm{view}}
+
\Gamma_{\mathrm{wall}}\varepsilon_{\mathrm{wall}}.
\end{equation}
\par
Multiplying the last inequality by $4J_{\upsilon}L_{\upsilon}$ and inserting it into the first bound proves the uniform version.\par
\end{proof}
Hence multi-source training supplies two independent gains in the TWR HAR target risk bound: the larger total sample size $N_{\Sigma}$ contracts the statistical complexity term, whereas the smaller coverage radius $r_{\mathrm{cov}}(\mathfrak{d}_{\mathrm{tar}};\mathcal{D}_{\mathrm{tr}})$ contracts the structured physical shift term. Accordingly, under the coverage and approximation assumptions stated above, increasing the diversity of training persons, training views, and training wall conditions can tighten the upper bound for unseen target domains.\par
\begin{table}[!t]
\begin{center}
\caption{Radar and Scene Settings.\label{tab:exp_radar_scene}}
\vspace{-0.2cm}
\resizebox{0.38\textwidth}{!}{
\begin{tabular}{cc}
\hline\hline
\textbf{Parameter} & \textbf{Value} \\
\hline
Antenna Height to Ground & $1.5 \mathrm{~m}$ \\
Start Frequency & $2.5 \mathrm{~GHz}$ \\
Bandwidth & $1 \mathrm{~GHz}$ \\
Pulse Repetition Frequency & $125 \mathrm{~Hz}$ \\
Fast-Time Samples & $3940$/chirp \\
Human Motion Range from Radar & $1 \sim 6 \mathrm{~m}$ \\
Number of Activities & $12$ \\
\hline\hline
\end{tabular}
}
\end{center}
\vspace{-0.2cm}
\end{table}\par
\begin{table}[!t]
\begin{center}
\caption{Fixed MLP Recognizer Settings$^{*}$.\label{tab:exp_mlp_settings}}
\vspace{-0.2cm}
\resizebox{0.36\textwidth}{!}{
\begin{tabular}{cc}
\hline\hline
\textbf{Setting} & \textbf{Value} \\
\hline
Hidden Layer 1 Width & $256$ \\
Hidden Layer 2 Width & $128$ \\
Hidden Activation & ReLU \\
Output Width & $12$ \\
Output Activation & Softmax \\
Loss Function & Cross-Entropy \\
Optimizer & Adam \\
Batch Size & $32$ \\
Maximum Training Epochs & $80$ \\
Initial Learning Rate & $0.00147$ \\
Regularization Method & $L_2$ \\
\hline\hline
\end{tabular}
}
\end{center}
\vspace{-0.3cm}
\begin{center}
\parbox{0.36\textwidth}{\scriptsize $^{*}$ The same MLP is used for RT, DT, MD, RT+DT+MD, and PHY under all transfer settings.\par}
\end{center}
\vspace{-0.4cm}
\end{table}\par

\section{Simulated and Measured Experiments}
This section first gives the experimental settings. Then, the cross-person, cross-view, and cross-wall results under different feature representations are reported, and the effects of multi-source training, parameter-space coverage, and limited target-set expansion are further examined. Finally, brief discussions on the current theoretical scope are given.\par
\begin{figure*}[!t]
\centering
\includegraphics[width=\textwidth]{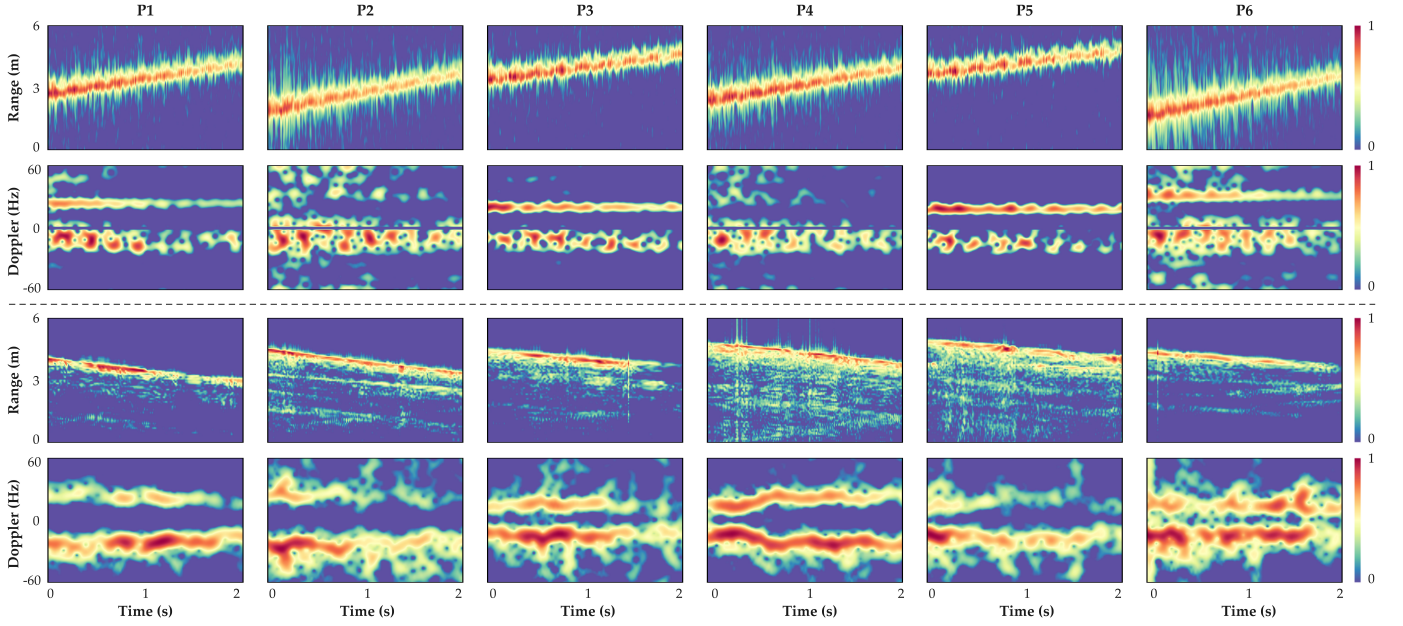}
\caption{Cross-person visualization for walking ($S8$), where columns correspond to $P1$ to $P6$, and rows $1$--$4$ show simulated range-time maps (RTMs), simulated Doppler-time maps (DTMs), measured RTMs, and measured DTMs, respectively.}
\label{fig:vis_cross_person}
\vspace{-0.0cm}
\end{figure*}\par
\begin{figure*}[!t]
\centering
\includegraphics[width=\textwidth]{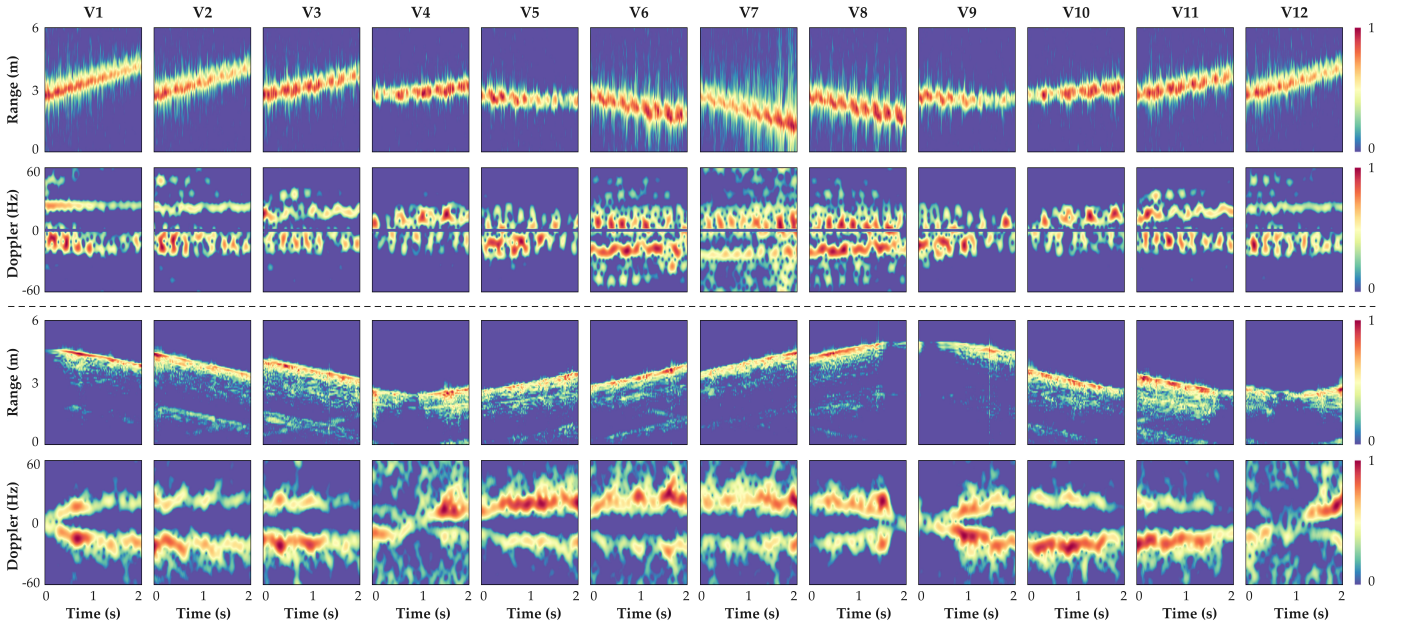}
\caption{Cross-view visualization for walking ($S8$), where columns correspond to $V1$ to $V12$, and rows $1$--$4$ show simulated RTMs, simulated DTMs, measured RTMs, and measured DTMs, respectively.}
\label{fig:vis_cross_view}
\vspace{-0.2cm}
\end{figure*}\par
\begin{figure}[!t]
\centering
\begin{minipage}{0.48\textwidth}
\centering
\includegraphics[width=\textwidth]{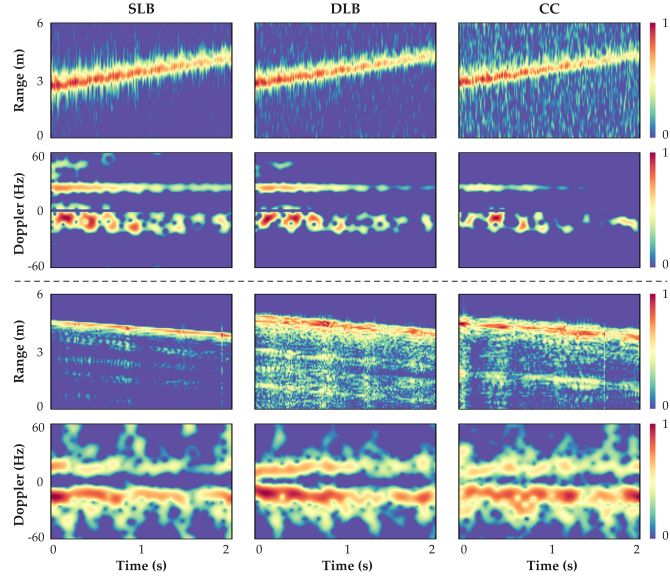}
\caption{Cross-wall visualization for walking ($S8$), where columns correspond to SLB, DLB, and CC, and rows $1$--$4$ show simulated RTMs, simulated DTMs, measured RTMs, and measured DTMs, respectively.}
\label{fig:vis_cross_wall}
\end{minipage}
\vspace{-0.0cm}
\end{figure}\par
\subsection{Experimental Settings}
Simulated and measured experiments are jointly used to examine the theoretical trends predicted by Section III. The simulated dataset is generated by RadHARSimulator V1 \cite{RadHARSimulatorV1}, and the measured dataset is collected by our ultrawideband TWR platform in typical urban indoor scenes. The common radar and scene settings are summarized in Table~\ref{tab:exp_radar_scene}. The activity-label set is fixed as $\mathcal{A}=\{S1,\ldots,S12\}$, where $S1$ denotes Empty, $S2$ denotes Punching, $S3$ denotes Kicking, $S4$ denotes Grabbing, $S5$ denotes Sitting Down, $S6$ denotes Standing Up, $S7$ denotes Rotating, $S8$ denotes Walking, $S9$ denotes Sitting to Walking, $S10$ denotes Walking to Sitting, $S11$ denotes Falling to Walking, and $S12$ denotes Walking to Falling. For cross-person evaluation, six persons $P1$--$P6$ are used, where $P1$ is used for training and validation and $P2$--$P6$ are used for testing. For cross-view evaluation, twelve views $V1$--$V12$ are used, where $V1$ is used for training and validation and $V2$--$V12$ are used for testing, with view $V_{L}$ corresponding to $30^{\circ}(L-1)$ for $1\leq L\leq 12$. For cross-wall evaluation, three wall conditions are used, where single-layer brick (SLB) is used for training and validation and double-layer brick (DLB) together with concrete (CC) are used for testing. In all source-domain experiments, each class contributes $300$ groups, giving $3600$ groups in total, and the source data are split by $8:2$ into training and validation. In all target-domain experiments, each testing configuration contributes $30$ groups per class, giving $360$ groups for each target domain.\par
Since classifier design is not the focus of this paper, a fixed multilayer perceptron (MLP) is used in all experiments, and its parameters are summarized in Table~\ref{tab:exp_mlp_settings}. The two hidden-layer widths are $256$ and $128$, and the same optimization setting is used for all transfer tasks.\par
\subsection{Generalization Experiments of Cross-Person/View/Wall and Feature Representation}
As shown in Figs.~\ref{fig:vis_cross_person}, \ref{fig:vis_cross_view}, and \ref{fig:vis_cross_wall}, walking ($S8$) is used as the visualization example, and different structured shifts leave different image distortions. Further radar-image visualization comparisons across different activities can be found in \cite{MicroDopplerRep2025,LightweightMacroMicro2023,MIMOFusion2026}. Under cross-person transfer, the main RTM trajectories and DTM bands are preserved, while echo intensity, stripe thickness, and local fragments are changed, so the person shift is largely reflected as a statistical mismatch. Under cross-view transfer, the RTM slope, DTM bandwidth, and positive/negative Doppler allocation are changed most strongly, with the clearest degradation around $V4$ to $V10$, which is consistent with a stronger geometric perturbation. Under cross-wall transfer, the main target trajectory is retained, but attenuation, background texture, and noise are progressively strengthened from SLB to DLB and CC, which is consistent with propagation-loss- and signal-to-noise-ratio-driven mismatch. The measured results are rougher than the simulated ones, but the same trends remain, which motivates feature-level shift compression.\par
\begin{figure*}[!t]
\centering
\includegraphics[width=\textwidth]{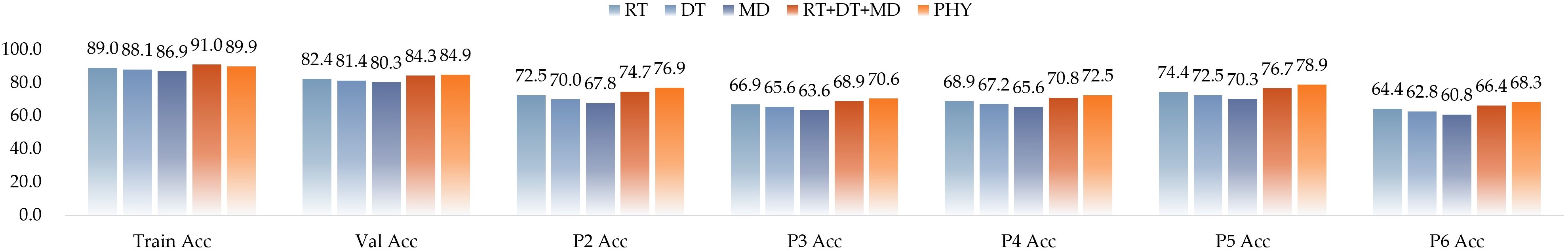}
\caption{Simulated recognition accuracies under the cross-person transfer task. Here, RT, DT, MD, PHY, and Acc denote RTM, DTM, micro-Doppler signature, physics-guided low-dimensional representation, and accuracy, respectively. All accuracy values are percentages.}
\label{fig:simulated_cross_person_results}
\vspace{-0.0cm}
\end{figure*}\par
\begin{figure*}[!t]
\centering
\includegraphics[width=\textwidth]{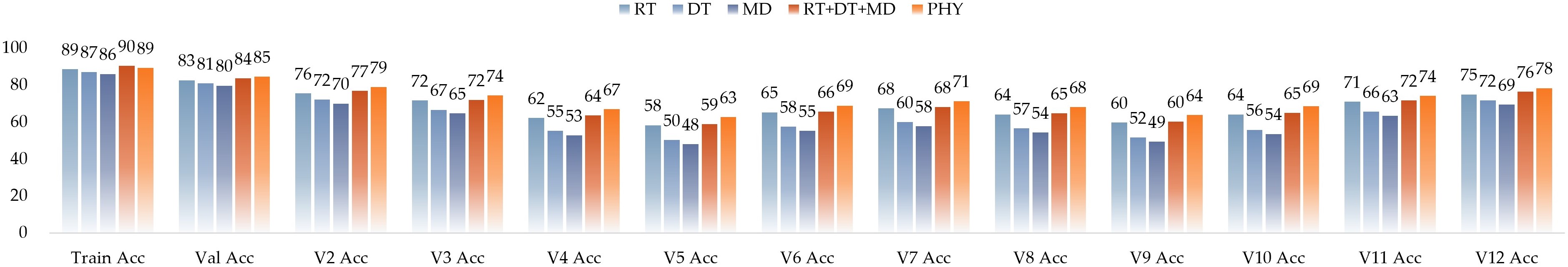}
\caption{Simulated recognition accuracies under the cross-view transfer task. Here, RT, DT, MD, PHY, and Acc denote RTM, DTM, micro-Doppler signature, physics-guided low-dimensional representation, and accuracy, respectively. All accuracy values are percentages.}
\label{fig:simulated_cross_view_results}
\vspace{-0.2cm}
\end{figure*}\par
\begin{figure*}[!t]
\centering
\includegraphics[width=\textwidth]{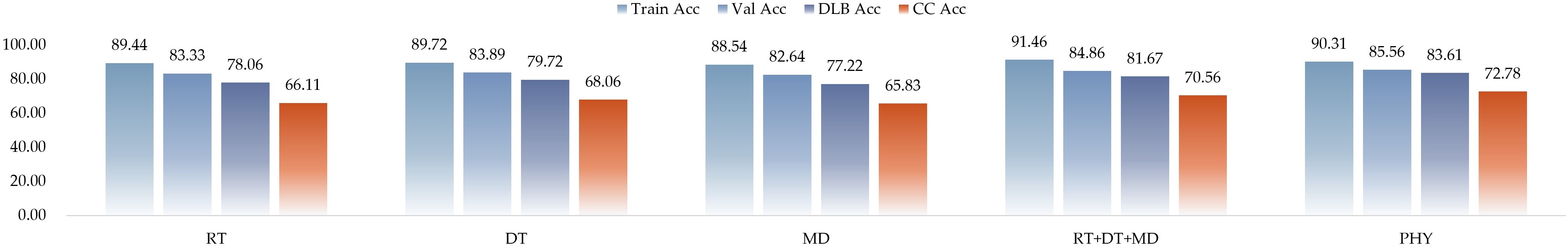}
\caption{Simulated recognition accuracies under the cross-wall transfer task. Here, RT, DT, MD, PHY, and Acc denote RTM, DTM, micro-Doppler signature, physics-guided low-dimensional representation, and accuracy, respectively. All accuracy values are percentages.}
\label{fig:simulated_cross_wall_results}
\vspace{-0.0cm}
\end{figure*}\par
\begin{figure*}[!t]
\centering
\includegraphics[width=\textwidth]{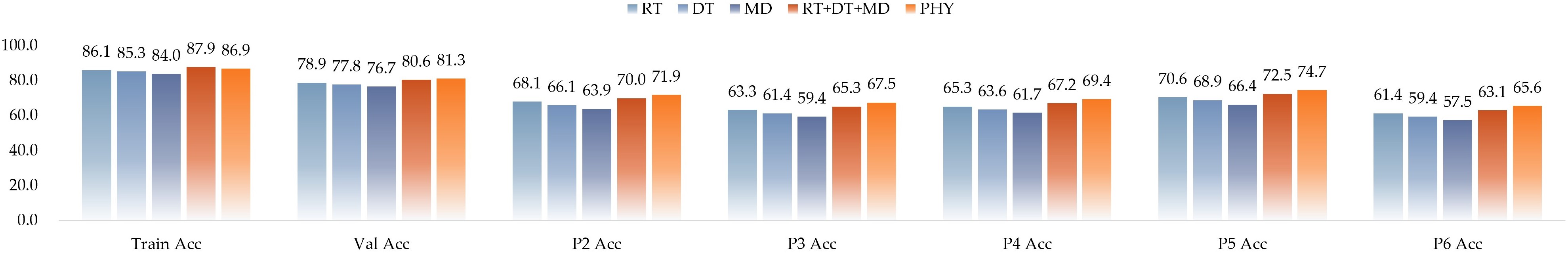}
\caption{Measured recognition accuracies under the cross-person transfer task. Here, RT, DT, MD, PHY, and Acc denote RTM, DTM, micro-Doppler signature, physics-guided low-dimensional representation, and accuracy, respectively. All accuracy values are percentages.}
\label{fig:measured_cross_person_results}
\vspace{-0.0cm}
\end{figure*}\par
\begin{figure*}[!t]
\centering
\includegraphics[width=\textwidth]{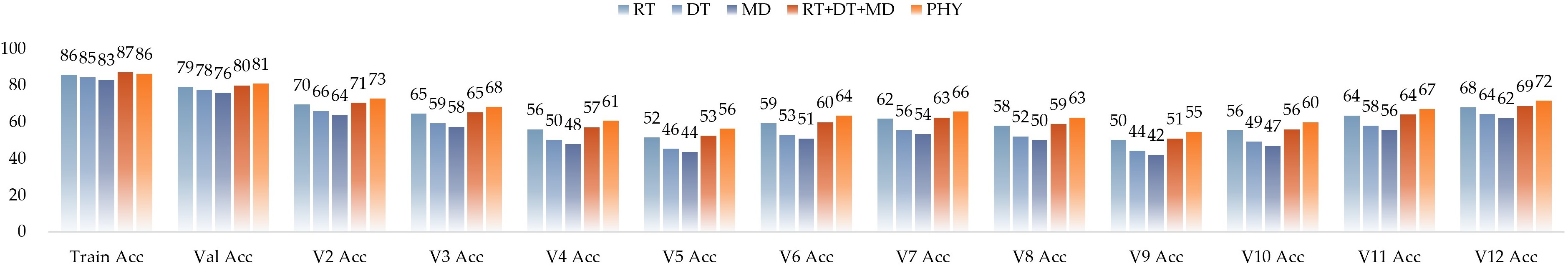}
\caption{Measured recognition accuracies under the cross-view transfer task. Here, RT, DT, MD, PHY, and Acc denote RTM, DTM, micro-Doppler signature, physics-guided low-dimensional representation, and accuracy, respectively. All accuracy values are percentages.}
\label{fig:measured_cross_view_results}
\vspace{-0.0cm}
\end{figure*}\par
\begin{figure*}[!t]
\centering
\includegraphics[width=\textwidth]{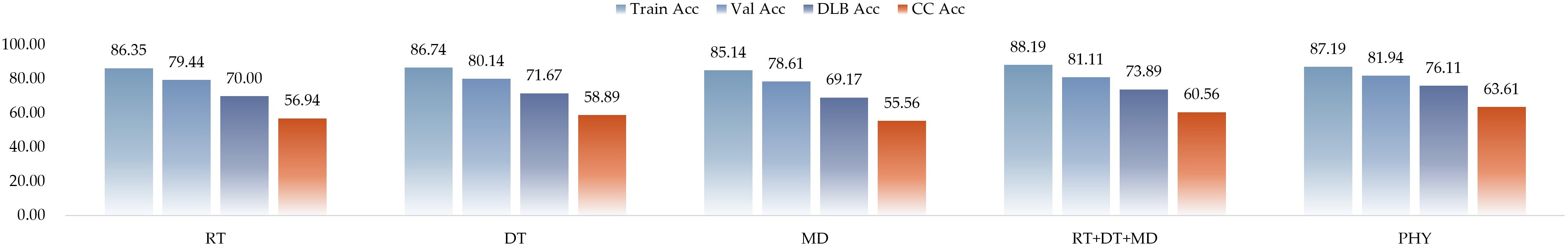}
\caption{Measured recognition accuracies under the cross-wall transfer task. Here, RT, DT, MD, PHY, and Acc denote RTM, DTM, micro-Doppler signature, physics-guided low-dimensional representation, and accuracy, respectively. All accuracy values are percentages.}
\label{fig:measured_cross_wall_results}
\vspace{-0.2cm}
\end{figure*}\par
As shown in Figs.~\ref{fig:simulated_cross_person_results}--\ref{fig:simulated_cross_wall_results} and Figs.~\ref{fig:measured_cross_person_results}--\ref{fig:measured_cross_wall_results}, the reported average accuracies consistently place PHY first and RT+DT+MD second across the simulated and measured cross-person, cross-view, and cross-wall tasks. In cross-person transfer, PHY reaches average accuracies of about $73.44\%$ in simulation and $69.83\%$ in measurement, both higher than RT at $69.44\%$ and $65.72\%$. In cross-view transfer, the lowest accuracies are concentrated around $V5$ and the nearby symmetric views, and the reported unseen-view averages remain the lowest among the three structured-shift settings, indicating that cross-view transfer is the most challenging case in these experiments. In cross-wall transfer, CC is clearly harder than DLB, while PHY attains the best simulated DLB/CC accuracies of $83.61\%/72.78\%$ and the best measured DLB/CC accuracies of $76.11\%/63.61\%$. These trends are consistent with the bound-tightening discussion in Subsection III-C, which can be interpreted as empirical support for the proposed theory.\par
The preceding cross-person and cross-wall experiments cover only five target persons and two target wall conditions. To improve completeness, additional simulated parameter-sweep experiments are further conducted over broader height-weight and wall-parameter ranges, and the corresponding target-domain test accuracies are summarized in Tables~\ref{tab:simh_cross_person_sweep} and \ref{tab:simh_cross_wall_sweep}. Across the listed settings, PHY achieves the highest test accuracy. Its cross-person test accuracy ranges from $66.67$ to $79.17$, and its cross-wall test accuracy ranges from $69.72$ to $87.50$. Along the main wall-parameter progression from WL-01 to WL-07, an overall accuracy degradation is also observed as the wall condition becomes harsher. These supplementary results further support the representation-induced trend observed above.\par
\begin{table*}[!t]
\centering
\caption{Simulated cross-person parameter-sweep results\textsuperscript{*}.\label{tab:simh_cross_person_sweep}}
\vspace{-0.15cm}
\sbox{\simhpersontabbox}{%
{\footnotesize
\setlength{\tabcolsep}{3.5pt}
\begin{tabular}{ccccccccc}
\hline\hline
\multirow{2}{*}{\textbf{Cfg.\textsuperscript{1}}} & \multirow{2}{*}{\textbf{Map.\textsuperscript{1}}} & \multirow{2}{*}{\textbf{Height\textsuperscript{2}}} & \multirow{2}{*}{\textbf{Weight\textsuperscript{2}}} & \multicolumn{5}{c}{\textbf{Test Accuracy\textsuperscript{2}}} \\
\cline{5-9}
 &  &  &  & \textbf{RT\textsuperscript{1}} & \textbf{DT\textsuperscript{1}} & \textbf{MD\textsuperscript{1}} & \textbf{RT+DT+MD\textsuperscript{1}} & \textbf{PHY\textsuperscript{1}} \\
\hline
HW-01 & P2 & 1.62 & 54 & 72.50 & 70.00 & 67.78 & 74.72 & 76.94 \\
HW-02 & P3 & 1.68 & 60 & 66.94 & 65.56 & 63.61 & 68.89 & 70.56 \\
HW-03 & P4 & 1.71 & 82 & 68.89 & 67.22 & 65.56 & 70.83 & 72.50 \\
HW-04 & P5 & 1.79 & 74 & 74.44 & 72.50 & 70.28 & 76.67 & 78.89 \\
HW-05 & P6 & 1.86 & 90 & 64.44 & 62.78 & 60.83 & 66.39 & 68.33 \\
HW-06 & -- & 1.76 & 70 & 75.83 & 73.89 & 71.67 & 77.78 & 79.17 \\
HW-07 & -- & 1.67 & 78 & 68.33 & 66.67 & 64.72 & 70.56 & 72.78 \\
HW-08 & -- & 1.60 & 51 & 65.56 & 63.61 & 61.39 & 67.50 & 69.44 \\
HW-09 & -- & 1.89 & 94 & 63.06 & 61.39 & 59.17 & 65.28 & 66.67 \\
HW-10 & -- & 1.73 & 66 & 71.67 & 69.17 & 66.94 & 73.33 & 75.83 \\
HW-11 & -- & 1.75 & 57 & 74.17 & 71.94 & 69.72 & 76.11 & 77.78 \\
HW-12 & -- & 1.82 & 84 & 67.50 & 65.83 & 63.61 & 70.00 & 72.22 \\
HW-13 & -- & 1.64 & 69 & 70.00 & 68.06 & 65.83 & 71.94 & 73.89 \\
HW-14 & -- & 1.92 & 88 & 65.00 & 63.33 & 61.11 & 66.94 & 68.61 \\
HW-15 & -- & 1.69 & 75 & 68.61 & 66.67 & 64.44 & 70.56 & 73.06 \\
\hline\hline
\end{tabular}
}}
\makebox[\textwidth][c]{\usebox{\simhpersontabbox}}\par
\vspace{0.05cm}
\makebox[\textwidth][c]{\parbox{\wd\simhpersontabbox}{\scriptsize
\textsuperscript{*} Simulated cross-person sweep results under different height-weight configurations. All accuracy values are percentages.\par
\textsuperscript{1} Cfg. denotes configuration, Map. denotes mapped domain, RT denotes RTM, DT denotes DTM, MD denotes micro-Doppler signature, RT+DT+MD denotes the concatenated RTM-DTM-micro-Doppler representation, and PHY denotes physics-guided low-dimensional representation.\par
\textsuperscript{2} Height and weight are expressed in meters and kilograms, respectively.\par}}\par
\vspace{-0.0cm}
\end{table*}\par
\begin{table*}[!t]
\centering
\caption{Simulated cross-wall parameter-sweep results\textsuperscript{*}.\label{tab:simh_cross_wall_sweep}}
\vspace{-0.15cm}
\sbox{\simhwalltabbox}{%
{\footnotesize
\setlength{\tabcolsep}{3pt}
\begin{tabular}{cccccccccc}
\hline\hline
\multirow{2}{*}{\textbf{Cfg.}} & \multirow{2}{*}{\textbf{Map.}} & \multirow{2}{*}{\textbf{Rel. Perm.\textsuperscript{1,2}}} & \multirow{2}{*}{\textbf{Loss Tan.\textsuperscript{1,2}}} & \multirow{2}{*}{\textbf{Thick.\textsuperscript{1,2}}} & \multicolumn{5}{c}{\textbf{Test Accuracy\textsuperscript{2}}} \\
\cline{6-10}
 &  &  &  &  & \textbf{RT} & \textbf{DT} & \textbf{MD} & \textbf{RT+DT+MD} & \textbf{PHY} \\
\hline
WL-01 & -- & 4.4 & 0.016 & 11 & 81.67 & 83.06 & 80.83 & 85.28 & 87.50 \\
WL-02 & DLB\textsuperscript{1} & 5.2 & 0.024 & 24 & 78.06 & 79.72 & 77.22 & 81.67 & 83.61 \\
WL-03 & -- & 5.8 & 0.029 & 18 & 76.67 & 78.33 & 75.56 & 80.28 & 82.50 \\
WL-04 & -- & 7.1 & 0.048 & 16 & 73.33 & 75.28 & 72.50 & 77.50 & 79.72 \\
WL-05 & -- & 6.6 & 0.041 & 22 & 71.94 & 73.89 & 71.11 & 76.11 & 78.33 \\
WL-06 & CC\textsuperscript{1} & 7.8 & 0.055 & 30 & 66.11 & 68.06 & 65.83 & 70.56 & 72.78 \\
WL-07 & -- & 8.5 & 0.064 & 26 & 63.89 & 65.28 & 63.06 & 67.78 & 69.72 \\
WL-08 & -- & 4.8 & 0.020 & 13 & 80.00 & 81.39 & 79.17 & 83.33 & 85.56 \\
WL-09 & -- & 5.5 & 0.026 & 20 & 77.50 & 78.89 & 76.67 & 80.83 & 83.06 \\
WL-10 & -- & 6.2 & 0.034 & 17 & 75.28 & 76.94 & 74.44 & 78.61 & 80.83 \\
WL-11 & -- & 6.9 & 0.043 & 25 & 72.78 & 74.44 & 71.94 & 76.67 & 78.89 \\
WL-12 & -- & 7.4 & 0.050 & 28 & 68.61 & 70.28 & 67.78 & 72.50 & 75.00 \\
WL-13 & -- & 8.1 & 0.059 & 32 & 65.00 & 66.67 & 64.17 & 68.89 & 71.39 \\
\hline\hline
\end{tabular}
}}
\makebox[\textwidth][c]{\usebox{\simhwalltabbox}}\par
\vspace{0.05cm}
\makebox[\textwidth][c]{\parbox{\wd\simhwalltabbox}{\scriptsize
\textsuperscript{*} Simulated cross-wall sweep results under different wall-parameter configurations. All accuracy values are percentages.\par
\textsuperscript{1} Rel. Perm. denotes relative permittivity, Loss Tan. denotes loss tangent, and Thick. denotes thickness.\par
\textsuperscript{2} Thickness is expressed in centimeters, and relative permittivity and loss tangent are dimensionless.\par}}\par
\vspace{-0.0cm}
\end{table*}\par
\subsection{Multi-Source Training and Parameter-Space Coverage Results}
In Tables~\ref{tab:simh_train_expansion_person}--\ref{tab:simh_train_expansion_wall} and Tables~\ref{tab:rw_train_expansion_person}--\ref{tab:rw_train_expansion_wall}, additional training-set expansion experiments are further reported. For each cross-person, cross-view, and cross-wall transfer setting, $15$ sample groups are taken from every target testing domain, added into the source training set, and then the recognizer is retrained, revalidated, and retested. The unseen-target averages are increased for all representations in both simulation and measurement. For PHY, the average gains are $+8.44/+8.66/+10.14$ percentage points in the simulated cross-person, cross-view, and cross-wall tasks, and $+7.39/+6.16/+7.64$ percentage points in the measured counterparts. Similar improvements are also observed for RT, DT, MD, and RT+DT+MD, which is consistent with sample insufficiency being one contributor to the generalization error. Nevertheless, the representation ordering is almost unchanged, PHY retains the best average performance across the three structured shifts, and cross-view transfer remains the hardest case. This pattern is consistent with the sample-size and coverage terms in Section III-D, which can also be interpreted as empirical support for the proposed theory.\par
\begin{table}[!t]
\centering
\caption{Simulated training-set expansion results under cross-person transfer\textsuperscript{*}.\label{tab:simh_train_expansion_person}}
\vspace{-0.1cm}
\sbox{\traintabbox}{%
{\footnotesize
\setlength{\tabcolsep}{3.5pt}
\begin{tabular}{cccccccc}
\hline\hline
\textbf{Feature} & \textbf{Train} & \textbf{Val} & \textbf{P2} & \textbf{P3} & \textbf{P4} & \textbf{P5} & \textbf{P6} \\
\hline
RT & 88.31 & 83.61 & 79.44 & 73.33 & 75.56 & 82.78 & 70.56 \\
DT & 87.57 & 82.64 & 77.22 & 71.67 & 73.89 & 80.56 & 68.33 \\
MD & 86.40 & 81.53 & 74.44 & 68.89 & 71.67 & 77.78 & 66.11 \\
RT+DT+MD & 89.63 & 84.58 & 82.22 & 76.67 & 78.89 & 85.56 & 72.78 \\
PHY & 89.02 & 85.42 & 85.00 & 79.44 & 81.67 & 87.78 & 75.56 \\
\hline\hline
\end{tabular}}}
\usebox{\traintabbox}\par
\vspace{0.05cm}
\parbox{\wd\traintabbox}{\scriptsize \textsuperscript{*} In this task, $15$ sample groups are added from every target testing domain into the source training set, and the recognizer is retrained, revalidated, and retested. RT denotes RTM, DT denotes DTM, MD denotes micro-Doppler signature, and PHY denotes the physics-guided low-dimensional representation. All accuracy values are percentages.\par}
\vspace{-0.4cm}
\end{table}\par
\begin{table*}[!t]
\centering
\caption{Simulated training-set expansion results under cross-view transfer\textsuperscript{*}.\label{tab:simh_train_expansion_view}}
\vspace{-0.1cm}
\sbox{\traintabbox}{%
{\footnotesize
\setlength{\tabcolsep}{2.5pt}
\begin{tabular}{cccccccccccccc}
\hline\hline
\textbf{Feature} & \textbf{Train} & \textbf{Val} & \textbf{V2} & \textbf{V3} & \textbf{V4} & \textbf{V5} & \textbf{V6} & \textbf{V7} & \textbf{V8} & \textbf{V9} & \textbf{V10} & \textbf{V11} & \textbf{V12} \\
\hline
RT & 87.53 & 83.33 & 80.56 & 75.56 & 65.56 & 61.67 & 70.00 & 73.33 & 68.89 & 62.78 & 68.33 & 76.11 & 80.00 \\
DT & 86.52 & 81.94 & 77.22 & 71.11 & 60.00 & 54.44 & 63.33 & 66.67 & 62.22 & 56.11 & 61.67 & 71.67 & 76.11 \\
MD & 85.80 & 80.83 & 73.89 & 68.89 & 56.11 & 51.11 & 59.44 & 62.78 & 58.33 & 53.33 & 57.78 & 68.33 & 72.78 \\
RT+DT+MD & 88.48 & 84.44 & 83.89 & 78.89 & 70.56 & 65.56 & 72.78 & 76.11 & 71.67 & 66.11 & 71.11 & 80.00 & 83.33 \\
PHY & 88.02 & 85.42 & 86.67 & 82.78 & 76.11 & 71.67 & 77.78 & 81.11 & 76.67 & 71.11 & 76.11 & 83.89 & 87.78 \\
\hline\hline
\end{tabular}}}
\makebox[\textwidth][c]{\usebox{\traintabbox}}\par
\vspace{0.05cm}
\makebox[\textwidth][c]{\parbox{\wd\traintabbox}{\scriptsize \textsuperscript{*} Settings are consistent with Table~\ref{tab:simh_train_expansion_person}. RT denotes RTM, DT denotes DTM, MD denotes micro-Doppler signature, and PHY denotes the physics-guided low-dimensional representation. All accuracy values are percentages.\par}}\par
\vspace{-0.0cm}
\end{table*}\par
\begin{table}[!t]
\centering
\caption{Simulated training-set expansion results under cross-wall transfer\textsuperscript{*}.\label{tab:simh_train_expansion_wall}}
\vspace{-0.1cm}
\sbox{\traintabbox}{%
{\footnotesize
\setlength{\tabcolsep}{4.5pt}
\begin{tabular}{ccccc}
\hline\hline
\textbf{Feature} & \textbf{Train} & \textbf{Val} & \textbf{DLB} & \textbf{CC} \\
\hline
RT & 88.27 & 84.86 & 85.00 & 73.33 \\
DT & 88.73 & 85.56 & 87.78 & 75.56 \\
MD & 87.59 & 84.31 & 83.33 & 72.22 \\
RT+DT+MD & 89.63 & 86.25 & 90.00 & 79.44 \\
PHY & 89.23 & 86.81 & 93.33 & 83.33 \\
\hline\hline
\end{tabular}}}
\usebox{\traintabbox}\par
\vspace{0.05cm}
\parbox{\wd\traintabbox}{\scriptsize \textsuperscript{*} Settings are consistent with Table~\ref{tab:simh_train_expansion_person}. RT denotes RTM, DT denotes DTM, MD denotes micro-Doppler signature, and PHY denotes the physics-guided low-dimensional representation. All accuracy values are percentages.\par}
\vspace{-0.0cm}
\end{table}\par
\begin{table}[!t]
\centering
\caption{Measured training-set expansion results under cross-person transfer\textsuperscript{*}.\label{tab:rw_train_expansion_person}}
\vspace{-0.1cm}
\sbox{\traintabbox}{%
{\footnotesize
\setlength{\tabcolsep}{3.5pt}
\begin{tabular}{cccccccc}
\hline\hline
\textbf{Feature} & \textbf{Train} & \textbf{Val} & \textbf{P2} & \textbf{P3} & \textbf{P4} & \textbf{P5} & \textbf{P6} \\
\hline
RT & 85.56 & 80.14 & 73.89 & 68.33 & 70.00 & 76.67 & 65.00 \\
DT & 84.81 & 79.17 & 71.11 & 65.56 & 67.78 & 73.89 & 62.78 \\
MD & 83.65 & 78.06 & 68.33 & 62.78 & 65.00 & 70.56 & 60.56 \\
RT+DT+MD & 86.85 & 81.39 & 77.22 & 71.67 & 73.89 & 80.00 & 68.33 \\
PHY & 86.11 & 82.22 & 80.56 & 75.00 & 77.22 & 83.33 & 71.11 \\
\hline\hline
\end{tabular}}}
\usebox{\traintabbox}\par
\vspace{0.05cm}
\parbox{\wd\traintabbox}{\scriptsize \textsuperscript{*} In this task, $15$ sample groups are added from every target testing domain into the source training set, and the recognizer is retrained, revalidated, and retested. RT denotes RTM, DT denotes DTM, MD denotes micro-Doppler signature, and PHY denotes the physics-guided low-dimensional representation. All accuracy values are percentages.\par}
\vspace{-0.0cm}
\end{table}\par
\begin{table*}[!t]
\centering
\caption{Measured training-set expansion results under cross-view transfer\textsuperscript{*}.\label{tab:rw_train_expansion_view}}
\vspace{-0.1cm}
\sbox{\traintabbox}{%
{\footnotesize
\setlength{\tabcolsep}{2.5pt}
\begin{tabular}{cccccccccccccc}
\hline\hline
\textbf{Feature} & \textbf{Train} & \textbf{Val} & \textbf{V2} & \textbf{V3} & \textbf{V4} & \textbf{V5} & \textbf{V6} & \textbf{V7} & \textbf{V8} & \textbf{V9} & \textbf{V10} & \textbf{V11} & \textbf{V12} \\
\hline
RT & 85.60 & 80.42 & 71.67 & 67.22 & 57.78 & 52.78 & 61.67 & 65.00 & 60.56 & 53.89 & 59.44 & 67.22 & 71.11 \\
DT & 84.67 & 79.17 & 67.78 & 61.67 & 52.22 & 47.22 & 55.56 & 58.33 & 54.44 & 48.33 & 53.33 & 61.67 & 66.67 \\
MD & 83.70 & 77.78 & 64.44 & 58.89 & 48.89 & 44.44 & 51.67 & 54.44 & 50.56 & 45.00 & 49.44 & 57.78 & 63.33 \\
RT+DT+MD & 86.58 & 81.11 & 75.00 & 70.00 & 61.11 & 56.11 & 64.44 & 67.78 & 63.33 & 56.67 & 61.67 & 70.56 & 74.44 \\
PHY & 86.07 & 81.94 & 77.78 & 73.33 & 66.11 & 60.56 & 68.89 & 72.22 & 67.78 & 61.11 & 66.67 & 73.89 & 77.78 \\
\hline\hline
\end{tabular}}}
\makebox[\textwidth][c]{\usebox{\traintabbox}}\par
\vspace{0.05cm}
\makebox[\textwidth][c]{\parbox{\wd\traintabbox}{\scriptsize \textsuperscript{*} Settings are consistent with Table~\ref{tab:rw_train_expansion_person}. RT denotes RTM, DT denotes DTM, MD denotes micro-Doppler signature, and PHY denotes the physics-guided low-dimensional representation. All accuracy values are percentages.\par}}\par
\vspace{-0.0cm}
\end{table*}\par
\begin{table}[!t]
\centering
\caption{Measured training-set expansion results under cross-wall transfer\textsuperscript{*}.\label{tab:rw_train_expansion_wall}}
\vspace{-0.1cm}
\sbox{\traintabbox}{%
{\footnotesize
\setlength{\tabcolsep}{4.5pt}
\begin{tabular}{ccccc}
\hline\hline
\textbf{Feature} & \textbf{Train} & \textbf{Val} & \textbf{DLB} & \textbf{CC} \\
\hline
RT & 84.81 & 80.69 & 73.33 & 61.67 \\
DT & 85.22 & 81.25 & 75.56 & 65.00 \\
MD & 83.95 & 79.72 & 71.11 & 60.00 \\
RT+DT+MD & 86.11 & 81.94 & 78.89 & 68.89 \\
PHY & 85.65 & 82.64 & 82.78 & 73.33 \\
\hline\hline
\end{tabular}}}
\usebox{\traintabbox}\par
\vspace{0.05cm}
\parbox{\wd\traintabbox}{\scriptsize \textsuperscript{*} Settings are consistent with Table~\ref{tab:rw_train_expansion_person}. RT denotes RTM, DT denotes DTM, MD denotes micro-Doppler signature, and PHY denotes the physics-guided low-dimensional representation. All accuracy values are percentages.\par}
\vspace{-0.4cm}
\end{table}\par
\subsection{Discussions}
Although the theoretical analysis in the paper is proved in a closed-loop manner, we have nevertheless strengthened the empirical experiments, and the results remain consistent with the theoretical proof. In addition, regarding the theory proposed in this paper, we believe that the following two points merit further investigation:\par
\begin{enumerate}
\renewcommand{\labelenumi}{(\arabic{enumi})}
\setlength{\itemsep}{0em}
\setlength{\parsep}{0em}
\item The present generalization bound is established only for the single effective observation channel adopted in this paper. Although the simulated and measured experiments have supported the plausibility and practical relevance of this formulation, the MIMO case contains additional channel coupling, spatial diversity, and multistatic fusion effects, so the corresponding bound extension is left for future study.
\item This paper focuses on an upper bound that can directly support target risk analysis. By comparison, a lower bound on the generalization error of TWR HAR may offer limited direct guidance for engineering design. Nevertheless, its mathematical proof is still worthwhile, because the intrinsic task difficulty and the tightness of the current upper-bound analysis can then be better characterized.
\end{enumerate}\par
\section{Conclusion}
In this paper, a generalization-analysis framework for through-the-wall radar human activity recognition has been established. Indoor human kinematics, TWR echo generation, feature representation, bounded-weight recognition, and source-to-target risk definitions have been unified into one statistical learning model, and a target-domain generalization bound has been derived. The contributions of cross-person, cross-view, and cross-wall shifts have then been explicitly decomposed, and the tightening effects of physical low-dimensional representations, multi-source training, and parameter-space coverage have been analyzed. Simulated and measured experiments have provided empirical support for the resulting theoretical analysis and illustrated its application value.\par

\end{document}